%% file: main.tex
\documentclass[11pt]{article}

\usepackage{amsmath}
\usepackage{amsthm}
\usepackage{amssymb}
\usepackage{mathrsfs}
\usepackage{mathtools}
\usepackage{complexity}
\usepackage{hyperref}
\usepackage[capitalize]{cleveref}
\usepackage{tikz}
\usepackage{tkz-berge}
\usetikzlibrary{shapes.geometric, arrows}
\usepackage{float}

\usepackage{xspace}

\usepackage{caption}
\usepackage{subcaption}

\usepackage{tabularx}
\usepackage{listofitems}
\usepackage{array}
\newcolumntype{Y}{>{\centering\arraybackslash}X}

\usepackage{thmtools, thm-restate}
\usepackage[margin=1in]{geometry}

\usepackage{anyfontsize}

\newtheorem{theorem}{Theorem}[section]
\newtheorem{lemma}[theorem]{Lemma}
\newtheorem{corollary}[theorem]{Corollary}
\newtheorem{definition}[theorem]{Definition}
\newtheorem{proposition}[theorem]{Proposition}
\newtheorem{claim}{Claim}

\DeclareMathOperator{\hol}{Holant}
\DeclareMathOperator{\hols}{Holant^*}
\DeclareMathOperator{\holbs}{Holant^*_2}
\DeclareMathOperator{\holts}{Holant^*_3}
\DeclareMathOperator{\supp}{supp}
\DeclareMathOperator{\rank}{rank}
\DeclareMathOperator{\spn}{span}

\DeclareMathOperator{\parf}{\mathsf{PARITY}}

\DeclareMathOperator{\typei}{I}
\DeclareMathOperator{\typeii}{II}
\DeclareMathOperator{\proj}{proj}

\newcommand{\db}{\mathsf{B}}
\newcommand{\dg}{\mathsf{G}}
\newcommand{\dr}{\mathsf{R}}

\newcommand{\geneq}{\textsf{GenEQ}\xspace}
\newcommand{\genperm}{\textsf{GenPerm}\xspace}

\newcommand{\swhelper}[1]{$\mathsf{Swap}_{#1}$\xspace}
\newcommand{\swbg}{\swhelper{\db \dg; \dr}}
\newcommand{\swbr}{\swhelper{\db \dr; \dg}}
\newcommand{\swgr}{\swhelper{\dg \dr; \db}}

\newcommand{\octgroup}{O_h}

\newcommand{\sph}{\#\P-hard\xspace}

\newcommand{\teh}{^{\otimes 3}}
\newcommand{\tew}{^{\otimes 2}}
\newcommand{\transpose}{^\intercal}

\newcommand{\domres}[1]{
  ^{*\to\{#1\}}
}

\newcommand{\strspt}{\textsf{EBD}\xspace}

\renewcommand{\phi}{\varphi}

\tikzset{
mycirc/.style={circle,fill=black, minimum size=0.05cm}
}
\tikzset{
mysq/.style={rectangle,fill=black, minimum size=0.05cm}
}

\tikzstyle{block} = [rectangle, rounded corners, minimum width=1cm, minimum height=1cm,text centered, draw=black]
\tikzstyle{arrow} = [thick,->,>=stealth]

\newcommand{\bdgr}{\dg \dr | \db}

\newcommand{\tractbinary}{$\mathscr{A}$\xspace}
\newcommand{\tractE}{$\mathscr{B}$\xspace}
\newcommand{\tractBG}{$\mathscr{C}$\xspace}
\newcommand{\tractBGR}{$\mathscr{D}$\xspace}
\newcommand{\tractBGGRBR}{$\mathscr{E}$\xspace}

\newcommand{\ternarytractgeneq}{$\mathfrak{A}$\xspace}
\newcommand{\ternarytractz}{$\mathfrak{B}$\xspace}

\title{
Holant* Dichotomy on Domain Size 3: A Geometric Perspective
}

\author{Jin-Yi Cai\thanks{University of Wisconsin-Madison. {\tt jyc@cs.wisc.edu}}
\and Jin Soo Ihm\thanks{University of Wisconsin-Madison.
{\tt ihm2@wisc.edu}}}

\begin{document}

\date{}
\maketitle
\begin{abstract}
Holant problems are a general framework to study the computational complexity of counting problems. It is a more expressive framework than
counting constraint satisfaction problems (CSP) which are in turn
more expressive than counting graph homomorphisms (GH).
In this paper, we prove the first complexity dichotomy of $\holts(\mathcal{F})$ where $\mathcal{F}$ is an arbitrary set of symmetric, real valued constraint functions on domain size $3$.
We give an explicit tractability criterion
and prove that, if  $\mathcal{F}$ satisfies  this  
 criterion
 then $\holts(\mathcal{F})$ is   polynomial time
computable, and otherwise it is  \sph, with no intermediate cases. 
We show that the geometry of the tensor decomposition of the constraint functions 
plays a central role in the formulation as well as the
structural internal logic of the dichotomy.
\end{abstract}

\input{sections/introduction.tex}

\input{sections/notations.tex}

\input{sections/theorem.tex}
\input{sections/tractability.tex}
\input{sections/hardness.tex}

\input{sections/one_ternary_one_binary.tex}

\input{sections/two_ternaries.tex}

\input{sections/higher_arity.tex}

\input{sections/set_of_signatures.tex}

\bibliography{refs}

\end{document}

%% file: sections/introduction.tex
\section{Introduction}
Holant problems were introduced in \cite{cai_computational_2011} as a broad framework to study the computational complexity of counting problems.
Counting CSP  is a special case of Holant problems~\cite{Creignou1996ComplexityOG, bulatov-dalmau-csp, 10.1145/2528400, doi:10.1137/070690201, 10.1145/1536414.1536511, DBLP:journals/siamcomp/DyerR13, 5959820, Cai2011ComplexityOC}.
In turn, counting CSP includes counting graph homomorphisms (GH),
introduced by Lov\'asz~\cite{Lovsz1967OperationsWS, Hell2004GraphsAH},
which is a special case with a single binary constraint function.
Typical Holant problems include counting all matchings, counting
perfect matchings \#PM (including all weighted versions), counting cycle
covers, counting edge colorings, and many other natural problems.
It is strictly more expressive than GH; for example, it is known that
\#PM cannot be expressed in the framework of GH~\cite{Freedman2004ReflectionPR,DBLP:journals/cpc/CaiG22}.

The complexity classification program of counting problems is to classify as broad a class of  problems as possible
according to their inherent computational complexity within these
frameworks.
Let $\mathcal{F}$ be a set of (real or complex valued) constraint functions defined on some domain set $D$.
It defines a Holant problem $\hol(\mathcal{F})$ as follows.
An input consists of a graph $G = (V, E)$, where each $v \in V$ has an associated $\mathbf{F} \in \mathcal{F}$, with incident edges to $v$ labeled as input variables of $\mathbf{F}$.
The output is the sum of products of evaluations of the constraint functions over all assignments over $D$ for the variables.
The goal of the complexity classification of Holant problems is to classify the complexity of $\hol(\mathcal{F})$.
A complexity dichotomy theorem for counting problems classifies every problem
in a broad class of problems $\mathcal{F}$ 
to be either polynomial time solvable or \sph.

There has been tremendous progress in the classification of 
counting GH and counting CSP~\cite{Creignou1996ComplexityOG, doi:10.1137/070690201, 10.1145/1536414.1536511,
10.1145/2528400, DBLP:journals/siamcomp/DyerR13, 5959820, Cai2011ComplexityOC, 
dyer-graph-hom,
bulatov_complexity_2005, goldberg-partition,  hutchison_graph_2010}.
Much progress was also made in the classification of 
Holant problems, particularly on the Boolean domain ($|D|=2$), 
i.e., when variables take 0-1 values (but constraint functions take
arbitrary values, such as partition functions from
statistical physics). This includes the
dichotomy for all complex-valued symmetric constraint functions~\cite{complex-symmetric-holant-dichotomy}
and for all real-valued not necessarily symmetric constraint functions~\cite{real-holant-dichotomy}.
On the other hand, obtaining higher domain Holant dichotomy has been far more challenging.
There is a huge increase in difficulty in proving dichotomy theorems for domain size $> 2$, as already seen in decision CSP of domain size $3$, a major achievement by Bulatov \cite{bulatov-csp-domain-3}.
Toward proving 
these dichotomies one often first considers restricted classes of
Holant problems assuming some particular set of constraint functions
are present. Two 
sets stand out:
(1) the set of equality functions
 ${\cal EQ}$ of all arities (this is the class of all counting CSP problems)
 and (2) the set of all unary functions ${\cal U}$, i.e., functions of arity one.
 Indeed, \#CSP($\mathcal{F})
=\hol(\mathcal{F}
\cup {\cal EQ})$; i.e., counting CSP  are the special case of Holant problems
with  ${\cal EQ}$ assumed to be present.
In this paper we study (2):
$\holts(\mathcal{F}) := \hol_3(\mathcal{F} \cup \mathcal{U})$,
for an arbitrary set $\mathcal{F}$ of symmetric real-valued  constraint functions
on domain size 3. 


Previously there were only two significant Holant dichotomies on higher domains.
One is for
a single ternary constraint function that has a strong symmetry property
called domain permutation invariance~\cite{cai_complexity_2016}. That work also 
solves a decades-old open problem of
the complexity of counting edge colorings. The other is a
dichotomy for $\holts(f)$ where $f$ is a single symmetric complex-valued
ternary constraint function on 
domain size 3~\cite{cai_dichotomy_2013}. Extending this
dichotomy to an arbitrary constraint function, or more ambitiously,
to a set of constraint functions has been a goal for more than 
10 years without much progress.

In this paper, we extend the result in
\cite{cai_dichotomy_2013} to an arbitrary
set of real-valued symmetric constraint functions.
In~\cite{liu_restricted_nodate}
an interesting observation was made that an exceptional form of 
complex-valued tractable constraint functions 
does not occur when the function is 
real-valued.
By restricting ourselves to a set  $\mathcal{F}$ of 
real-valued constraint functions, we can bypass a lot  of difficulty
associated with this exceptional form.
Another major source of intricacy is related to the interaction of binary constraint functions
with other constraint functions in  $\mathcal{F}$.
We introduce a new geometric perspective that provides 
a unifying
principle in the formulation as well as
a structural internal logic of what leads to tractability and what  leads  
to \#\P-hardness. After discovering some new tractable classes of functions
aided by the geometric perspective, we 
are able to prove a $\holts(\mathcal{F})$ dichotomy.
This dichotomy is dictated by the geometry of the tensor decomposition of constraint functions.

Suppose $\mathbf{G}$ is a binary constraint function and $\mathbf{F}$ is a tenary constraint function, with $\mathbf{F} = \mathbf{u}\teh + \mathbf{v}\teh$ its tensor decomposition.
One of the simplest constructions possible with $\mathbf{G}$ and $\mathbf{F}$ is to connect $\mathbf{G}$ at the three edges of $\mathbf{F}$;
the resulting constraint function is $\mathbf{G} \teh \mathbf{F}$ which has tensor decomposition $(\mathbf{G} \mathbf{u})\teh + (\mathbf{G} \mathbf{v})\teh$.
We see that this gadget construction plays nicely with the tensor decomposition.
Generalizing this idea, suppose $\mathcal{B}$ is a set of binary constraint functions and $\mathcal{T}$ is a set of ternary constraint functions.
Let $\langle \mathcal{B} \rangle$ be the monoid generated by $\mathcal{B}$.
We may consider the orbit  $\mathcal{O}$ of $\mathcal{T}$ under the monoid action of $\langle \mathcal{B} \rangle$,
such that $\mathbf{G} \in \langle \mathcal{B} \rangle$ acts on $\mathbf{F} \in \mathcal{T}$ by $\mathbf{G} : \mathbf{F} \mapsto \mathbf{G}\teh \mathbf{F}$.
Although the constraint functions in $\mathcal{O}$ are the results of a very simple gadget construction, we show that $\mathcal{O}$ contains sufficient information about the interaction of  binary constraint functions and  other constraint functions, and the simplicity allows us to analyze it by considering the geometry of the vectors of the tensor decomposition of the constraint functions in $\mathcal{O}$.

Compared to 
the Boolean domain dichotomy theorem, stated in explicit
  recurrences
on the values of the signatures (see Theorem 2.12 in \cite{cai_complexity_2017}) 
the dichotomy theorem (\cref{thm:dichotomy-set-of-domain-3})
we wish to prove has a more
non-explicit
form, which is also more conceptual. This is informed by the
geometric perspective, but it also causes some difficulty in its proof,
when we try to extend
to a set of constraint functions of arbitrary arities.
We introduce a new  technique   to overcome this difficulty.
 First (and this is quite a surprise), it turns out that a dichotomy of two constraint functions of arity $3$ is easier to state and prove than the dichotomy of one binary and one ternary constraint functions. Also they can be proven independently of each other.
This is a departure from all previous proofs of dichotomy theorems in this area.
Second, using the unary constraint functions available in $\hols$, any symmetric constraint function $\mathbf{F}$ of arity $4$  defines a linear transformation from $\mathbb{R}^3$ to the space of symmetric constraint functions of arity $3$, which corresponds to the ternary constraint functions constructible by connecting a unary function to $\mathbf{F}$.
In particular, the image of this map, $\mathscr{F}$, is a linear subspace. 
In particular, the image $\mathscr{F}$ of this map is a linear subspace.
Considering the space $\mathscr{F}$ instead of specific 
sub-functions allows us to bypass the difficulty from the non-explicit form of the  dichotomy
statement, which is  in terms of tensor decompositions up to 
an orthogonal transformation.
We show that a dichotomy of two ternary constraint functions and the fact that $\mathscr{F}$ is closed under linear combinations imply that $\mathscr{F}$ must be of a very special form for $\holts(\mathscr{F})$ to be tractable, which in turn implies that $\mathbf{F}$ must possess a certain regularity.

While the tractability criterion in
 \cref{thm:dichotomy-set-of-domain-3} is stated in a conceptual and succinct way,
 the tractable cases are actually quite rich and varied. We present here specific examples of new tractable cases.
Denote the domain by  $D = \{\db, \dg, \dr\}$.
We use the notation in \cref{fig:ternary-signature-notation} to denote a symmetric ternary constraint function on domain $D$.
\begin{figure}
\centering
  \begin{tikzpicture}[scale=0.3, every node/.style={scale=0.7}]
  \node at (0, 0) {$f_{\db \db \db}$};
  \node at (-2, -1) {$f_{\db \db \dg}$};
  \node at (2, -1) {$f_{\db \db \dr}$};
  \node at(-4, -2) {$f_{\db \dg \dg}$};
  \node at (0, -2) {$f_{\db \dg \dr}$};
  \node at (4, -2) {$f_{\db \dr \dr}$};
  \node at (-6, -3) {$f_{\dg \dg \dg}$};
  \node at (-2, -3) {$f_{\dg \dg \dr}$};
  \node at (2, -3) {$f_{\dg \dr \dr}$};
  \node at (6, -3) {$f_{\dr \dr \dr}$};
\end{tikzpicture}
\caption{Notation for expressing a symmetric ternary domain $3$ constraint functions. 
This notation can be extended for higher arity signatures by using a larger triangle.
}
\label{fig:ternary-signature-notation}
\end{figure}
Consider the four constraint functions $\mathbf{F}_1, \mathbf{G}_1, \mathbf{H}_1, \mathbf{B}_1$ in \cref{fig:f-g-h-tractable}.
\begin{figure}
  \centering
  \begin{subfigure}[b]{0.22\textwidth}
    \centering
  \begin{tikzpicture}[scale=0.3, every node/.style={scale=0.7}]
  \node at (0, 0) {$2$};
  \node at (-1, -1) {$2$};
  \node at (1, -1) {$-1$};
  \node at(-2, -2) {$2$};
  \node at (0, -2) {$-1$};
  \node at (2, -2) {$5$};
  \node at (-3, -3) {$2$};
  \node at (-1, -3) {$-1$};
  \node at (1, -3) {$5$};
  \node at (3, -3) {$-7$};
\end{tikzpicture}
\caption{$\mathbf{F}_1$}
\end{subfigure}
\begin{subfigure}[b]{0.22\textwidth}
  \centering
  \begin{tikzpicture}[scale=0.3, every node/.style={scale=0.7}]
  \node at (0, 0) {$-7$};
  \node at (-1, -1) {$5$};
  \node at (1, -1) {$-1$};
  \node at(-2, -2) {$5$};
  \node at (0, -2) {$5$};
  \node at (2, -2) {$2$};
  \node at (-3, -3) {$-7$};
  \node at (-1, -3) {$-1$};
  \node at (1, -3) {$2$};
  \node at (3, -3) {$2$};
\end{tikzpicture}
\caption{$\mathbf{G}_1$}
\end{subfigure}
\begin{subfigure}[b]{0.22\textwidth}
  \centering
  \begin{tikzpicture}[scale=0.3, every node/.style={scale=0.7}]
  \node at (0, 0) {$-2$};
  \node at (-1, -1) {$1$};
  \node at (1, -1) {$1$};
  \node at(-2, -2) {$1$};
  \node at (0, -2) {$-2$};
  \node at (2, -2) {$1$};
  \node at (-3, -3) {$-2$};
  \node at (-1, -3) {$1$};
  \node at (1, -3) {$1$};
  \node at (3, -3) {$-2$};
\end{tikzpicture}
\caption{$\mathbf{H}_1$}
\end{subfigure} 
\begin{subfigure}[b]{0.22\textwidth}
  \centering
  \begin{tikzpicture}[scale=0.3, every node/.style={scale=0.7}]
  \node at (0, 0) {$3+2\sqrt{2}$};
  \node at (-1.5, -1) {$-3+2\sqrt{2}$};
  \node at (2, -1) {$-\sqrt{2}$};
  \node at(-3, -2) {$3+2\sqrt{2}$};
  \node at (0.5, -2) {$-\sqrt{2}$};
  \node at (3.5, -2) {$-4\sqrt{2}$};
\end{tikzpicture}
\caption{$\mathbf{B}_1$}
\end{subfigure}
  \caption{Ternary constraint functions $\mathbf{F}_1$, $\mathbf{G}_1$, $\mathbf{H}_1$, and a binary constraint function $\mathbf{B}_1$.} \label{fig:f-g-h-tractable}
\end{figure}
It is not obvious that $\holts(\mathbf{F}_1, \mathbf{G}_1, \mathbf{H}_1, \mathbf{B}_1)$ is polynomial-time computable.

We apply the orthogonal transform $T = \frac{1}{\sqrt{6}} \left[ \begin{smallmatrix}
    \sqrt{2} & \sqrt{2} & \sqrt{2} \\
    1 & 1 & -2\\
    \sqrt{3} & - \sqrt{3} & 0
\end{smallmatrix} \right]$
which transforms $\mathbf{F}_1, \mathbf{G}_1$ and $\mathbf{H}_1$ to be supported
in $\{\db, \dg\}^*, \{\db, \dr\}^*, \{\dg, \dr\}^*$ respectively,
Their tensor decompositions  have a revealing structure. 
Ignoring the scalar constants, we have\footnote{Complex numbers do appear,  even though the signatures are all real valued. This is similar to  eigenvalues.}
\begin{align*}
    \mathbf{F}_1' &= T \teh \mathbf{F}_1 = 3 \sqrt{3} (1, 0, 0)\teh + 6 \sqrt{6} (0, 1, 0)\teh = 3 \sqrt{3}  \, \mathbf{e}_1\teh + 6 \sqrt{6} \,\mathbf{e}_2\teh \\
    \mathbf{G}_1' &= T \teh \mathbf{G}_1 = (1, 0, i)\teh + (1, 0, -i)\teh = (\mathbf{e}_1 + i \mathbf{e}_3)\teh + (\mathbf{e}_1 - i \mathbf{e}_3)\teh \\
    \mathbf{H}_1' &= T \teh \mathbf{H}_1 = (0, 1, i)\teh + (0, 1, -i)\teh = (\mathbf{e}_2 + i \mathbf{e}_3)\teh + (\mathbf{e}_2 - i \mathbf{e}_3)\teh
\end{align*}
The vectors in tensor decompositions show that geometrically, $\mathbf{F}_1'$, $\mathbf{G}_1'$, and $\mathbf{H}_1'$ are associated with 
three coordinate planes.
The function $\mathbf{B}_1' = T \tew \mathbf{B}_1$ written in matrix form
is $\left[ \begin{smallmatrix}
    0 & 1 & 0\\
    1 & 0 & 0 \\
    0 & 0 & 1
\end{smallmatrix} \right]$, where the $(i,j)$ entry
is the function value $\mathbf{B}_1'(i,j)$, for $i, j$ in the new domain set.
Applying  \cref{thm:dichotomy-set-of-domain-3}
we can conclude  that $\{\mathbf{F}_1, \mathbf{G}_1, \mathbf{H}_1, \mathbf{B}_1\}$ is 
in tractable class \tractBGGRBR. 

For the second example we apply the orthogonal transform
$T = \frac{1}{\sqrt{6}}
\left[ \begin{smallmatrix}
\sqrt{2} & \sqrt{2} & - \sqrt{2}\\
1 & 1 & 2\\
-\sqrt{3} & \sqrt{3} & 0
\end{smallmatrix} \right]$ to the constraint functions in \cref{fig:bg-r-tractable}.
\begin{figure}
  \centering
  \begin{subfigure}[b]{0.22\textwidth}
    \centering
  \begin{tikzpicture}[scale=0.3, every node/.style={scale=0.7}]
  \node at (0, 0) {$-3$};
  \node at (-1, -1) {$1$};
  \node at (1, -1) {$-5$};
  \node at(-2, -2) {$-3$};
  \node at (0, -2) {$-5$};
  \node at (2, -2) {$2$};
  \node at (-3, -3) {$1$};
  \node at (-1, -3) {$-5$};
  \node at (1, -3) {$2$};
  \node at (3, -3) {$10$};
\end{tikzpicture}
\caption{$\mathbf{F}_2$}
\end{subfigure}
\begin{subfigure}[b]{0.22\textwidth}
  \centering
  \begin{tikzpicture}[scale=0.3, every node/.style={scale=0.7}]
  \node at (0, 0) {$5$};
  \node at (-1, -1) {$11$};
  \node at (1, -1) {$4$};
  \node at(-2, -2) {$5$};
  \node at (0, -2) {$4$};
  \node at (2, -2) {$2$};
  \node at (-3, -3) {$11$};
  \node at (-1, -3) {$4$};
  \node at (1, -3) {$2$};
  \node at (3, -3) {$1$};
\end{tikzpicture}
\caption{$\mathbf{G}_2$}
\end{subfigure}
\begin{subfigure}[b]{0.3\textwidth}
  \centering
  \begin{tikzpicture}[scale=0.3, every node/.style={scale=0.7}]
  \node at (0, 0) {$4 + 2 \sqrt{2}$};
  \node at (-2, -1) {$-2 + 2 \sqrt{2}$};
  \node at (2, -1) {$-4 + \sqrt{2}$};
  \node at(-4, -2) {$4 + 2 \sqrt{2}$};
  \node at (0, -2) {$-4 + \sqrt{2}$};
  \node at (4, -2) {$-2 -4\sqrt{2}$};
\end{tikzpicture}
\caption{$\mathbf{H}_2$}
\end{subfigure} 
\begin{subfigure}[b]{0.23\textwidth}
  \centering
  \begin{tikzpicture}[scale=0.3, every node/.style={scale=0.7}]
  \node at (0, 0) {$2-2\sqrt{2}$};
  \node at (-1.5, -1) {$0$};
  \node at (1.5, -1) {$2 + \sqrt{2}$};
  \node at(-4, -2) {$-2 + 2\sqrt{2}$};
  \node at (0, -2) {$-2-\sqrt{2}$};
  \node at (4, -2) {$0$};
\end{tikzpicture}
\caption{$\mathbf{B}_2$}
\end{subfigure}
  \caption{Ternary constraint functions $\mathbf{F}_2$, $\mathbf{G}_2$ and bianry constraint functions $\mathbf{H}_2, \mathbf{B}_2$.} \label{fig:bg-r-tractable}
\end{figure}
%
{\small
\begin{align*}
    T \teh \mathbf{F}_2 = 3 \sqrt{3} ((1, i, 0)\teh + (1, -i, 0)\teh) + 4 \sqrt{2} \mathbf{e}_3\teh , \quad
    T \teh \mathbf{G}_2 = (\sqrt{3}, \sqrt{6}, 0)\teh + 6 \sqrt{2} \mathbf{e}_3\teh 
\end{align*}
}
and $T\teh \mathbf{H}_2$ and $T\teh \mathbf{B}_2$ can be written in matrix form $\begin{bsmallmatrix}
        1 & 1 & 0 \\
        1 & -1 & 0 \\
        0 & 0 & 1
    \end{bsmallmatrix}$ and $\begin{bsmallmatrix}
    0 & 0 & 1 \\ 
    0 & 0 & -1 \\
    1 & -1 & 0
    \end{bsmallmatrix}$ respectively, up to scalar constants.
Applying  \cref{thm:dichotomy-set-of-domain-3}
we can conclude  that $\{\mathbf{F}_2, \mathbf{G}_2, \mathbf{H}_2, \mathbf{B}_2\}$ is 
in tractable class \tractBGR.

Our new algorithm also solves some natural problems.
Consider the following problem.
For $n \in \mathbb{N}$, $i \ne j \in \{\db, \dg, \dr\}$, and any $a, b \in \mathbb{R}$, let $\parf^{n, i, j}_{a, b} : \{\db, \dg, \dr\}^n \to \mathbb{R}$ be the function
{\small
\[
\parf^{n, i, j}_{a, b} (\mathbf{x}) = \begin{cases*}
    a & if $\mathbf{x} \in \{i, j\}^n$ and $\mathbf{x}$ contains even number of $i$ \\
    b & if $\mathbf{x} \in \{i, j\}^n$ and $\mathbf{x}$ contains odd number of $i$ \\
    0 & otherwise
\end{cases*}
\]
}
Let $(\ne)_{pq;r} : \{\db, \dg, \dr\}^2 \to \{0, 1\}$ for distinct $p, q, r \in \{\db, \dg, \dr \}$ be the function
{\small
\[
(\ne)_{pq;r}(x, y) = \begin{cases*}
    1 & if $x, y \in \{p, q\}$ and $x \ne y$ \\
    1 & if $x = y = r$ \\
    0 & otherwise
\end{cases*}
\]
}
Let 
\(\mathcal{F} = \{ \parf^{n, i, j}_{a, b} : i \ne j \in \{\db, \dg, \dr\}, a, b \in \mathbb{R}  \} \cup \{ (\ne)_{pq;r} : p, q, r \in \{\db, \dg, \dr\}  \} \, . \)
There is a related constraint satisfaction decision problem, where 
\[\mathcal{F}^b = \{ \parf^{n, i, j}_{a, b} : i \ne j \in \{\db, \dg, \dr\}, a, b \in \{0, 1\}  \} \cup \{ (\ne)_{pq;r} : p, q, r \in \{\db, \dg, \dr\}  \} \, , \]
and we ask if an $\mathcal{F}^b$ signature grid has a nonzero assignment. 
It is not even immediately obvious whether this decision problem is solvable in polynomial time.
\cref{thm:dichotomy-set-of-domain-3} tells us that $\mathcal{F}$ is in class \tractBGGRBR and thus $\holts(\mathcal{F})$ is computable in polynomial time, which implies that the decision problem is also solvable in polynomial time.

%% file: sections/notations.tex
\section{Preliminaries}
\subsection{Definitions}
Definitions of Holant problem and gadget are introduced in this subsection.

Let $D$ be a finite domain set, and $\mathcal{F}$ be a set of constraint functions, called signatures.
Each $\mathbf{F} \in \mathcal{F}$ is a mapping from $D^k \to \mathbb{C}$ for some arity $k$.
If the image of $\mathbf{F}$ is contained in $\mathbb{R}$, we say $\mathbf{F}$ is real-valued.

A \textit{signature grid} $\Omega = (G, \mathcal{F}, \pi)$ consists of a graph $G = (V, E)$ where each vertex is labeled by
a function $\mathbf{F}_v \in \mathcal{F}$ and $\pi$ is the labeling.
The arity of $\mathbf{F}_v$ must match the degree of $v$.
The Holant problem on instance $\Omega$ is to evaluate
\begin{equation}
  \hol_{\Omega} = \sum_{\sigma} \prod_{v \in V} \mathbf{F}_v(\sigma |_{E(v)}) \, ,
\end{equation}
where the sum is over all edge assignments $\sigma : E \to D$ and $E(v)$ is the edges adjacent to $v$, and $\mathbf{F}_v(\sigma |_{E(v)})$ is the evaluation of $\mathbf{F}_v$ on the ordered input tuple $\sigma |_{E(v)}$.

A signature $\mathbf{F}_v$ is listed by its values lexicographically as a table, or it can be expressed as a tensor in $(\mathbb{C}^{\lvert D \rvert})^{\otimes \deg(v)}$.
We can identify a unary function $\mathbf{F}(x): D \to \mathbb{C}$ with a vector $\mathbf{u} \in \mathbb{C}^{\lvert D \rvert}$.
Given two vectors $\mathbf{u}$ and $\mathbf{v}$ of dimension $\lvert D \rvert$, the tensor product $\mathbf{u} \otimes \mathbf{v}$ is a vector in 
$\mathbb{C}^{\lvert D \rvert^2}$, with entries $u_i v_j$ for $1 \le i, j \le \lvert D \rvert$.
For matrices $A = (a_{ij})$ and $B = (b_{kl})$ the tensor product (or Kronecker product) $A \otimes B$ is defined similarly;
it has entries $a_{ij}b_{kl}$ indexed by $((i, k), (j, l))$ lexicographically.
We write $\mathbf{u}^{\otimes k}$ for $\mathbf{u}\otimes \cdots \otimes \mathbf{u}$ with $k$ copies of $\mathbf{u}$.
$A^{\otimes k}$ is similarly defined.
We have $(A \otimes B)(A' \otimes B') = (A A' \otimes BB')$ whenever the matrix products are defined.
In particular, $A^{\otimes k}(\mathbf{u}_1 \otimes \cdots \otimes \mathbf{u}_k) = A \mathbf{u}_1 \otimes \cdots \otimes A \mathbf{u}_k$ when the matrix-vector 
products $A \mathbf{u}_i$ are defined.

A signature $\mathbf{F}$ of arity $k$ is \textit{degenerate} if $\mathbf{F} = \mathbf{u}_1 \otimes \cdots \otimes \mathbf{u}_k$ for some vectors $\mathbf{u}_i$.
Such a signature is very weak; there is no interaction between the variables.
If every signature in $\mathcal{F}$ is degenerate, then $\hol_\Omega$ for any $\Omega = (G, \mathcal{F}, \pi)$ is computable in polynomial time in a trivial way:
Simply split every vertex $v$ into $\deg(v)$ vertices each assigned a unary $\mathbf{F}_i$ and connected to the incident edge.
Then $\hol_{\Omega}$ becomes a product over each component of a single edge.
Thus degenerate signatures are weak and should be properly understood as made up by unary signatures.
To concentrate on the essential features that differentiates tractability from intractability, $\hols$ was introduced in \cite{10.1145/1536414.1536511, cai_computational_2011}.
These are the problems where all unary signatures are assumed to be present, i.e. $\hols(\mathcal{F}) = \hol(\mathcal{F} \cup \mathcal{U})$ where $\mathcal{U}$ is the set of all unary signatures.
 We note that for real valued $\mathcal{F}$ the complexity of $\hols(\mathcal{F})$
 is  unchanged whether we use real valued or complex valued 
 $\mathcal{U}$~\cite{liu_restricted_nodate}(Lemma 9), and
 hence in this paper we use real valued $\mathcal{U}$.
 In the proof of \#\P-hardness, we freely use complex valued unary functions and apply the known $\hols$ dichotomy theorems that may use complex valued unary functions.

Our proof uses the  notion of a \textit{gadget}.
Consider a type of graph $G = (V, I, E)$ where $I$ and $E$ are two kinds of edges.
Edges in $I$ are ordinary internal edges with two endpoints in $V$.
Edges in $E$ are external edges (also called dangling edges) which have only one end point in $V$.
Such a graph can be made into a part of a larger graph as follows: 
Given a graph $G'$ and a vertex $v$ of $G'$, we may replace $v$ with a graph $G$ by merging the external edges of $E$ to the incident edges of $v$.

A $\mathcal{F}$ gadget consists of a graph $G = (V, I, E)$ and a labeling $\pi$ where each vertex $v \in V$ is labeled by $\mathbf{F}_v \in \mathcal{F}$.
We may view $G$ as a function $\mathbf{F}_G$, such that if we replace a vertex $v$ of a graph $G'$ by $G$, the Holant value of the resulting instance is 
as if we assign $F_G$ to $v$.
For this to hold, $\mathbf{F}_G$ must be such that for an assignment $\tau: E \to D$, 
\[
  \mathbf{F}_G(\tau) = \sum_{\sigma}\prod_{u \in V} \mathbf{F}_u (\tau \sigma |_{E(u)}) \, ,
\]
where the sum is over all edge assignments $\sigma: I \to D$ and $\tau \sigma$ is the combined assignment on $E \cup I$.

\subsection{Holographic Transformation}
To describe the idea of holographic transformations, it is convenient to consider bipartite graphs.
For a general graph, we can always transform it into a bipartite graph while preserving the Holant value, as follows:
for each edge in the graph, we replace it by a path of length $2$, and assign to the new vertex the binary Equality function $(=_2)$.

We use the notation $\hol(\mathcal{R} | \mathcal{G})$ to denote the Holant problem on bipartite graphs $H = (U, V, E)$,
where each signature for a vertex in $U$ or $V$ is from $\mathcal{R}$ or $\mathcal{G}$, respectively.
An input instance for the bipartite Holant problem is a bipartite signature grid and is denoted as $\Omega = (H ; \mathcal{R} | \mathcal{G} ; \pi)$.
Signatures in $\mathcal{R}$ are considered as row vectors (or covariant tensors); signatures in $\mathcal{G}$ are considered as column vectors (or contravariant tensors).

For a $\lvert D \rvert \times \lvert D \rvert$ matrix $T$ and a signature set $\mathcal{F}$, define
\[
  T \mathcal{F} = \{\mathbf{G} : \exists \mathbf{F} \in \mathcal{F} \text{ of arity n, such that } \mathbf{G} = T^{\otimes n} \mathbf{F}\} \, ,
\]
and similarly for $\mathcal{F}T$.
Whenever we write $T^{\otimes n} \mathbf{F}$ or $T \mathcal{F}$, we view the signatures as column vectors; similarly $\mathbf{F} T^{\otimes n}$ or $\mathcal{F} T$ as row vectors.
A holographic transformation by $T$ is the following operation:
given a signature grid $\Omega = (H ; \mathcal{R} | \mathcal{G} ; \pi)$, for the same graph $H$,
we get a new grid $\Omega' = (H ; \mathcal{R} T | T^{-1} \mathcal{G}; \pi)$ by replacing each signature in $\mathcal{R}$ or $\mathcal{G}$
with the corresponding signature in $\mathcal{R}T$ or $T^{-1} \mathcal{G}$.

\begin{theorem}[Valiant's Holant Theorem \cite{10.1109/FOCS.2006.7}]
If there is a holographic transformation mapping signature grid $\Omega$ to $\Omega'$, then $\hol_{\Omega} = \hol_{\Omega'}$.
\end{theorem}

Therefore, an invertible holographic transformation does not change the complexity of the Holant problem in the bipartite setting.
Furthermore, if $T$ is orthogonal, then $(=_2)T^{\otimes 2} = T\transpose I T = I$, so it preserves binary equality.
This means that an orthogonal holographic transformation can be used freely in the standard setting.
\begin{corollary}
Suppose $T$ is an orthogonal matrix, $T\transpose T = I$, and let $\Omega = (G, \mathcal{F}, \pi)$ be a signature grid.
Under a holographic transformation by $T$, we get a new signature grid $\Omega' = (G, T \mathcal{F}, \pi)$ and $\hol_{\Omega} = \hol_{\Omega'}$.
\end{corollary}

\subsection{Notation} \label{subsec:notations}
For two nonzero vectors $\mathbf{x}, \mathbf{y} \in \mathbb{C}^n$, we write $\mathbf{x} \sim \mathbf{y}$ to denote projective equality, i.e. $\mathbf{x} = \lambda \mathbf{y}$ for some nonzero $\lambda \in \mathbb{C}$.
For two tuples of vectors $X = (\mathbf{x}_1, \ldots, \mathbf{x}_m)$ and $Y = (\mathbf{y}_1, \ldots, \mathbf{y}_m)$, we write $X \sim Y$ if $\mathbf{x}_i \sim \mathbf{y}_i$ for all $1 \le i \le m$ after some reordering.
Throughout this paper, the symbol $\langle \mathbf{u}, \mathbf{v} \rangle$ for $\mathbf{u}, \mathbf{v} \in \mathbb{C}^n$ denotes the dot product,
i.e. $\langle \mathbf{u}, \mathbf{v} \rangle = \sum u_i v_i$.
We say $\mathbf{u}, \mathbf{v} \in \mathbb{C}^n$ are orthogonal if $\langle \mathbf{u}, \mathbf{v} \rangle = 0$

For the ease of notation, we do not distinguish between column vectors and row vectors in this paper when the intention is clear from the context.

A signature $\mathbf{F}$ of arity $k$ is \textit{symmetric} if $\mathbf{F}(x_1, \ldots, x_k) = \mathbf{F}(x_{\sigma(1)}, \ldots, x_{\sigma(k)})$ for all $\sigma \in S_k$, the symmetric group.
It can be shown that a symmetric signature is degenerate if and only if $\mathbf{F} = \mathbf{u}^{\otimes k}$ for some unary $\mathbf{u}$.
In this paper, if not further specified, a signature $\mathbf{F}$ is assumed to be real-valued, symmetric, and on domain $3$.

We consider a signature $\mathbf{F}$ and its nonzero multiple $c \mathbf{F}$ as the same signature, since replacing $\mathbf{F}$ by $c \mathbf{F}$ only introduces a easily computable global factor in the Holant value.

A symmetric signature $\mathbf{F}$ on $k$ Boolean variables $\{0, 1\}$ can be expressed as $[f_0, f_1, \ldots, f_k]$ where $f_i$ is the value of $\mathbf{F}$ 
on inputs of Hamming weight $i$.
In this paper, we focus on signatures on domain size $3$, and we use the symbols $\{\db, \dg, \dr\}$ to denote the domain elements.
%
A binary signature $\mathbf{F}$ (not necessarily symmetric) can be expressed as a $\lvert D \rvert \times \lvert D \rvert$ matrix $M_{\mathbf{F}}$, where the entry $(i, j) \in D \times D$
is the value of $\mathbf{F}(i, j)$.
For the ease of notation, we use the term matrix and binary signature interchangeably, and use $\mathbf{F}$ to refer to both a signature and its matrix $M_{\mathbf{F}}$.
To fix an ordering, binary signature on domain $3$ is expressed as the following:
\[
  \mathbf{F} = \begin{bmatrix}
    f_{\db\db} & f_{\db \dg} & f_{\db \dr} \\
    f_{\dg \db} & f_{\dg \dg} & f_{\dg \dr} \\
    f_{\dr \db} & f_{\dr \dg} & f_{\dr \dr} 
  \end{bmatrix} \, .
\]
If $\mathbf{F}$ is a symmetric signature, then $\mathbf{F}$ is a symmetric matrix.

Let $\mathbf{G}$ be a binary signature and $\mathbf{F}$ be a symmetric signature of arity $k \ge 2$.
We use $\mathbf{G}^{\otimes k} \mathbf{F}$ to denote the gadget constructed by attaching a $\mathbf{G}$ at the edges of $\mathbf{F}$.
An example for $k = 3$ is shown in \cref{fig:gf}.
\begin{figure}
\centering
      \begin{tikzpicture}[scale=0.7]
        \node[mycirc, label=right:{$\mathbf{F}$}] (f1) at (19, 0) {};
        \node[mycirc, label=above:{$\mathbf{G}$}] (f2) at (17, 0) {};

        \node[mycirc, label=above:{$\mathbf{G}$}] (g11) at (17, 2) {};
        \node[mycirc, label=above:{$\mathbf{G}$}] (g21) at (17, -2) {};

        \draw (f1) to[out=90, in=0] (g11) -- (16, 2);
        \draw (f1) to[out=-90, in=0] (g21) -- (16, -2);

        \draw (f1) -- (f2) -- (16, 0);
  \end{tikzpicture}
  \caption{Gadget $\mathbf{G}^{\otimes 3} \mathbf{F}$}\label{fig:gf}
\end{figure}
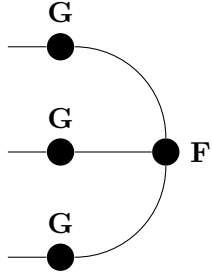
If $\mathbf{F}$ is written in a tensor form, i.e. $\mathbf{F} = \mathbf{v}_1^{\otimes k} + \cdots + \mathbf{v}_s^{\otimes k}$ for $\mathbf{v}_i \in \mathbb{C}^{\lvert D \rvert}$
we can easily check that $\mathbf{G}^{\otimes k} \mathbf{F} = (\mathbf{G} \mathbf{v}_1)^{\otimes k} + \cdots + (\mathbf{G} \mathbf{v}_s)^{\otimes k}$.
This gadget construction will be used throughout the paper.

Another gadget construction common in this paper is connecting a unary signature.
For the ease of notation, we identify a vector $\mathbf{u} \in \mathbb{C}^{\lvert D \rvert}$ with a unary signature on domain $D$.
Let $\mathbf{u} \in \mathbb{C}^{\lvert D \rvert}$ and $\mathbf{F}$ be a symmetric signature on domain $D$ of arity $k$.
Then, $\langle \mathbf{F}, \mathbf{u} \rangle$ is the arity $k - 1$ gadget obtained by connecting $\mathbf{u}$ to any edge of $\mathbf{F}$.
Since $\mathbf{F}$ is symmetric, the choice of the edge does not matter.

We use $\hol_2$ to denote the Holant problem on Boolean domain $\{0, 1\}$, and $\hol_3$ to denote the Holant problem on domain $\{\db, \dg, \dr\}$.
We say two sets of signatures $\mathcal{F}$ and $\mathcal{G}$ are \textit{compatible} if $\hol(\mathcal{F} \cup \mathcal{G})$ is tractable.

For a domain $3$ signature, we use the symbol $\mathbf{F}\domres{i, j}$ to denote the domain $2$ signature obtained by restricting the inputs of $\mathbf{F}$ to be from $\{i, j\} \subset \{\db, \dg, \dr\}$.
We extend this notation to set of signatures $\mathcal{F}$: $\mathcal{F}\domres{i, j} := \{ \mathbf{F}\domres{i, j} : \mathbf{F} \in \mathcal{F} \}$.
When we take a domain restriction of a domain $3$ signature to Boolean domain, we will identify the domains in the following way:
\begin{itemize}
  \item $\mathbf{F}\domres{\db, \dg}$ is viewed as identifying $\db$ as $0$ and $\dg$ as $1$.
  \item $\mathbf{F}\domres{\db, \dr}$ is viewed as identifying $\db$ as $0$ and $\dr$ as $1$.
  \item $\mathbf{F}\domres{\dg, \dr}$ is viewed as identifying $\dg$ as $0$ and $\dr$ as $1$.
\end{itemize}

Let $\mathbf{F}$ be a domain $3$ signature.
We define $\supp \mathbf{F}$ to be the set of inputs for which $\mathbf{F}$ is nonzero.
We say $\mathbf{F}$ is \strspt (a signature defined essentially on a Boolean domain) if $\supp \mathbf{F} \subseteq \{\db, \dg\}^*$, $\supp \mathbf{F} \subseteq \{\db, \dr\}^*$ or $\supp \mathbf{F} \subseteq \{\dg, \dr\}^*$.
Note that if $\supp \mathbf{F} = \{\db\}^*$ for a symmetric $\mathbf{F}$, then $\mathbf{F} = a \mathbf{e_1}^{\otimes n}$, and thus degenerate.

We say $\mathbf{F}$ is \textit{domain separated} to $\{\db, \dg\}$ and $\{\dr\}$, written $\mathbf{F}$ is $\db \dg | \dr$, if $\supp \mathbf{F} \subseteq \{\db, \dg\}^* \cup \{\dr\}^*$.
In other words, $\mathbf{F}$ is zero on inputs that take values from both $\{\db, \dg\}$ and $\{\dr\}$.
It is possible that $\supp \mathbf{F} \subseteq \{\db\}^* \cup \{\dg\}^* \cup \{\dr\}^*$, in which case $\mathbf{F} =a  \mathbf{e_1}^{\otimes n} + b \mathbf{e_2}^{\otimes n} + c \mathbf{e_3}^{\otimes n}$ for some $a, b, c \in \mathbb{R}$.
We call such $\mathbf{F}$ a \geneq signature.
We similarly define domain separation to $\{i, j\}$ and $\{k\}$ and write $ij | k$ for any distinct $i, j, k \in \{\db, \dg, \dr\}$.
We also refer to a matrix $\db \dg | \dr$ if it can be viewed as a $\db \dg | \dr$ binary signature.
For example, $M = \begin{bsmallmatrix}
  * & * & 0 \\
  * & * & 0 \\
  0 & 0 & *
\end{bsmallmatrix}$ is a $\db \dg | \dr$ matrix.

Let $M$ be a $\db \dg | \dr$ matrix.
We can easily check the following two facts.
If $\mathbf{F}$ is a $\db \dg | \dr$ signature of arity $k$, then $M^{\otimes k} \mathbf{F}$ is $\db \dg | \dr$ as well.
If $\mathbf{G}$ is a signature of arity $k$ such that $\supp \mathbf{G} \subseteq \{\db, \dg\}^*$, then $\supp M^{\otimes k} \mathbf{G} \subseteq \{\db, \dg\}^*$ as well.


Denote by $\mathcal{E}$  the set of all functions $\mathbf{F}$ such that if $\mathbf{F}$ has arity $n$, then $\supp \mathbf{F} \subseteq \{a, b, c\}$ for $a = (a_1, \ldots, a_n), b = (b_1, \ldots, b_n), c = (c_1, \ldots, c_n) \in \{\db, \dg, \dr\}^n$ such that for all $1 \le i \le n$, $a_i, b_i, c_i$ are all distinct.
We think of $\mathcal{E}$ as a generalized form of \geneq function to not necessarily symmetric functions.

We use $\mathcal{D}$ to denote the set of $3 \times 3$ matrices such that the first two columns are linearly dependent and also the first two rows are linearly dependent.
In other words,
\[
  \mathcal{D} = \left\{ \begin{bmatrix}
      - & x \mathbf{v} & - \\
      - & y \mathbf{v} & - \\
      - & \mathbf{u} & -
      \end{bmatrix} = \begin{bmatrix}
      \vert & \vert & \vert \\
      x' \mathbf{v}' & y' \mathbf{v}' & \mathbf{u}' \\
      \vert & \vert & \vert
  \end{bmatrix} : x, y, x', y' \in \mathbb{R}, \mathbf{v}, \mathbf{u}, \mathbf{v}', \mathbf{u}' \in \mathbb{R}^3 \right\}
  \,.
\]
We can easily check that $\mathcal{D}$ is closed under multiplication.
If $M$ is a $\db \dg | \dr$ matrix, we can see that $M \mathcal{D}, \mathcal{D} M \subseteq \mathcal{D}$.
Also, the symmetric matrices in $\mathcal{D}$ have the following form:
$\begin{bsmallmatrix}
  a x^2 & a xy & bx \\
  axy & ay^2 & by \\
  bx & by & c
\end{bsmallmatrix}$ 
for some $a, b, c, x, y \in \mathbb{R}$.

We use $\genperm$ to denote the $3 \times 3$ generalized permutation matrices, which are matrices such that each row and column contains at most one nonzero real value.
We use $\octgroup$ to denote the symmetry group of an octahedron.
As a subgroup of the real $3 \times 3$ orthogonal group $O(3)$, $\octgroup$ consists of the matrices
\[
\begin{bmatrix}
    \epsilon_1 & 0 & 0 \\
    0 & \epsilon_2 & 0 \\
    0 & 0 & \epsilon_3
\end{bmatrix},
\begin{bmatrix}
    \epsilon_1 & 0 & 0 \\
    0 & 0 & \epsilon_2 \\
    0 & \epsilon_3 & 0
\end{bmatrix},
\begin{bmatrix}
    0 & \epsilon_1 & 0 \\
    \epsilon_2 & 0 & 0 \\
    0 & 0 & \epsilon_3
\end{bmatrix},
\begin{bmatrix}
    0 & \epsilon_1 & 0 \\
    0 & 0 & \epsilon_2 \\
    \epsilon_3 & 0 & 0 
\end{bmatrix},
\begin{bmatrix}
    0 & 0 & \epsilon_1 \\
    \epsilon_2 & 0 & 0 \\
    0 & \epsilon_3 & 0
\end{bmatrix},
\begin{bmatrix}
    0 & 0 & \epsilon_1 \\
    0 &  \epsilon_2 & 0\\
     \epsilon_3 & 0 & 0
\end{bmatrix}
\]
for $\epsilon_1, \epsilon_2, \epsilon_3 \in \{1, -1\}$.
Thus, $\octgroup$ has order $48$.
The symmetric matrices in $\octgroup$ are of the form 
\[
  \begin{bmatrix}
      \epsilon_1 & 0 & 0 \\
      0 & 0 & \epsilon_2 \\
      0 & \epsilon_2 & 0
      \end{bmatrix}, 
       \begin{bmatrix}
      0 & \epsilon_2 & 0 \\
      \epsilon_2 & 0 & 0 \\
      0 & 0 & \epsilon_1
      \end{bmatrix}, 
       \begin{bmatrix}
      0 & 0 & \epsilon_2 \\
      0 & \epsilon_1 & 0 \\
      \epsilon_2 & 0 & 0
  \end{bmatrix} 
\]
for $\epsilon_1, \epsilon_2 \in \{1, -1\}$.

We call a signature/matrix of the form $\begin{bsmallmatrix}
  0 & 0 & a \\
  0 & 0 & b \\
  c & d & 0
\end{bsmallmatrix}$ for some $a, b, c, d \in \mathbb{R}$ as a \swbg signature/matrix.
The intuition is that a signature of this form swaps the domains $\{\db, \dg\}$ and $\{\dr\}$ in a degenerate way.
Similarly, we call $\begin{bsmallmatrix}
  0 & a & 0 \\
  b & 0 & c \\
  0 & d & 0
\end{bsmallmatrix}$ as \swbr and $\begin{bsmallmatrix}
  0 & a & b \\
  c & 0 & 0 \\
  d & 0 & 0
\end{bsmallmatrix}$ as \swgr.

\subsection{Symmetric Tensor Rank}
We follow \cite{comon_symmetric_2008}.
$\mathsf{S}^k(\mathbb{C}^n)$ is the set of order-$k$ symmetric tensor over $\mathbb{C}^n$.
In our setting, we may think of it as the set of symmetric complex-valued $k$-arity signatures over a domain of size $n$.
\begin{definition}[Definition 4.1 of \cite{comon_symmetric_2008}]
  The symmetric rank of $A \in \mathsf{S}^k(\mathbb{C}^n)$ is defined as 
  \[
    \rank(A) := \min \{s : A = \sum_{i = 1}^s \mathbf{y}_i^{\otimes k}\}
  \]
\end{definition}
By Lemma 4.2 in \cite{comon_symmetric_2008}, symmetric rank always exists.

\begin{corollary}[Corollary 4.4 of \cite{comon_symmetric_2008}]\label{cor:pairiwse-linearly-independent-tensor}
  Let $\mathbf{v}_1, \ldots, \mathbf{v}_r \in \mathbb{C}^n$ be $r$ pairwise linearly independent vectors.
  Then, for any $k \ge r - 1$, the rank-$1$ symmetric tensors
  \[
    \mathbf{v}_1^{\otimes k}, \ldots, \mathbf{v}_r^{\otimes k} \in \mathsf{S}^k(\mathbb{C}^n)
  \]
  are linearly independent.
\end{corollary}

\begin{lemma}[Lemma 5.1 of \cite{comon_symmetric_2008}]\label{lem:linearly-independent-tensor-rank}
  Let $\mathbf{y}_1, \ldots, \mathbf{y}_s \in \mathbb{C}^n$ be linearly independent and $k \ge 2$ an integer.
  Then, the symmetric tensor defined by $A := \sum_{i = 1}^{s} \mathbf{y}_i^{\otimes k}$ has $\rank(A) = s$.
\end{lemma}
We derive a simple result on the uniqueness by adapting the proof of \cref{lem:linearly-independent-tensor-rank}.
\begin{proposition}\label{prop:uniqueness-tensor-decomposition}
  Let $k \ge 3$.
  Let $\mathbf{y}_1, \ldots, \mathbf{y}_s, \mathbf{z}_1, \ldots, \mathbf{z}_s \in \mathbb{C}^n$ such that $\{\mathbf{y}_i\}_{1 \le i \le s}$ and $\{\mathbf{z}_i\}_{1 \le i \le s}$ are two sets of linearly independent vectors.
  Suppose
  \begin{equation}\label{eq:symmetric-sum-equal}
    \sum_{i = 1}^{s} \mathbf{y}_i^{\otimes k} = \sum_{i = 1}^{s} \mathbf{z}_i^{\otimes k} \, .
  \end{equation}
  Then, $(\mathbf{y}_1, \ldots, \mathbf{y}_s) \sim (\mathbf{z}_1, \ldots, \mathbf{z}_s)$.
\end{proposition}
\begin{proof}
By linear independence of $\mathbf{y}_1, \ldots, \mathbf{y}_s$, there exist covectors $\phi_1, \ldots, \phi_s \in (\mathbb{C}^n)^*$ that are dual to $\mathbf{y}_1, \ldots, \mathbf{y}_s$, i.e. 
$\phi_i(\mathbf{y}_j) = \delta_{ij}$ where $\delta_{ij}$ is the Kronecker delta.
Similarly, let $\psi_1, \ldots, \psi_s \in (\mathbb{C}^n)^*$ be covectors dual to $\mathbf{z}_1, \ldots, \mathbf{z}_s$.

Applying $\phi_1$ to both sides of \cref{eq:symmetric-sum-equal} gives
$\mathbf{y}_1 ^{\otimes k - 1} = \sum_{i = 1}^{s} c_i \mathbf{z}_i^{\otimes k - 1}$ for $c_i = \phi_1(\mathbf{z}_i)$.
Since the left side is nonzero, it must be the case that there exists some $i$ such that $c_i \ne 0$.
Applying $\psi_i$ to both sides gives
$\psi_i(\mathbf{y}_1) \mathbf{y}_1^{\otimes k - 2} = c_i \mathbf{z}_i^{\otimes k - 2}$.
Since the right side is nonzero, it must be the case that $\psi_i(\mathbf{y}_1) \ne 0$.
Since $k - 2 \ge 1$ by assumption, we may repeatedly apply $\phi_1$ to both sides until we have $a \mathbf{y}_1 = b \mathbf{z}_i$ for some nonzero $a, b \in \mathbb{C}$, so $\mathbf{y}_1 \sim \mathbf{z}_i$.

After reordering, we may assume $i = 1$ and move $\mathbf{z}_1^{\otimes k}$ to the other side in \cref{eq:symmetric-sum-equal} to obtain
\begin{equation}\label{eq:after-moving-z}
\sum_{i = 2}^{s} \mathbf{z}_i^{\otimes k} = 
\mathbf{y}_1^{\otimes k} - \mathbf{z}_1^{\otimes k} + \sum_{i = 2}^s \mathbf{y}_i^{\otimes k} = c \mathbf{y}_1^{\otimes k} + \sum_{i = 2}^{s} \mathbf{y}_i^{\otimes k}
\end{equation}
for some $c \in \mathbb{C}$ since $\mathbf{y}_1 \sim \mathbf{z}_1$.
If $c \ne 0$, the tensor on the right hand side of \cref{eq:after-moving-z} has rank $s$ while the left hand side has rank $s - 1$ by \cref{lem:linearly-independent-tensor-rank}.
This is a contradiction, so $c = 0$.
(Note that this does not directly imply $\mathbf{y}_1 = \mathbf{z}_1$ since it may be the case that for even $k$, $\mathbf{y}_1 = - \mathbf{z}_1$.)

Since $c = 0$ in \cref{eq:after-moving-z}, we may complete the proof by induction.
\end{proof}

\subsection{Known Dichotomy Theorems}
A $\holbs(\mathcal{F})$ dichotomy on a set of symmetric, complex-valued signatures $\mathcal{F}$ is known.
\begin{definition}[Definition 2.9 in \cite{cai_complexity_2017}]\label{def:boolean-signature-type-i-ii}
  A signature $[x_0, x_1, \ldots, x_n]$, where $n \ge 2$, has type $\typei(a, b)$, if there exist $a$ and $b$ (not both $0$), such that $a x_k + b x_{k+1} = a x_{k+2}$ for $0 \le k \le n -2$.
  We say it it is of type $\typeii$, if $x_k = - x_{k+2}$ for $0 \le k \le n - 2$.
\end{definition}
\begin{theorem}[Theorem 2.12 in \cite{cai_complexity_2017}]\label{thm:dich-sym-Boolean}
Let $\mathcal{F}$ be a set of nondegenerate symmetric signatures over $\mathbb{C}$ in Boolean variables.
Then $\holbs(\mathcal{F})$ is computable in polynomial time for the following three classes of $\mathcal{F}$.
In all other cases, $\holbs(\mathcal{F})$ is \#\P-hard.
\begin{enumerate}
  \item Every signature in $\mathcal{F}$ is of arity $\le 2$.
  \item There exists $a$ and $b$ (not both $0$, depending only on $\mathcal{F}$), such that every signature in $\mathcal{F}$ either (1) has type $\typei(a,b)$ or (2) has arity $2$ and is of the form $[2a \lambda, b \lambda, -2 a \lambda]$.
  \item Every signature in $\mathcal{F}$ either (1) has type $\typeii$ or (2) has arity $2$ and is of the form $[\lambda, 0, \lambda]$.
\end{enumerate}
\end{theorem}

A $\holts(\mathbf{F})$ dichotomy for a single symmetric, complex-valued signature of arity $3$ is known, but it has an exceptional tractable case, the case 3.
\begin{theorem}[Theorem 3.1, 3.2 in \cite{cai_dichotomy_2013}]
Let $\mathbf{F}$ be a symmetric, complex valued, ternary signature over domain $\{\db, \dg, \dr\}$.
Then $\holts(\mathbf{F})$ \sph unless $\mathbf{F}$ is one of the following forms, in which case the problem is solvable in polynomial time.
\begin{enumerate}
  \item There exists three vectors $\boldsymbol{\alpha}, \boldsymbol{\beta}, \boldsymbol{\gamma} \in \mathbb{C}^3$ such that $\langle \boldsymbol{\alpha}, \boldsymbol{\beta} \rangle = 0$, $\langle \boldsymbol{\alpha}, \boldsymbol{\gamma} \rangle = 0$, and $\langle \boldsymbol{\beta}, \boldsymbol{\gamma} \rangle = 0$, and 
    \[
      \mathbf{F} = \boldsymbol{\alpha}\teh + \boldsymbol{\beta}\teh + \boldsymbol{\gamma}\teh
    \]
  \item There exists three vectors $\boldsymbol{\alpha}, \boldsymbol{\beta}_1, \boldsymbol{\beta}_2 \in \mathbb{C}^3$ such that 
    $\langle \boldsymbol{\alpha}, \boldsymbol{\beta}_1 \rangle = 0$, $\langle \boldsymbol{\alpha}, \boldsymbol{\beta}_2 \rangle = 0$, 
    $\langle \boldsymbol{\beta}_1, \boldsymbol{\beta}_1 \rangle = 0$,
    $\langle \boldsymbol{\beta}_2, \boldsymbol{\beta}_2 \rangle = 0$ and 
    \[
      \mathbf{F} = \boldsymbol{\alpha}\teh + \boldsymbol{\beta}_1\teh + \boldsymbol{\beta}_2\teh
    \]
  \item There exists two vectors $\boldsymbol{\beta}, \boldsymbol{\gamma} \in \mathbb{C}^3$ and a signature $\mathbf{F}_{\boldsymbol{\beta}}$ of arity $3$, such that
    $\boldsymbol{\beta} \ne 0$, $\langle \boldsymbol{\beta}, \boldsymbol{\beta} \rangle = 0$, $\langle \mathbf{F}_{\boldsymbol{\beta}}, \boldsymbol{\beta} \rangle = 0$ and
    \[
      \mathbf{F} = \mathbf{F}_{\boldsymbol{\beta}} + \boldsymbol{\beta}^{\otimes 2} \otimes \boldsymbol{\gamma} + \boldsymbol{\beta} \otimes \boldsymbol{\gamma} \otimes \boldsymbol{\beta} + \boldsymbol{\gamma} \otimes \boldsymbol{\beta}^{\otimes 2}
    \]
\end{enumerate}
These cases are equivalent to an existence of an orthogonal transformation $T$, such that 
\begin{enumerate}
  \item 
    For some $a, b, c \in \mathbb{C}$
    \[
      T\teh \mathbf{F} = a \mathbf{e}_1\teh + b \mathbf{e}_2\teh + c \mathbf{e}_3\teh 
    \]
  \item 
    For some $c \ne 0$ and $\lambda \in \mathbb{C}$,
    \[
      c T\teh \mathbf{F} =  \boldsymbol{\beta}_0\teh + \overline{\boldsymbol{\beta}_0}\teh + \lambda \mathbf{e}_3 \teh  
    \]
    where $\boldsymbol{\beta}_0 = \frac{1}{\sqrt{2}}(1, i, 0)\transpose$ and $\overline{\boldsymbol{\beta}}$ is its conjugate $\frac{1}{\sqrt{2}}(1, -i, 0)\transpose$.
  \item 
    For $\epsilon \in \{0, 1\}$,
    \[
     T\teh \mathbf{F} = 
     \mathbf{F}_0 + \epsilon \left( \boldsymbol{\beta}^{\otimes 2} \otimes \overline{\boldsymbol{\beta}} + \boldsymbol{\beta} \otimes \overline{\boldsymbol{\beta}} \otimes \boldsymbol{\beta} + \overline{\boldsymbol{\beta}} \otimes \boldsymbol{\beta}^{\otimes 2}\right)
    \]
    where $\mathbf{F}_0$ satisfies the annihilation condition $\langle \mathbf{F}_0, \boldsymbol{\beta}_0 \rangle = 0$.
\end{enumerate}
\end{theorem}

It is shown in \cite{liu_restricted_nodate} that when $\mathbf{F}$ is assumed to be a real-valued signature, case 3 does not occur.
\begin{theorem}[Theorem 2 in \cite{liu_restricted_nodate}]
  \label{thm:dich-single-sym-ter-dom3-real}
  Let $\mathbf{F}$ be a real-valued symmetric ternary function over domain $\{\db, \dg, \dr\}$. 
  Then, $\holts(\mathbf{F})$ is \#\P-hard unless the function $\mathbf{F}$ is expressible as one of the following two forms, in which case the problem is in \P.
  \begin{enumerate}
    \item[\ternarytractgeneq.] \label{ternary-tract-geneq}
      $\mathbf{F} = \boldsymbol{\alpha} \teh + \boldsymbol{\beta} \teh + \boldsymbol{\gamma}\teh$ 
      where $\boldsymbol{\alpha}, \boldsymbol{\beta}, \boldsymbol{\gamma} \in \mathbb{R}^3$ 
      and $\langle \boldsymbol{\alpha}, \boldsymbol{\gamma} \rangle = \langle \boldsymbol{\beta}, \boldsymbol{\gamma} \rangle = \langle \boldsymbol{\gamma}, \boldsymbol{\alpha} \rangle = 0$.
    \item[\ternarytractz.] \label{ternary-tract-Z}
      $\mathbf{F} = \boldsymbol{\alpha} \teh + \boldsymbol{\beta} \teh + \overline{\boldsymbol{\beta}}\teh$ where $\boldsymbol{\alpha} \in \mathbb{R}^3$, $\langle \boldsymbol{\alpha}, \boldsymbol{\beta} \rangle = \langle \boldsymbol{\beta}, \boldsymbol{\beta} \rangle = 0$.
  \end{enumerate}
  This is equivalent to the existence of a real orthogonal transformation $T$, such that 
  \begin{enumerate}
    \item[\ternarytractgeneq.] $T\teh \mathbf{F} = a \mathbf{e_1}\teh + b \mathbf{e_2} \teh + c \mathbf{e_3}\teh$ for some $a, b, c \in \mathbb{R}$.
    \item[\ternarytractz.] $c T\teh \mathbf{F}  = \epsilon(\boldsymbol{\beta_0}\teh + \overline{\boldsymbol{\beta_0}}\teh) + \lambda \mathbf{e_3}\teh$ where $\boldsymbol{\beta_0} = (1, i, 0)\transpose$, $\epsilon \in \{0, 1\}$ and for some $c, \lambda \in \mathbb{R}$ and $c \ne 0$.
  \end{enumerate}
\end{theorem}

\subsection{Miscellaneous}
\begin{proposition}\label{prop:real-degenerate-real-vector}
  Suppose $\mathbf{v}^{\otimes k}$ for a nonzero $\mathbf{v} \in \mathbb{C}^n$ is real-valued.
  Then, there exists a nonzero $\mathbf{u} \in \mathbb{R}^n$ and nonzero $\lambda \in \mathbb{R}$ such that $\mathbf{v}^{\otimes k} = \lambda \mathbf{u}^{\otimes k}$.
\end{proposition}
\begin{proof}
  Suppose $k = 1$.
  Then, $\mathbf{v}^{\otimes 1} = \mathbf{v}$ is real-valued if and only if $\mathbf{v} \in \mathbb{R}^n$.

  Suppose $n = 1$.
  Then, $\mathbf{v}$ is a complex number $v \in \mathbb{C}$.
  We write $v = r e^{i\theta}$ a nonzero $r \in \mathbb{R}$ and $v^k = r^k e^{i k \theta}$.
  $v^k \in \mathbb{R}$ implies that $e^{i k \theta} = \epsilon$ for $\epsilon \in \{1, -1\}$.
  So, $v^k = \epsilon r^k$.

  Suppose $k \ge 2$ and $n \ge 2$.
  We write $\mathbf{v} = (v_1, v_2, \ldots, v_n)$.
  We write $v_j = r_j e^{i \theta_j}$ and 
  $\mathbf{v}^{\otimes k}$ being real-valued implies that $v_j^k \in \mathbb{R}$ for each $j$, which implies that $k \theta_j$ is an integer multiple of $\pi$.
  Therefore, we may choose $\theta_j = \frac{2 \pi m_j}{2k}$ for some $m_j \in \mathbb{Z}$.
  Also, we must have $v_1 v_j^{k-1} \in \mathbb{R}$ for all $2 \le j \le n$, which in turn implies that $\theta_1 + (k-1)\theta_j$ is an integer multiple of $\pi$.
  Then, it must be the case that $\frac{m_1}{k} + \frac{(k-1)m_j}{k}$ is an integer, so $m_1 + (k-1) m_j \equiv 0 \mod k$.
  Therefore, $m_j = m_1 + k a_j$ for some $a_j \in \mathbb{Z}$.

  We may factor out $e^{i \theta_1}$ from $\mathbf{v}$ and obtain $\mathbf{u} = (u_1, u_2, \ldots, u_n) = e^{-i \theta_1} \mathbf{v}$ where
  \[
    u_j = e^{-i \theta_1} v_j = e^{-i \pi m_1/k} r_j e^{i \pi m_j/k} = r_j e^{i \pi a_j} \in \mathbb{R}
  \]
  Therefore, $\mathbf{u} \in \mathbb{R}^n$.
  Let $\lambda = (e^{i \theta_1})^k \in \mathbb{R}$.
  Then, $\lambda \mathbf{u}^{\otimes k} = (e^{i \theta_1} \mathbf{u})^{\otimes k} = \mathbf{v}^{\otimes k}$ as desired.
\end{proof}
Therefore, we may always assume that a real-valued, symmetric, degenerate signature is a tensor power of a real vector by \cref{prop:real-degenerate-real-vector} and rescaling.

A characterization of a signature that will be often used in this paper is in terms of the symmetric tensor decomposition.
Proposition 2.10, 2.11, and the tractability proof of Theorem 2.12 of \cite{cai_complexity_2017} yields the following:
\begin{proposition}\label{prop:boolean-type-i-tensor-decomposition}
  A nondegenerate Boolean domain signature $\mathbf{F}$ is of type $\typei(a, b)$ if and only if $\mathbf{F} = c \begin{bmatrix}
  \alpha \\ \beta
  \end{bmatrix}^{\otimes n} + d \begin{bmatrix}
  \gamma \\ \delta
  \end{bmatrix}^{\otimes n}$ where $\alpha \gamma + \beta \delta = 0$ and $a = \alpha \gamma = - \beta \delta$ and $b = \alpha \delta + \beta \gamma$.
  Also, the set of orthogonal vectors $\left\{ \begin{bmatrix}
  \alpha \\ \beta
  \end{bmatrix}, \begin{bmatrix}
  \gamma \\ \delta
  \end{bmatrix} \right\}$ is uniquely determined by $a, b$ up to a scalar multiple.
\end{proposition}
\begin{proposition}\label{prop-boolean-type-ii-tensor-decomposition}
  A nondegenerate Boolean domain signature $\mathbf{F}$ is of type $\typeii$ if and only if 
  $\mathbf{F} = c \begin{bmatrix}
  1 \\ i
  \end{bmatrix}^{\otimes n} + d \begin{bmatrix}
  1 \\ -i
  \end{bmatrix}^{\otimes n}$.
  A real valued $\mathbf{F}$ is of type $\typeii$ if and only if $\mathbf{F} = (\mathbf{u} + i \mathbf{v})^{\otimes n} + (\mathbf{u} - i \mathbf{v})^{\otimes n}$ for orthogonal $\mathbf{u}, \mathbf{v} \in \mathbb{R}^2$ with $\langle \mathbf{u}, \mathbf{u} \rangle = \langle \mathbf{v}, \mathbf{v} \rangle$.
\end{proposition}
\begin{proof}
  First part is from Proposition 2.11 of \cite{cai_complexity_2017}.
  The if part is immediate.

  Suppose $\mathbf{F}$ is real valued and of type $\typeii$. 
  By the first part, for $\mathbf{F} = [x, y, -x, -y, \ldots]$, we have
  $x = c + d$ and $y = i(c - d)$ for some $c, d \in \mathbb{C}$.
  Writing $c = p + iq$ and $d = r + is$, $x, y \in \mathbb{R}$ implies $q = -s$ and $p = r$, meaning $c = \overline{d}$.
  Those in turn imply that $x = 2p$ and $y = -2q$.
  Let $a + ib$ for $a, b \in \mathbb{R}$ be any $n$th root of $c$.
  We have
  \[
    \mathbf{F} = c \begin{bmatrix}
      1 \\ i
      \end{bmatrix}^{\otimes n} + \overline{c} \begin{bmatrix}
      1 \\ -i
      \end{bmatrix}^{\otimes n} = \begin{bmatrix}
      (a + ib) \\ i(a + ib)
      \end{bmatrix}^{\otimes n} + \begin{bmatrix}
      (a - ib) \\ -i(a - ib)
    \end{bmatrix}^{\otimes n} = (\mathbf{u} + i \mathbf{v})^{\otimes n} + (\mathbf{u} - i \mathbf{v})^{\otimes n}
  \]
  where $\mathbf{u} = (a, -b)\transpose$ and $\mathbf{v} = (b, a)\transpose$.
\end{proof}

The dichotomy theorems with symmetric tensor decomposition and \cref{prop:uniqueness-tensor-decomposition} imply the following criteria for determining the hardness of a signature from the tensor decomposition.
\begin{corollary}\label{cor:non-orthogonal-linearly-independent-hard}
  Let $\boldsymbol{\alpha}, \boldsymbol{\beta} \in \mathbb{R}^2$ be linearly independent and $\langle \boldsymbol{\alpha}, \boldsymbol{\beta} \rangle \ne 0$.
  Then, the following are \sph.
  \begin{enumerate}
    \item $\holbs(\boldsymbol{\alpha}\teh + \boldsymbol{\beta}\teh)$
    \item $\holbs( (\boldsymbol{\alpha} + i \boldsymbol{\beta})\teh + (\boldsymbol{\alpha} - i \boldsymbol{\beta})\teh)$
  \end{enumerate}
  Let $\boldsymbol{\alpha}, \boldsymbol{\beta}, \boldsymbol{\gamma} \in \mathbb{R}^3$ be linearly independent and $\langle \boldsymbol{\alpha}, \boldsymbol{\beta} \rangle \ne 0$.
Then, the following are \sph.
\begin{enumerate}
\item $\holts(\boldsymbol{\alpha}\teh + \boldsymbol{\beta}\teh)$
  \item $\holts(\boldsymbol{\alpha}\teh + \boldsymbol{\beta}\teh + \boldsymbol{\gamma}\teh)$
  \item $\holts( (\boldsymbol{\alpha} + i \boldsymbol{\beta})\teh + (\boldsymbol{\alpha} - i \boldsymbol{\beta})\teh)$
  \item $\holts( (\boldsymbol{\alpha} + i \boldsymbol{\beta})\teh + (\boldsymbol{\alpha} - i \boldsymbol{\beta})\teh + \boldsymbol{\gamma} \teh)$ 
\end{enumerate}
\end{corollary}
\begin{proof}
We will only prove the statements involving $\boldsymbol{\alpha} \pm i \boldsymbol{\beta}$; the others are clear.

Suppose $(\boldsymbol{\alpha} + i \boldsymbol{\beta})\teh + (\boldsymbol{\alpha} - i \boldsymbol{\beta})\teh$ is tractable.
Then, it must be equal to either $\mathbf{u}\teh + \mathbf{v}\teh$ or $(\mathbf{u} + i \mathbf{v}) \teh + (\mathbf{u} - i \mathbf{v})\teh$ for 
$\mathbf{u}, \mathbf{v} \in \mathbb{R}^3$ such that $\langle \mathbf{u}, \mathbf{v} \rangle = 0$, with the additional condition that $\langle \mathbf{u}, \mathbf{u} \rangle = \langle \mathbf{v}, \mathbf{v} \rangle$ in the second case.
For the first case, by \cref{prop:uniqueness-tensor-decomposition}, we may assume that $\boldsymbol{\alpha} + i \boldsymbol{\beta} = c_1 \mathbf{u}$ and $\boldsymbol{\alpha} - i \boldsymbol{\beta} = c_2 \mathbf{v}$ for some $c_1, c_2 \in \mathbb{C}$.
Then, we must have $\langle \boldsymbol{\alpha} + i \boldsymbol{\beta}, \boldsymbol{\alpha} - i \boldsymbol{\beta} \rangle = \langle c_1 \mathbf{u}, c_2 \mathbf{v} \rangle = 0$. 
This implies $\langle \boldsymbol{\alpha}, \boldsymbol{\alpha} \rangle + \langle \boldsymbol{\beta}, \boldsymbol{\beta} \rangle = 0$ which cannot happen since $\boldsymbol{\alpha}, \boldsymbol{\beta}$ are nonzero real vectors.
For the second case, by \cref{prop:uniqueness-tensor-decomposition}, we may assume that $\boldsymbol{\alpha} + i \boldsymbol{\beta} = c_1 (\mathbf{u} + i \mathbf{v})$ and $\boldsymbol{\alpha} - i \boldsymbol{\beta} = c_2 (\mathbf{u} - i \mathbf{v})$ for some $c_1, c_2 \in \mathbb{C}$.
This implies $\langle \boldsymbol{\alpha} + i \boldsymbol{\beta} , \boldsymbol{\alpha} + i \boldsymbol{\beta}  \rangle = c_1^2 \langle \mathbf{u} + i \mathbf{v}, \mathbf{u} + i \mathbf{v} \rangle = c_1^2(\langle \mathbf{u}, \mathbf{u} \rangle - \langle \mathbf{v}, \mathbf{v} \rangle + 2 i \langle \mathbf{u}, \mathbf{v} \rangle) = 0$.
However, we also have that $\langle \boldsymbol{\alpha} + i \boldsymbol{\beta}, \boldsymbol{\alpha} + i \boldsymbol{\beta} \rangle = \langle \boldsymbol{\alpha}, \boldsymbol{\alpha} \rangle - \langle \boldsymbol{\beta}, \boldsymbol{\beta} \rangle + 2 i \langle \boldsymbol{\alpha}, \boldsymbol{\beta} \rangle $ which has a nonzero imaginary part by assumption.
\end{proof}

%% file: sections/theorem.tex
\section{Statement of the Dichotomy Theorem}
Let $\mathcal{F}$ be a set of nondegenerate, real-valued, symmetric signatures over domain $\{\db, \dg, \dr\}$.
\begin{restatable}{theorem}{maintheorem}\label{thm:dichotomy-set-of-domain-3}
$\holts(\mathcal{F})$ is computable in polynomial time if there exists a real orthogonal $T$, such that one of the following conditions holds.
  In all other cases, $\holts(\mathcal{F})$ is \#\P-hard.
  \begin{enumerate}
    \item[\tractbinary.] \label{tract-binary} Every signature in $\mathcal{F}$ has arity $\le 2$.
    \item[\tractE.] \label{tract-E} $T \mathcal{F} \subseteq \mathcal{E}$.
    \item[\tractBG.] \label{tract-BG} 
      \begin{enumerate}
        \item For all $\mathbf{F} \in T \mathcal{F}$ of arity $ \ge 3$, $\supp \mathbf{F} \subseteq \{\db, \dg\}^*$, and
        \item For all binary $\mathbf{G} \in T \mathcal{F}$, either $\mathbf{G} \in \mathcal{D}$ or $\mathbf{G}$ is $\db \dg | \dr$, and
        \item $\holbs((T\mathcal{F})\domres{\db, \dg})$ is tractable.
      \end{enumerate}
    \item[\tractBGR.] \label{tract-BG-R} 
      \begin{enumerate}
        \item For all $\mathbf{F} \in T \mathcal{F}$ of arity $ \ge 3$, $\mathbf{F}$ is $\db \dg | \dr$, and
        \item For all binary $\mathbf{G} \in T \mathcal{F}$, either $\mathbf{G}$ is $\db \dg | \dr$ or $\mathbf{G}$ is \swbg, and
        \item $\holbs((T\mathcal{F})\domres{\db, \dg})$ is tractable.
      \end{enumerate}
    \item[\tractBGGRBR.] \label{tract-BG-GR-BR} Let $\mathcal{F}_{ij} = \{\mathbf{F} \in T \mathcal{F} : \supp \mathbf{F} \subseteq \{i, j\}^*\}$.
      Let $\mathcal{R} = T \mathcal{F} - (\mathcal{F}_{\db \dg} \cup \mathcal{F}_{\db \dr} \cup \mathcal{F}_{\dg \dr})$. 
      \begin{enumerate}
        \item $\mathcal{R} \subseteq \mathbb{R}\octgroup$,
        and $\langle \mathcal{R}' \rangle \subseteq \octgroup$, where
        $\langle \mathcal{R}' \rangle$ is the monoid generated by $\mathcal{R}'
        =\mathbb{R} \mathcal{R} \cap \octgroup$, and
        \item 
        For all $i, j \in \{\db, \dg, \dr\}$, $\holbs( \langle \mathcal{R}' \rangle\domres{i, j} \cup \mathcal{F}_{ij})$ is tractable, and
        \item 
        For all $i, j \in \{ \db, \dg, \dr\}$, 
        $\holbs( (\bigcup_{\mathbf{G} \in \langle \mathcal{R}' \rangle} \mathbf{G} (T \mathcal{F})) \domres{i, j})$ is tractable.
      \end{enumerate}
  \end{enumerate}
\end{restatable}


In classes \tractBG, \tractBGR, and \tractBGGRBR, we refer to the $\holbs$ tractability.
By the Boolean domain dichotomy, it is necessary that those $\holbs$ problems are tractable.
However, it is not immediately clear why they are sufficient conditions.

In case (b) of class \tractBGGRBR, it seems like we may need to use the asymmetric $\holbs$ dichotomy  (which is known \cite{cai_dichotomy_2020_asymm}), because while the signatures in $\mathcal{R}$ are symmetric, $\langle \mathcal{R} \rangle$ may contain asymmetric signatures.
We claim that is not necessary.
Let $\mathbf{G} \in \octgroup$.
If $\mathbf{G}\domres{i, j}$ is nondegenerate and asymmetric, then $\mathbf{G}\domres{i, j} = \pm \begin{bsmallmatrix}
  0 & 1 \\
  -1 & 0
\end{bsmallmatrix}$.
We can deduce that $\mathbf{G}\domres{i, j}$ is \emph{universally compatible}, i.e., when appended to
any $\holbs$ tractable class results in a tractable set, without referring to any asymmetric $\holbs$ dichotomy.
For type $\typei(a, b)$ (see \cref{def:boolean-signature-type-i-ii,thm:dich-sym-Boolean}), we see that $\begin{bsmallmatrix}
  0 & 1 \\ -1 & 0
\end{bsmallmatrix} = \frac{1}{4a^2 + b^2} \begin{bsmallmatrix}
  b & - 2a \\
  - 2a & -b
\end{bsmallmatrix} \begin{bsmallmatrix}
  2a & b \\
  b & - 2a
\end{bsmallmatrix}$.
The first matrix in the symmetric signature notation is $[b, -2a, -b]$, which satisfies the recurrence $a \cdot b + b (-2a) = a  \cdot (-b)$.
The second matrix in the symmetric signature notation is $[2a, b, -2a]$, which is the other specified form.
For type $\typeii$, we see that $\begin{bsmallmatrix}
  0 & 1 \\
  -1 & 0
\end{bsmallmatrix} = \begin{bsmallmatrix}
  1 & 0 \\ 0 & -1
\end{bsmallmatrix} \begin{bsmallmatrix}
  0 & 1 \\ 1 & 0
\end{bsmallmatrix}$, and $[1, 0, -1]$  and $[0, 1, 0]$ are type $\typeii$.

This is not a coincidence.
We may view the $\holbs$ dichotomy in a geometric way.
Consider the tractable type $\typei(a, b)$.
There are two norm $1$ orthogonal vectors $\mathbf{u}, \mathbf{v} \in \mathbb{R}^2$ such that any signature of arity $\ge 3$ of  type $\typei(a, b)$
is of the form $c \mathbf{u}^{\otimes n} + d \mathbf{v}^{\otimes n}$.
Essentially, the type $\typei(a,b)$ signatures of arity $\ge 3$  
can be represented by two orthogonal lines in the $\mathbb{R}^2$ plane.
Consider the gadget construction of connecting a binary signature $\mathbf{G}$ to all of the edges of a signature.
In the tensor decomposition form, we have $\mathbf{G}^{\otimes n}(c \mathbf{u}^{\otimes n} + d \mathbf{v}^{\otimes n}) = c (\mathbf{G} \mathbf{u})^{\otimes n} + d (\mathbf{G} \mathbf{v})^{\otimes n}$.
For the binary signatures, we can check that the tractable signatures correspond to the linear transformations that fix the union of the two orthogonal lines as a set.
This is easily verified for the case of type $\typei(0, 1)$, when $\mathbf{u} = \mathbf{e}_1$ and $\mathbf{v} = \mathbf{e}_2$, since $[*, 0, *]$ signatures correspond to scaling $\mathbf{e}_1$ and $\mathbf{e}_2$, while $[0, *, 0]$ signatures correspond to reflection along $x = y$ line. 
Since type $\typei(a, b)$ for different $a, b$ are all related by a holographic transformation by a rotation, it is sufficient to check for type $\typei(0, 1)$.
The asymmetric signature above is a $\pi/2$ rotation, which fixes any pair of orthogonal lines, so it is tractable with all $\typei(a, b)$ types.
Another geometric view is that a $\pi/2$ rotation can be written as composition of a reflection along the line at the $\pi/4$ angle between $\mathbf{u}$ and $\mathbf{v}$ and the reflection along $\mathbf{u}$.
The two signatures that we used to construct the $\pi/2$ rotation above are precisely these geometric transformations.
\begin{figure}
  \centering
  \begin{tikzpicture}[scale=0.5]
    \draw[->, thin, dashed] (-5, 0) -- (5, 0) node[right]{$x$};
    \draw[->,  thin, dashed] (0, -5) -- (0, 5) node[above]{$y$};

    \draw[-] (-4, -1) -- (4, 1) node[right]{$\mathbf{u}$};
    \draw[-] (-1, 4) node[above]{$\mathbf{v}$} -- (1, -4) ;

    \node (b) at (7.5, 1) {$\mathbf{G}$};
    \node (a) at (7.5, 0) {$\leadsto$};

    \tikzset{shift={(14, 0)}}

    \draw[->,  thin, dashed] (-5, 0) -- (5, 0) node[right]{$x$};
    \draw[->,  thin, dashed] (0, -5) -- (0, 5) node[above]{$y$};

    \draw[-] (-3, -2) -- (3, 2) node[right]{$ \mathbf{G} \mathbf{u}$};
    \draw[-] (-2, 3) node[above]{$\mathbf{G} \mathbf{v}$} -- (2, -3) ;
  \end{tikzpicture}
  \caption{Visualization of type $\typei(a, b)$.} \label{fig:type-i-geometric}
\end{figure}

The tractable type $\typeii$ may be viewed as a circle in the $\mathbb{R}^2$ plane.
The justification is that by \cref{lem:Z-normalization}, \cref{cor:z-normalization-orthogonal}, and \cref{prop:typeii-same-norm}, any type $\typeii$ signature can be written as $(\mathbf{u} + i \mathbf{v})^{\otimes n} + (\mathbf{u} - i \mathbf{v})^{\otimes n}$ for two orthogonal $\mathbf{u}, \mathbf{v} \in \mathbb{R}^2$ of the same norm, and all such signatures can be constructed from a single type $\typeii$ signature.
So, in a sense, after normalization, type $\typeii$ signatures span the whole unit circle.
And the compatible binary signature are the linear transformations that transform a circle to a circle.
These signatures are precisely the scalar multiples of the elements in the orthogonal group $O(2)$.
We see that the signature $[x, y, -x]$ corresponds to a scalar multiple of a reflection matrix.
Also, since any rotation can be written as a product of two reflections, we have that the rotation matrices viewed as  signatures, possibly asymmetric, are also compatible with type $\typeii$ signatures.
\begin{figure}
  \centering
  \begin{tikzpicture}[scale=0.5]
    \draw[->, thin, dashed] (-5, 0) -- (5, 0) node[right]{$x$};
    \draw[->, thin, dashed] (0, -5) -- (0, 5) node[above]{$y$};

    \draw (0, 0) circle (4);
    \node (c) at (3.5,3.5) {$S^1$};

    \node (b) at (7.5, 1) {$\mathbf{G}$};
    \node (a) at (7.5, 0) {$\leadsto$};

    \tikzset{shift={(14, 0)}}

    \draw[->, thin, dashed] (-5, 0) -- (5, 0) node[right]{$x$};
    \draw[->, thin, dashed] (0, -5) -- (0, 5) node[above]{$y$};

    \draw (0, 0) ellipse (3.5 and 2.5);
    \node (c) at (3.5,2) {$\mathbf{G} S^1$};
  \end{tikzpicture}
  \caption{Visualization of type $\typeii$.} \label{fig:type-ii-geometric}
\end{figure}

A similar analogy can be made for domain $3$ signatures as well, since the tensor decomposition forms in \cref{thm:dich-single-sym-ter-dom3-real} are also about orthogonality of the vectors. 
Similar to the Boolean domain tractable signatures, a tractable domain $3$ signature also can be represented as a set of vectors in the three dimensional space.
For instance, let $\mathbf{u}_1\teh + \mathbf{u}_2\teh$, $\mathbf{v}_1\teh + \mathbf{v}_2\teh$ and $\mathbf{w}_1\teh + \mathbf{w}_2\teh$, 
be three signatures.
The idea is depicted in \cref{fig:geometric-intuition-3-signatures}, and this intuition is formalized in \cref{sec:set-of-signatures}.
Class \tractBG says that the signatures of arity $3$ or higher must live in 
some plane, and we can assume it is the $xy$-plane after an orthogonal transformation.
The compatible binary signatures are those fixing the $xy$-plane ($\db \dg | \dr$) or 
behaving like a degenerate transformation on the $xy$ plane.
The class \tractBGR says that the signatures of arity $3$ or higher are formed by $xy$-plane and a $z$-axis line.
The compatible binary signatures, after the same orthogonal transformation,
are those fixing the $xy$-plane and $z$-axis line ($\db \dg | \dr$) or mapping between the $xy$-plane and the $z$-axis (\swbg).
The class \tractBGGRBR says that the signatures of arity $3$ or higher must live in one of $xy$-plane, $yz$-plane, or $xz$-plane.
The compatible binary signatures are those permuting the three planes without stretching. 
Such transformations are precisely the group $O_h$.

This idea of viewing a binary signature as a transformation is also captured in the final proof in \cref{sec:set-of-signatures}, where we define a set $\mathcal{O}$ which is the set of all gadgets constructible by connecting a chain of binary signatures on a ternary signature. 
In terms of the vectors in the tensor decomposition, $\mathcal{O}$ can be viewed as the orbit under the monoid action of the binary signatures.
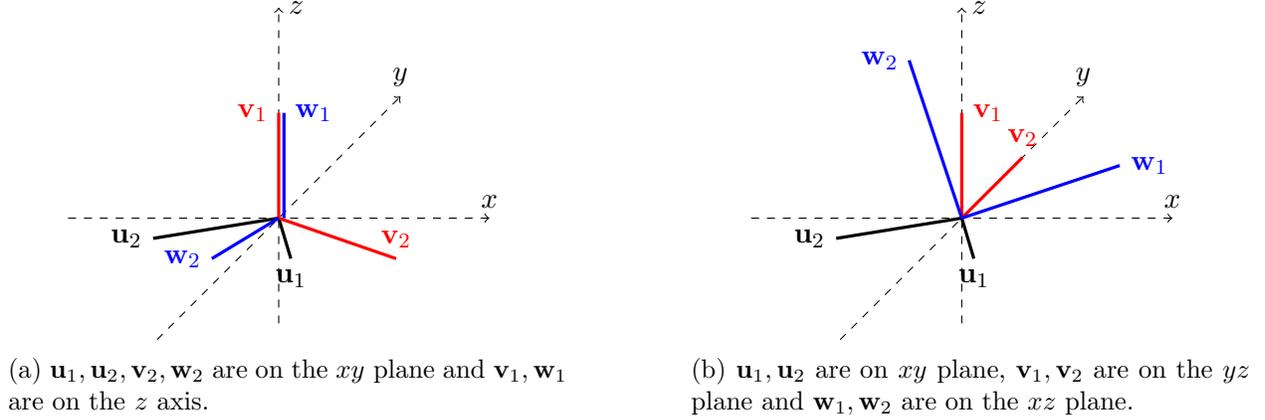
\begin{figure}
  \centering
  \begin{subfigure}[b]{0.45\textwidth}
    \centering
    \begin{tikzpicture}[scale=0.7]
      \draw[thin, dashed, ->] (xyz cs:x=-4) -- (xyz cs:x=4) node[above] {$x$};
      \draw[thin, dashed, ->] (xyz cs:y=-2) -- (xyz cs:y=4) node[right] {$z$};
      \draw[thin, dashed, ->] (xyz cs:z=6) -- (xyz cs:z=-6) node[above] {$y$};

      \draw[very thick] (xyz cs:x=0) -- (xyz cs:x=1,z=2) node[below] {$\mathbf{u}_1$};
      \draw[very thick] (xyz cs:z=0) -- (xyz cs:x=-2,z=1) node[left] {$\mathbf{u}_2$};

      \draw[very thick, red] (xyz cs:x=0) -- (xyz cs:y=2) node[left] {$\mathbf{v}_1$};
      \draw[very thick, red] (xyz cs:x=0) -- (xyz cs:x=3,z=2) node[above] {$\mathbf{v}_2$};

      \draw[very thick, blue] (xyz cs:x=0.1) -- (xyz cs:x=0.1,y=2) node[right] {$\mathbf{w}_1$};
      \draw[very thick, blue] (xyz cs:x=0) -- (xyz cs:x=-0.5,z=2) node[left] {$\mathbf{w}_2$};
    \end{tikzpicture}
    \caption{$\mathbf{u}_1, \mathbf{u}_2, \mathbf{v}_2, \mathbf{w}_2$ are on the $xy$ plane and $\mathbf{v}_1, \mathbf{w}_1$ are on the $z$ axis.}
  \end{subfigure}
  \hfill
  \begin{subfigure}[b]{0.45\textwidth}
    \centering
    \begin{tikzpicture}[scale=0.7]
      \draw[thin, dashed, ->] (xyz cs:x=-4) -- (xyz cs:x=4) node[above] {$x$};
      \draw[thin, dashed, ->] (xyz cs:y=-2) -- (xyz cs:y=4) node[right] {$z$};
      \draw[thin, dashed, ->] (xyz cs:z=6) -- (xyz cs:z=-6) node[above] {$y$};

      \draw[very thick] (xyz cs:x=0) -- (xyz cs:x=1,z=2) node[below] {$\mathbf{u}_1$};
      \draw[very thick] (xyz cs:z=0) -- (xyz cs:x=-2,z=1) node[left] {$\mathbf{u}_2$};

      \draw[very thick, red] (xyz cs:x=0) -- (xyz cs:y=2) node[right] {$\mathbf{v}_1$};
      \draw[very thick, red] (xyz cs:x=0) -- (xyz cs:z=-3) node[above] {$\mathbf{v}_2$};

      \draw[very thick, blue] (xyz cs:x=0) -- (xyz cs:x=3,y=1) node[right] {$\mathbf{w}_1$};
      \draw[very thick, blue] (xyz cs:x=0) -- (xyz cs:x=-1,y=3) node[left] {$\mathbf{w}_2$};
    \end{tikzpicture}
    \caption{$\mathbf{u}_1, \mathbf{u}_2$ are on $xy$ plane, $\mathbf{v}_1, \mathbf{v}_2$ are on the $yz$ plane and $\mathbf{w}_1, \mathbf{w}_2$ are on the $xz$ plane.}
  \end{subfigure}
  \caption{Geometric idea behind the dichotomy of $\holts(\mathcal{F})$}\label{fig:geometric-intuition-3-signatures}
\end{figure}

%% file: sections/tractability.tex
\section{Tractability}
In this section, we prove tractability.

Any $3 \times 3$ orthogonal matrix $T$ keeps the binary equality $(=_2)$ over $\{\db, \dg, \dr\}$ unchanged, namely $T\transpose I_3 T = I_3$ in matrix notation.
Hence $\holts(\mathcal{F})$ is tractable if and only if $\holts(T \mathcal{F})$ is tractable. 
Also, we may always assume that the given signature grid is connected,
as the Holant value of any signature grid is the product over connected components.

\subsection{Class \texorpdfstring{\tractbinary}{A}}
Class \tractbinary is when every signature in $\mathcal{F}$ has arity $\le 2$.
Then, the graph of the signature grid $\Omega$ is a disjoint union of paths and cycles. By matrix multiplication, we can compute the Holant value for a path. 
The Holant value for a cycle is obtained by taking the trace of a path.

\subsection{Class \texorpdfstring{\tractE}{B}}
Class \tractE is when there exists an orthogonal $T$ such that $T \mathcal{F} \subseteq \mathcal{E}$.
We prove the tractability of $\holts(\mathcal{E})$.
Note that $\mathcal{U} \subseteq \mathcal{E}$, so we argue the tractability of $\hol_3(\mathcal{E})$.
Let $\Omega$ be a signature grid and $e$ any edge.
For any $\db, \dg, \dr$ assignment to $e$, the assignment must propagate uniquely to all edges, which implies that there are at most three assignments to the whole grid that can result in a nonzero Holant value.
Therefore, $\hol_3(\mathcal{E})$ is polynomial time computable.

\subsection{Class \texorpdfstring{\tractBG}{C}}
Class \tractBG is when there exists an orthogonal $T$ such that $T \mathcal{F}$ has the following properties:
\begin{enumerate}
        \item For all $\mathbf{F} \in T \mathcal{F}$ of arity $3$ or higher, $\supp \mathbf{F} \subseteq \{\db, \dg\}^*$.
        \item For all binary $\mathbf{G} \in T \mathcal{F}$, either $\mathbf{G} \in \mathcal{D}$ or $\mathbf{G}$ is $\db \dg | \dr$.
        \item $\holbs((T\mathcal{F})\domres{\db, \dg})$ is tractable.
      \end{enumerate}
Let $\Omega$ be a connected signature grid 
over $T \mathcal{F}$. We may assume  $\Omega$
is not in Class \tractbinary.
If a unary signature is connected to another unary
signature then this is the entire  $\Omega$ 
and we are done. If a unary signature is connected to a binary
signature then they  become another unary constraint which we can compute
its signature. So, by induction, 
we can assume any unary remaining is connected to
some signature 
$\mathbf{F} \in T \mathcal{F}$ of arity $\ge 3$, which has 
$\supp \mathbf{F} \subseteq \{\db, \dg\}^*$, and thus we can replace the
unary restricted to $\{\db, \dg\}^*$.
If there is a  chain of binary signatures $\mathbf{G}_1, \mathbf{G}_2, \ldots, \mathbf{G}_k$, we replace  it by a single binary signature $\mathbf{G}$ by taking the matrix product. If all $\mathbf{G}_i$ are $\db \dg | \dr$, then $\mathbf{G}$ is also $\db \dg | \dr$. In addition, the matrix of $\mathbf{G}\domres{\db, \dg}$ is equal to the product of the matrices of $\mathbf{G}_i\domres{\db, \dg}$.
If there is some $i$ such that $\mathbf{G}_i \in \mathcal{D}$, then $\mathbf{G} \in \mathcal{D}$ as well, since $\mathcal{D}$ is closed under multiplication and also closed under left or right multiplication by a $\db \dg | \dr$ matrix.
Note that if $\mathbf{G} \in \mathcal{D}$, then $\mathbf{G} \domres{\db, \dg}$ is degenerate. Now since any binary  $\mathbf{G}$ remaining can only be connected to
signatures $\mathbf{F}$ of arity $\ge 3$ with 
$\supp \mathbf{F} \subseteq \{\db, \dg\}^*$,
we can replace  $\mathbf{G}$ by  $\mathbf{G} \domres{\db, \dg}$.
Thus we obtain an equivalent signature grid 
$\Omega'$ as a $\holbs$ instance on domain $\{\db, \dg\}$.
Then condition 3. implies that  the Holant value of $\Omega'$ can be computed in polynomial time.

\subsection{Class \texorpdfstring{\tractBGR}{D}}
Class \tractBGR is when there exists an orthogonal $T$ such that $T \mathcal{F}$ has the following properties:
\begin{enumerate}
        \item For all $\mathbf{F} \in T \mathcal{F}$ of arity $3$ or higher, $\mathbf{F}$ is $\db \dg | \dr$.
        \item For all binary $\mathbf{G} \in T \mathcal{F}$, either $\mathbf{G}$ is $\db \dg | \dr$ or $\mathbf{G}$ is \swbg.
        \item $\holbs((T\mathcal{F})\domres{\db, \dg})$ is tractable.
      \end{enumerate}
Suppose we are given a $T \mathcal{F}$ signature grid $\Omega$.
If $\Omega$ does not use any binary signature of \swbg, then all signatures are $\db \dg | \dr$. 
Then, if any edge gets assigned $\db$ or $\dg$, all other edges must also be assigned $\db$ or $\dg$ for the assignment to result in a nonzero value.
Similarly, any assignment of $\dr$ to an edge must propagate as $\dr$ to all other edges for the assignment to result in a nonzero value.
Therefore, the Holant value is sum of all assignments taking values in $\{\db, \dg\}$ and an assignment that only assigns $\dr$.
The first sum can be computed in polynomial since $\holbs( (T\mathcal{F})\domres{\db, \dg})$ is tractable, and the second value is computable in polynomial time.

Now, suppose $\Omega$ contains \swbg signatures. 
First, note that a product of two \swbg signatures is $\db \dg | \dr$ signature that is degenerate on $\{\db, \dg\}$.
\[
  \begin{bmatrix}
    0 & 0 & a \\
    0 & 0 & b \\
    c & d & 0
    \end{bmatrix} \begin{bmatrix}
    0 & 0 & e \\
    0 & 0 & f \\
    g & h & 0
  \end{bmatrix} = \begin{bmatrix}
    ag & ah & 0 \\
    bg & bh & 0 \\
    0 & 0 & ce + df
  \end{bmatrix}
  \, .
\]
Hence, we may reduce any chain $\mathbf{G}_1, \mathbf{G}_2, \cdots, \mathbf{G}_k$ of \swbg signatures by matrix multiplication to a length $2$ chain $\mathbf{G}, \mathbf{G}_k$ if $k$ is odd and a single $\mathbf{G}$ if $k$ is even,
where $\mathbf{G}$ is a $\db \dg | \dr$ signature that is degenerate on $\{\db, \dg\}$.
In particular, $\mathbf{G}\domres{\db, \dg}$ is compatible with $\mathcal{F}\domres{\db, \dg}$.
Therefore, we may assume that we have a signature grid $\Omega'$ such that any \swbg signature is connected to two $\db \dg | \dr$ signatures.
Now, suppose we gather each connected component of $\db \dg | \dr$ signatures as a cluster, so that the connections between clusters are by \swbg signatures.
If we imagine a graph with vertices being the clusters and edges being the \swbg signatures, then this graph must be bipartite.
Otherwise, suppose there is an odd cycle of clusters $C_1, C_2, \cdots, C_k, C_1$.
In any nonzero assignment, each cluster can take only values from $\{\db, \dg\}$ or $\dr$.
If $C_1$ gets $\{\db, \dg\}$, then $C_2$ must have $\dr$ because of the \swbg connection.
Therefore, if there is an odd cycle, $C_1$ will have an incoming $\dr$, resulting in a zero evaluation.
Similar argument shows that an assignment of $\dr$ to $C_1$ evaluates to zero.
Also, by the same argument, there cannot be a self loop by \swbg signature within a single cluster. 

Let the bipartition be $L \sqcup R$.
There are only two types of nonzero assignments: 
(1) all clusters in $L$ get $\{\db, \dg\}$ and all clusters in $R$ get $\dr$;
(2) all clusters in $L$ get $\dr$ and all clusters in $R$ get $\{\db, \dg\}$.
For both types of assignments, the \swbg signatures connecting a cluster $C_L \in L$ and $C_R \in R$ factors into two degenerate signatures.
In particular, suppose $\mathbf{G} = \begin{bsmallmatrix}
  0 & 0 & a \\
  0 & 0 & b \\
  c & d & 0
\end{bsmallmatrix}$ is a \swbg signature connecting $C_L$ and $C_R$.
Let $\Omega'$ be obtained from $\Omega$ by removing $\mathbf{G}$ and connecting the unary $(a, b, 0)$ to $C_L$ and $(0, 0, 1)$ to $C_R$, as described in \cref{fig:swap-factor}.
Then, since $C_L$ only takes values in $\{\db, \dg\}$ and $C_R$ only takes values in $\dr$, 
the sum of type (1) assignments is the same as the Holant value of $\Omega'$.
\begin{figure}
  \centering
  \begin{tikzpicture}
    \node[mycirc, label=above:{$C_L$}] (a) at (0, 0) {};
    \node[mycirc, label=above:{$\mathbf{G}$}] (b) at (1, 0) {};
    \node[mycirc, label=above:{$C_R$}] (c) at (2, 0) {};

    \draw (-1, 0) -- (a) -- (b) -- (c) -- (3, 0);
    \draw (-1, 0.5) -- (a) -- (-1, -0.5);
    \draw (3, 0.5) -- (c) -- (3, -0.5);

    \node (l) at (4, 0) {$\Rightarrow$};

    \node[mycirc, label=above:{$C_L$}] (f) at (6, 0) {};
    \node[mycirc, label=above:{$(a, b, 0)$}] (g) at (7.5, 0) {};
    \node[mycirc, label=above:{$(0, 0, 1)$}] (h) at (9.5, 0) {};
    \node[mycirc, label=above:{$C_R$}] (i) at (11, 0) {};

    \draw (5, 0) -- (f) -- (g);
    \draw (h) -- (i) -- (12, 0);

    \draw(5, 0.5) -- (f);
    \draw(5, -0.5) -- (f);

    \draw(12, 0.5) -- (i);
    \draw(12, -0.5) -- (i);
  \end{tikzpicture}
  \caption{Factorization of \swbg in type (1) assignment}\label{fig:swap-factor}
\end{figure}
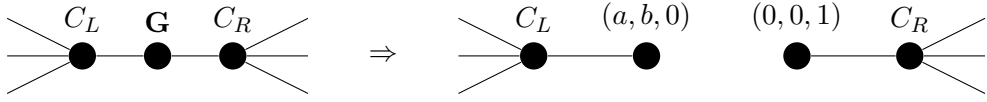
This equivalence can be shown formally by a calculation. 
Let $e_L$ denote the edge between $C_L$ and $\mathbf{G}$ and $e_R$ denote the edge between $\mathbf{G}$ and $C_R$.
Let $c_{ij}$ be the contribution from rest of the grid when $e_L = i$ and $e_R = j$.
In a type (1) assignment, $e_L$ can only take value $\{\db, \dg\}$ and $e_R$ can only take value $\dr$.
The sum of all type (1) assignments is then:
\[
  \mathbf{G}(\db, \dr) \cdot c_{\db \dg} + \mathbf{G}(\dg, \dr) \cdot c_{\dg \dr} = a \cdot c_{\db \dg} + b \cdot c_{\dg \dr} 
  \, .
\]
The Holant value of $\Omega'$ is: 
\[
  a \cdot 1 \cdot c_{\db \dg} + b \cdot 1 \cdot c_{\dg \dr}
  \, .
\]
Note that we may use the same $c_{ij}$ since the two new unary signatures force the edge assignment in $L$ to be from $\{\db, \dg\}$ and the edge assignment in $R$ to be $\dr$.
Therefore, to evaluate the sum of type (1) assignments, we may factor all the \swbg signatures in this way, evaluate in $\{\db, \dg\}$ domain $\holbs$ on $L$ and assign $\dr$ to all of $R$, and multiply the two resulting values.
Both can be done in polynomial time since $\mathcal{F}\domres{\db, \dg}$ is assumed to be tractable.
Similar argument shows that the sum of type (2) assignments can be also computed in polynomial time.
Then, the Holant value is just the sum of those two values. 

\subsection{Class \texorpdfstring{\tractBGGRBR}{E}}
Class \tractBGGRBR is when there exists an orthogonal $T$ such that $T \mathcal{F}$ has the following properties.
Let $\mathcal{F}_{ij} = \{\mathbf{F} \in T \mathcal{F} : \supp \mathbf{F} \subseteq \{i, j\}^*\}$.
Let $\mathcal{R} = T \mathcal{F} - (\mathcal{F}_{\db \dg} \cup \mathcal{F}_{\db \dr} \cup \mathcal{F}_{\dg \dr})$.
Let 
$\mathcal{R}' = \{ \lambda \mathbf{G} : \lambda \in \mathbb{R}, \mathbf{G} \in \mathcal{R} \text{ such that } \lambda \mathbf{G} \in \octgroup \}$.
\begin{enumerate}
        \item For any $\mathbf{G} \in \mathcal{R}$, there is some $\lambda \in \mathbb{R}$ such that $\lambda \mathbf{G} \in \octgroup$.
        \item For all $i, j \in \{\db, \dg, \dr\}$, $\holbs( \langle \mathcal{R}' \rangle\domres{i, j} \cup \mathcal{F}_{ij})$ is tractable.
        \item For all $i, j \in \{ \db, \dg, \dr\}$, 
        $\holbs( (\bigcup_{\mathbf{G} \in \langle \mathcal{R}' \rangle} \mathbf{G} (T \mathcal{F})) \domres{i, j})$ is tractable.
      \end{enumerate}
Suppose we are given a $T \mathcal{F}$ signature grid $\Omega$.
If $\Omega$ does not use any signature from $\mathcal{R}$, then all signatures in $\Omega$ are \strspt.
Then, whenever there is an edge between two signatures with different support, the edge factors as pinning.
For example, if there is an edge $e$ between a signature $\mathbf{F}_{\db \dg} \in \mathcal{F}_{\db \dg}$ and $\mathbf{F}_{\dg \dr} \in \mathcal{F}_{\dg \dr}$,
we may remove the edge and connect the unary $(0, 1, 0)$ to both $\mathbf{F}_{\db \dg}$ and $\mathbf{F}_{\dg \dr}$, as shown \cref{fig:support-factor}.
This is because $e$ can only take value $\dg$ in any nonzero assignment.
Also, if there is a connection between any unary signature $\mathbf{u}$ and a \strspt signature supported on $\{i, j\}$, then we may replace $\mathbf{u}$ with $\mathbf{u}\domres{ij}$.
We obtain a new signature grid $\Omega'$ after factoring all such edges, and each of the connected components of $\Omega'$ is an instance of $\holbs(\mathcal{F}_{ij})$ for some $i, j \in \{\db, \dg, \dr\}$.
By assumption, the Holant value of each of the connected components can be computed in polynomial time.
\begin{figure}
  \centering
  \begin{tikzpicture}
    \node[mycirc, label=above:{$\mathbf{F}_{\db\dg}$}] (a) at (0, 0) {};
    \node[mycirc, label=above:{$\mathbf{F}_{\dg \dr}$}] (c) at (2, 0) {};

    \draw (-1, 0) -- (a) -- (c) -- (3, 0);
    \draw (-1, 0.5) -- (a) -- (-1, -0.5);
    \draw (3, 0.5) -- (c) -- (3, -0.5);

    \node (l) at (4, 0) {$\Rightarrow$};

    \node[mycirc, label=above:{$\mathbf{F}_{\db\dg}$}] (f) at (6, 0) {};
    \node[mycirc, label=above:{$(0, 1, 0)$}] (g) at (7.5, 0) {};
    \node[mycirc, label=above:{$(0, 1, 0)$}] (h) at (9.5, 0) {};
    \node[mycirc, label=above:{$\mathbf{F}_{\dg\dr}$}] (i) at (11, 0) {};

    \draw (5, 0) -- (f) -- (g);
    \draw (h) -- (i) -- (12, 0);

    \draw (5, 0.5) -- (f) -- (5, -0.5);
    \draw (12, 0.5) -- (i) -- (12, -0.5);
  \end{tikzpicture}
  \caption{Factorization of an edge between \strspt signatures of different support}\label{fig:support-factor}
\end{figure}
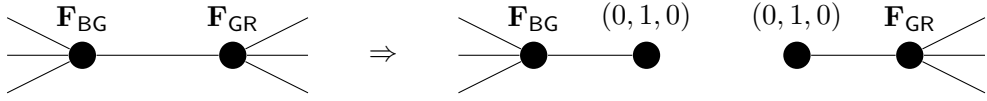

Now, suppose $\Omega$ contains signatures from $\mathcal{R}$.
We may replace those signatures by the corresponding scalar multiple in $\mathcal{R}'$.
Recall that the signatures in $\mathcal{R}'$ are not \strspt, but they are binary signatures of signed swap matrices by the assumption that $\mathcal{R}' \subseteq \octgroup$.
By the above, we may assume that in $\Omega$, there is no edge between \strspt signatures of different supports.
First, we reduce any chain of $\mathcal{R}'$ signatures into a single binary signature in $\langle \mathcal{R}' \rangle \subseteq \octgroup$.
Reducing any connection of $\mathcal{R}'$ signature with a unary signature to a unary, we obtain a new grid $\Omega'$ in which any $\langle \mathcal{R}' \rangle$ signature is between two signatures from $\mathcal{F}_{ij}$ and $\mathcal{F}_{i'j'}$.
We imagine $\Omega'$ is composed of clusters $C$ where each cluster is a connected component of $\mathcal{F}_{ij}$ signatures for some $i, j \in \{\db, \dg, \dr\}$, and each cluster has outgoing edges of $\langle \mathcal{R}' \rangle$ signatures.

Suppose $\mathbf{G} \in \langle \mathcal{R}' \rangle$ is a self loop on a cluster $C$.
Since its two end points are from $C$ and all signatures in $C$ are \strspt on the same support $\{i, j\}$, we may replace $\mathbf{G}$ with $\mathbf{G}\domres{i, j}$.
By assumption, $\mathbf{G}\domres{i, j}$ is compatible with signatures in $C$, so we may absorb it into $C$.

Suppose $\mathbf{G} \in \langle \mathcal{R}' \rangle$ connects two clusters $C_1$ and $C_2$ with supports $\{i_1, j_1\}$ and $\{i_2, j_2\}$ respectively.
Suppose $\mathbf{F}_1$ and $\mathbf{F}_2$ are the signatures in $C_1$ and $C_2$ connected by $\mathbf{G}$ respectively.
Let the arity of $\mathbf{F}_1$ be $n$.
Then, we replace all the edges of $\mathbf{F}_1$, except the one connecting to $\mathbf{G}$, with $\mathbf{G}^{-1} \mathbf{G} $.
Essentially, we are performing a local holographic transformation, so this does not change the Holant value.
The process is described in \cref{fig:local-holographic} for arity $3$ case.
We now have $\mathbf{G}^{\otimes n} \mathbf{F}_1$ connected directly to $\mathbf{F}_2$.
If $\lvert \supp \mathbf{G}^{\otimes n} \mathbf{F}_1 \cap \supp \mathbf{F}_2 \rvert \le 1$, the edge factors in to pinning.
Otherwise, by the assumption, $\mathbf{G}^{\otimes n} \mathbf{F}_1$ is compatible with $\mathbf{F}_2$, so we may absorb it into the $C_2$ cluster.

We choose a cluster to begin with, and repeatedly absorb its neighboring signatures or factor the edge into pinning by the above process.
Each step of the above process can be done in polynomial time.
Then, in the end, we will be left with multiple connected components in which all the signatures are \strspt, and their $\holbs$ values can be computed in polynomial time by the assumption.

\begin{figure}
\centering
      \begin{tikzpicture}[scale=0.55]
        \node[mycirc, label=right:{$\mathbf{F}_1$}] (f1) at (3, 0) {};
        \node[mycirc, label=above:{$\mathbf{G}$}] (g) at (1, 0) {};
        \node[mycirc, label=above:{$\mathbf{F}_2$}] (f2) at (-1, 0) {};

        \node (f3) at (-2, 0) {};

        \node[label=above:{}] (g2) at (1, 2) {};

        \draw (f1) to[out=90, in=0] (1, 2) -- (0, 2);
        \draw (f1) to[out=-90, in=0] (1, -2) -- (0, -2);

        \draw (f1) -- (g) -- (f2) -- (f3);

        \node[mycirc, label=right:{$\mathbf{F}_1$}] (f1) at (11, 0) {};
        \node[mycirc, label=above:{$\mathbf{G}$}] (g) at (9, 0) {};
        \node[mycirc, label=above:{$\mathbf{F}_2$}] (f2) at (7, 0) {};

        \node (f3) at (6, 0) {};

        \node[mycirc, label=above:{$\mathbf{G}$}] (g1) at (9, 2) {};
        \node[mycirc, label=above:{$\mathbf{G}$}] (g2) at (9, -2) {};
        \node[mycirc, label=above:{$\mathbf{G}^{-1}$}] (g11) at (8, 2) {};
        \node[mycirc, label=above:{$\mathbf{G}^{-1}$}] (g21) at (8, -2) {};

        \draw (f1) to[out=90, in=0] (g1) -- (g11) -- (7, 2);
        \draw (f1) to[out=-90, in=0] (g2) -- (g21) -- (7, -2);

        \draw (f1) -- (g) -- (f2) -- (f3);

        \node[mycirc, label=right:{$\mathbf{G}^{\otimes 3} \mathbf{F}_1$}] (f1) at (19, 0) {};
        \node[mycirc, label=above:{$\mathbf{F}_2$}] (f2) at (15, 0) {};

        \node (f3) at (14, 0) {};

        \node[mycirc, label=above:{$\mathbf{G}^{-1}$}] (g11) at (17, 2) {};
        \node[mycirc, label=above:{$\mathbf{G}^{-1}$}] (g21) at (17, -2) {};

        \draw (f1) to[out=90, in=0] (g11) -- (16, 2);
        \draw (f1) to[out=-90, in=0] (g21) -- (16, -2);

        \draw (f1) -- (f2) -- (f3);
  \end{tikzpicture}
  \caption{Local holographic transformation}\label{fig:local-holographic}
\end{figure}
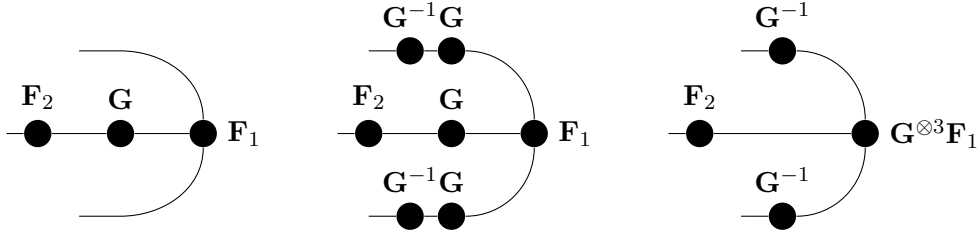

%% file: sections/hardness.tex
\section{Outline of Hardness}
For the rest of the paper, we prove the hardness results.
After a dichotomy of a single ternary signature~\cite{cai_dichotomy_2013}, a natural next step is proving a dichotomy of a pair of ternary and binary signatures (as binary signatures are the `simplest' signatures after unary signatures), and use it to prove further theorems.
However, for domain size 3 in the Holant setting, 
binary signatures actually allow nontrivial, and somewhat unanticipated, interactions with other signatures. 
Also, it turns out that a dichotomy for a pair of binary and ternary signatures,
while certainly needed on its own, is not easily applicable for showing further dichotomies. 
 We circumvent this difficulty by proving a dichotomy of a pair of ternary signatures directly.

In \cref{sec:single-ternary-signle-binary}, we show the dichotomy of $\holts(\mathbf{F}, \mathbf{G})$, in which $\mathbf{F}$ is a ternary signature and $\mathbf{G}$ is a binary signature.
In \cref{sec:two-ternary}, we show the dichotomy of $\holts(\mathbf{F}, \mathbf{G})$, in which $\mathbf{F}$ and $\mathbf{G}$ are both ternary signatures.
In \cref{sec:higher-arity}, we show the dichotomy of $\holts(\mathbf{F})$ for a signature $\mathbf{F}$ of arbitrary arity $\ge 3$.
In \cref{sec:set-of-signatures}, we show the dichotomy of $\holts(\mathcal{F})$ for an arbitrary set of signatures $\mathcal{F}$. 
A logical dependency is described by the diagram in 
\cref{fig:logical-dependency}.
In each case, we only consider signatures of rank $\ge 2$.
This is because rank $1$ signatures are degenerate and thus can be factored into unary signatures. 
Since $\hols$ assumes the existence of unary signatures, degenerate signatures do not impact the complexity.

\begin{figure}
\centering
    \begin{tikzpicture}[scale=0.7]
        \node (ar33) [block] at (0, 0) {$(3, 3)$};

        \node (ar44) [block, right of=ar33, xshift=1cm] {$(4, 4)$};
                \node (ar4) [block, above of=ar44, yshift=1cm] {$(4)$};

        \node (ar55) [block, right of=ar44, xshift=1cm] {$(5, 5)$};
                \node (ar5) [block, above of=ar55, yshift=1cm] {$(5)$};

        \node (dots) [right of=ar5, yshift=-1cm, xshift=0cm] {$\cdots$};

        \node (arn) [block, right of=dots] {$(n)$};

        \node (set) [block] at (12, 0) {Set};

        \draw[dashed, rounded corners] (1.5, -1.2) rectangle (10, 4);

        \node (ar23) [block] at (0, -2) {$(2, 3)$};

        \draw [arrow] (ar33) -- (ar4);
        \draw [arrow] (ar4) -- (ar44);
        \draw [arrow] (ar33) -- (ar44);

        \draw [arrow] (ar44) -- (ar5);
        \draw [arrow] (ar5) -- (ar55);
        \draw [arrow] (ar44) -- (ar55);

        \draw [arrow, rounded corners] (ar23) -- (12, -2) --  (set);
        \draw [arrow, rounded corners] (10, 1.4) -- (12, 1.4) -- (set);
    \end{tikzpicture}
    \caption{
Logical dependency diagram: $(n)$ refers to a dichotomy of $\holts(\mathbf{F})$ for an arity-$n$ signature $\mathbf{F}$.
    $(n, m)$ refers to a dichotomy of $\holts(\mathbf{F}, \mathbf{G})$ of an arity-$n$ signature $\mathbf{F}$ and an arity-$m$ signature $\mathbf{G}$.
    `Set' refers to the dichotomy of an arbitrary set of signatures.
    }\label{fig:logical-dependency}
\end{figure}
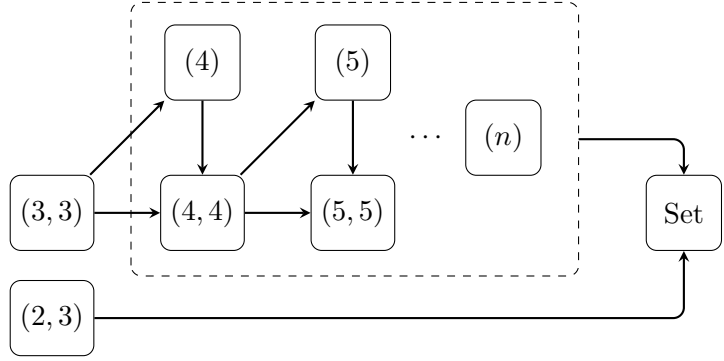

In the proof of the dichotomy theorems, 
by applying the known dichotomy theorems, we argue that to escape \#\P-hardness 
 the signatures must be of the tractable forms of \cref{thm:dichotomy-set-of-domain-3}. 
Therefore, the signatures are either tractable or fail to escape hardness, meaning they are \#\P-hard.
Hence, we obtain a dichotomy.

The main intuition behind the proof of hardness is the geometry of the vectors in a tensor decomposition of a signature.
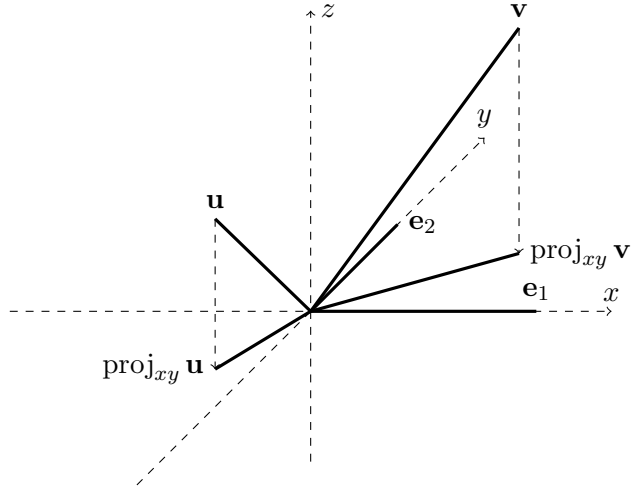
\begin{figure}
  \centering
  \begin{tikzpicture}[scale=1]
\draw[thin, dashed, ->] (xyz cs:x=-4) -- (xyz cs:x=4) node[above] {$x$};
\draw[thin, dashed, ->] (xyz cs:y=-2) -- (xyz cs:y=4) node[right] {$z$};
\draw[thin, dashed, ->] (xyz cs:z=6) -- (xyz cs:z=-6) node[above] {$y$};
\draw[very thick] (xyz cs:x=0) -- (xyz cs:x=3) node[above] {$\mathbf{e}_1$};
\draw[very thick] (xyz cs:z=0) -- (xyz cs:z=-3) node[right] {$\mathbf{e}_2$};

\draw[very thick] (xyz cs:x=0) -- (xyz cs:x=2,y=3,z=-2) node[above] {$\mathbf{v}$};
\draw[very thick] (xyz cs:x=0) -- (xyz cs:x=-0.5,y=2,z=2) node[above] {$\mathbf{u}$};

\draw[very thick] (xyz cs:x=0) -- (xyz cs:x=2,y=0,z=-2) node[right] {$\proj_{xy} \mathbf{v}$};
\draw[very thick] (xyz cs:x=0) -- (xyz cs:x=-0.5,y=0,z=2) node[left] {$\proj_{xy} \mathbf{u}$};

\draw[thin, dashed, ->]  (xyz cs:x=2,y=3,z=-2) -- (xyz cs:x=2,y=0,z=-2);
\draw[thin, dashed, ->] (xyz cs:x=-0.5,y=2,z=2) -- (xyz cs:x=-0.5,y=0,z=2);
\end{tikzpicture}
\caption{Geometric intuition of domain restriction}\label{fig:geometric-intuition}
\end{figure}
We always start with a canonical form of a signature (after a suitable orthogonal
transformation), for example, $\mathbf{F} = \mathbf{e}_1\teh + \mathbf{e}_2\teh$.
By \cref{lem:geneq-eqbg} and \cref{lem:domain-restriction}, $\mathbf{F}$ can realize $(=_{\db \dg})$ and take domain restriction to $\{\db, \dg\}$.
If the other signature is $\mathbf{G} = \mathbf{u}\teh + \mathbf{v}\teh$ for orthogonal $\mathbf{u},\mathbf{v} \in \mathbb{R}^3$,
the domain restriction is $\mathbf{G}\domres{\db, \dg} = (\proj_{xy} \mathbf{u})\teh + (\proj_{xy} \mathbf{v})\teh$.
Essentially, most of the arguments boil down to showing that the angle between $\proj_{xy} \mathbf{u}$ and $\proj_{xy} \mathbf{v}$ must be $\pi/2$ (with one exception), otherwise, 
$\mathbf{G}\domres{\db, \dg}$ is \sph by the Boolean domain dichotomy.
The other tractable possibility is when $\proj_{xy} \mathbf{u} \sim \proj_{xy} \mathbf{v}$, which can only happen if the plane that contains $\mathbf{u}$ and $\mathbf{v}$ contains the $z$-axis, which corresponds to the case when $\mathbf{G}\domres{\db, \dg}$ is degenerate.

%% file: sections/one_ternary_one_binary.tex
\section{A Pair of Ternary and Binary Signatures}\label{sec:single-ternary-signle-binary}
We show the dichotomy of $\holts(\mathbf{F}, \mathbf{G})$ for a ternary signature $\mathbf{F}$ and a binary signature $\mathbf{G}$.
We may assume that $\mathbf{F}$ is of the tractable form in \cref{thm:dich-single-sym-ter-dom3-real}, and we separate into four cases, in which $\mathbf{F}$ is: type \ternarytractgeneq and rank $2$; type \ternarytractz and rank $2$; type \ternarytractgeneq and rank $3$; type \ternarytractz and rank $3$.
Note that we do not need to consider a rank $1$ signature since such signature is degenerate.
In addition, we may assume that $\mathbf{F}$ is in the canonical form of $a \mathbf{e}_1 \teh + b \mathbf{e}_2 \teh + c \mathbf{e}_3 \teh$ or $\boldsymbol{\beta}_0\teh + \overline{\boldsymbol{\beta}_0}\teh + \lambda \mathbf{e}_3 \teh$ for some $a, b, c, \lambda \in \mathbb{R}$.

We prove some simple propositions about normalizing the individual scalar constants of the tensor decomposition and realizing similar signatures. 
This will allow us to assume that a type \ternarytractgeneq signature has the form $a \mathbf{e}_1\teh + b \mathbf{e}_2 \teh + c \mathbf{e}_3\teh$ for $a, b, c \in \{0, 1\}$ and a type \ternarytractz signature has the form $\boldsymbol{\beta}_0\teh + \overline{\boldsymbol{\beta}_0}\teh + \lambda \mathbf{e}_3 \teh$ for $\lambda \in \{0, 1\}$.
\begin{proposition}\label{lem:geneq-normalization}
  Let $\mathbf{F} = \mathbf{v}_1\teh + \mathbf{v}_2 \teh + \mathbf{v}_3 \teh$ such that $\mathbf{v}_1, \mathbf{v}_2, \mathbf{v}_3 \in \mathbb{R}^3$ are pairwise orthogonal.
  Let $\mathcal{F}$ be any set of signatures.
  Let $a, b, c \in \mathbb{R}$ be arbitrary.
  If $\mathbf{v}_1, \mathbf{v}_2, \mathbf{v}_3 \ne 0$, then 
  \[
    \holts(\mathcal{F} \cup \{a \mathbf{v}_1\teh + b\mathbf{v}_2 \teh + c \mathbf{v}_3 \teh\}) \le_T \holts(\mathcal{F} \cup \{ \mathbf{F} \}) \, .
  \]

  If $\mathbf{v}_1, \mathbf{v}_2 \ne 0$ and $\mathbf{v}_3 = 0$, then for any $a, b \in \mathbb{R}$,
  \[
    \holts(\mathcal{F} \cup \{a \mathbf{v}_1\teh + b \mathbf{v}_2 \teh\}) \le_T \holts(\mathcal{F}  \cup \{ \mathbf{F} \}) \, .
  \]
\end{proposition}
\begin{proof}
  Suppose $\mathbf{v}_i \ne 0$ for all $i$.
  Let $T$ be the orthogonal matrix with $\mathbf{u}_i$ as the $i$th row.
  Then, $T \teh \mathbf{F} = c_1 \mathbf{e}_1 \teh + c_2 \mathbf{e}_2 \teh + c_3 \mathbf{e}_3 \teh$,
  where $c_i = \|\mathbf{v}_i\|^3$.
  Consider the gadget in \cref{fig:geneq-normalization}.
  \begin{figure}
    \centering
    \begin{tikzpicture}
      \node[mycirc, label=left:{$T \teh \mathbf{F}$}] (f1) at (0, 0) {};
      \node[mycirc, label=right:{$T \teh \mathbf{F}$}] (f2) at (2, 0) {};
      \node[mycirc, label=above:{$(a/c_1^2, b/c_2^2, c/c_3^2)$}] (u) at (3, 1) {};

      \draw (-1, 1) -- (f1);
      \draw (-1, -1) -- (f1);
      \draw (f1) -- (f2);
      \draw (f2) -- (u);
      \draw (f2) -- (3, -1);
    \end{tikzpicture}
    \caption{Realization of $a \mathbf{e}_1\teh + b \mathbf{e}_2\teh + c \mathbf{e}_3\teh$} \label{fig:geneq-normalization}
  \end{figure}
  We can see that it is the signature $a\mathbf{e}_1\teh + b\mathbf{e}_2 \teh + c\mathbf{e}_3 \teh$.
  So, $\holts(T \mathcal{F} \cup \{a\mathbf{e_1}\teh + b\mathbf{e_2}\teh + c\mathbf{e_3}\teh\}) \le_T \holts(T \mathcal{F} \cup \{T \teh \mathbf{F}\})$.
  Applying $T^{-1}$ gives the result.

  If $\mathbf{v}_3 = 0$, then we may apply the same construction by finding an orthonormal basis containing $\mathbf{u}_1, \mathbf{u}_2$, and using the unary $(a/c_1^2, b/c_2^2, 0)$.
\end{proof}

\begin{proposition}\label{lem:Z-normalization}
  Let $\mathbf{F} = \boldsymbol{\beta} \teh + \overline{\boldsymbol{\beta}} \teh + \lambda \mathbf{e}_3 \teh$ for $\lambda \in \mathbb{R}$ where $\boldsymbol{\beta} = \frac{1}{\sqrt{2}} (1, i, 0)$.
  Let $\mathcal{F}$ be any set of signatures.
  If $\lambda \ne 0$, then
  \[
  \holts(\mathcal{F} \cup \{\mathbf{G}\}) \le_T
    \holts(\mathcal{F} \cup \{\mathbf{F}\})
  \]
  for any real valued symmetric ternary $\db \dg | \dr$ signature $\mathbf{G}$ such that 
   $\mathbf{G}\domres{\db, \dg}$ is of type $\typeii$,
   i.e., $\mathbf{G}\domres{\db, \dg} = [x, y, -x, -y]$ for some $x, y \in \mathbb{R}$.

   If $\lambda = 0$, then 
   \[
  \holts(\mathcal{F} \cup \{\mathbf{G}\}) \le_T
    \holts(\mathcal{F} \cup \{\mathbf{F}\})
  \]
  for any real valued symmetric ternary signature $\mathbf{G}$ such that $\supp \mathbf{G} \subseteq \{\db, \dg \}^*$ and 
   $\mathbf{G}\domres{\db, \dg}$ is of type $\typeii$,
   i.e., $\mathbf{G}\domres{\db, \dg} = [x, y, -x, -y]$ for some $x, y \in \mathbb{R}$.
\end{proposition}
\begin{proof}
  Consider the gadget in \cref{fig:Z-normalization}
  \begin{figure}
    \centering
    \begin{subfigure}[b]{0.4\textwidth}
    \centering
    \begin{tikzpicture}[scale=0.8]
      \node[mycirc, label=left:{$\mathbf{F}$}] (f1) at (0, 0) {};
      \node[mycirc, label=below:{$\mathbf{F}$}] (f2) at (2, 0) {};
      \node[mycirc, label=right:{$\mathbf{F}$}] (f3) at (4, 0) {};

      \node[mycirc, label=above:{$\mathbf{u}$}] (u) at (2, 2) {};

      \draw (-1, 1) -- (f1);
      \draw (-1, -1) -- (f1);
      \draw (f1) -- (f2) -- (f3);
      \draw (f3) -- (5, 1);
      \draw (f3) -- (5, -1);
      \draw (f2) -- (u);
    \end{tikzpicture}
    \caption{The gadget}
  \end{subfigure}
  \hfill
    \begin{subfigure}[b]{0.55\textwidth}
    \centering
      \begin{tikzpicture}[scale=0.8]
      \node[mycirc, label=left:{$M \teh \mathbf{F}$}] (f1) at (6.5, 0) {};
      \node[mycirc, label=below:{$M \teh \mathbf{F}$}] (f2) at (9, 0) {};
      \node[mycirc, label=right:{$M \teh \mathbf{F}$}] (f3) at (11.5, 0) {};

      \node[mycirc, label=above:{$M \mathbf{u}$}] (u) at (9, 2) {};

      \node[mysq, label=below:{$(\ne_{\db\dg;\dr})$}] (s1) at (7.5, 0) {};
      \node[mysq, label=below:{$(\ne_{\db\dg;\dr})$}] (s2) at (10.5, 0) {};
      \node[mysq, label=right:{$(\ne_{\db\dg;\dr})$}] (s3) at (9, 1) {};

      \draw (5.5, 1) -- (f1);
      \draw (5.5, -1) -- (f1);
      \draw (f1) -- (f2) -- (f3);
      \draw (f3) -- (12.5, 1);
      \draw (f3) -- (12.5, -1);
      \draw (f2) -- (u);
    \end{tikzpicture}
    \caption{After $M$ transformation}
  \end{subfigure}
    \caption{Realization of $ \boldsymbol{\beta}^{\otimes 4} + \overline{\boldsymbol{\beta}}^{\otimes 4} + \lambda^3 \mathbf{e}_3^{\otimes 4}$. Square nodes are $(\ne_{\db\dg;\dr})$.} \label{fig:Z-normalization}
  \end{figure}
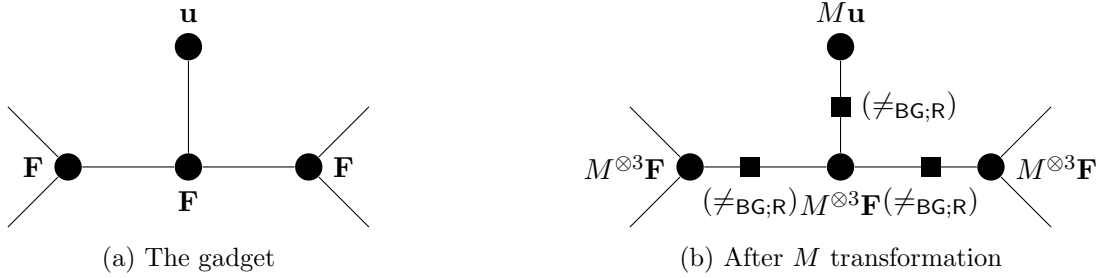
  It is not immediately obvious why the gadget on the left is the signature $\mathbf{F}' = \boldsymbol{\beta}^{\otimes 4} + \overline{\boldsymbol{\beta}}^{\otimes 4} + \lambda^3 \mathbf{e}_3^{\otimes 4}$, or even symmetric.
  To compute its signature, we compute it after applying the holographic transformation $M = \begin{bsmallmatrix}
    Z^{-1} & 0 \\
    0 & 1
    \end{bsmallmatrix}$ where $Z = \frac{1}{\sqrt{2}} \begin{bsmallmatrix}
    1 & 1 \\
    i & -i
    \end{bsmallmatrix}$ and $Z^{-1} = \frac{1}{\sqrt{2}}\begin{bsmallmatrix}
    1 & -i \\
    1 & i
  \end{bsmallmatrix}$.
  Then, $M\teh \mathbf{F} = \mathbf{e}_1 \teh + \mathbf{e}_2 \teh + \lambda \mathbf{e}_3 \teh$.
  Since $M$ is not orthogonal, it transforms the binary equality in the covariant way, $(=_2) (M^{-1})\tew$ which has the matrix form $(M^{-1})\transpose I (M^{-1}) = \begin{bsmallmatrix}
    0 & 1 & 0 \\
    1 & 0 & 0 \\
    0 & 0 & 1
  \end{bsmallmatrix}$, which can be viewed as a Disequality on $\{\db, \dg\}$ and Equality on $\dr$. We use $(\ne_{\db \dg; \dr})$ to denote this signature.
  Then, we can chosoe $\mathbf{u} = (\sqrt{2}, 0, 1)$, which satisfies $(\ne_{\db \dg; \dr}) M \mathbf{u} = (1, 1, 1)$.
  The middle part then becomes $\mathbf{e}_1 \tew + \mathbf{e}_2 \tew + \lambda \mathbf{e}_3 \tew$.
  It can absorb the two $(\ne_{\db \dg ; \dr})$ on both sides and still stay the same.
  So, the whole signature is $\mathbf{e}_1^{\otimes 4} + \mathbf{e}_2^{\otimes 4} + \lambda^3 \mathbf{e}_3^{\otimes 4}$.
  Applying $M^{-1}$ on this gadget, we get $\mathbf{F}' = \boldsymbol{\beta}^{\otimes 4} + \overline{\boldsymbol{\beta}}^{\otimes 4} + \lambda^3 \mathbf{e}_3 ^{\otimes 4}$.

  By \cref{prop-boolean-type-ii-tensor-decomposition}, we know there exists some $c, d \in \mathbb{C}$ and $z \in \mathbb{R}$ such that $\mathbf{G} = c \boldsymbol{\beta} \teh + d \overline{\boldsymbol{\beta}}\teh + \lambda^3 z \mathbf{e}_3$. From the proof of \cref{prop-boolean-type-ii-tensor-decomposition}, $d = \overline{c}$.
  Since $\boldsymbol{\beta}, \overline{\boldsymbol{\beta}}, \mathbf{e}_3$ are linearly independent,
  there exists a unary signature $\mathbf{v} \in \mathbb{R}^3$ so that 
  $\langle \mathbf{F}', \mathbf{v} \rangle = \mathbf{G}$.

\end{proof}

\begin{corollary}\label{cor:z-normalization-orthogonal}
  Let $\mathbf{F} = \boldsymbol{\beta} \teh + \overline{\boldsymbol{\beta}} \teh + \lambda \mathbf{e}_3 \teh$ for $\lambda \in \mathbb{R}$ where $\boldsymbol{\beta} = \frac{1}{\sqrt{2}} (1, i, 0)$.
  Let $\mathbf{G} = T \teh \mathbf{F}$ for some real $3 \times 3$ orthogonal $\db \dg | \dr$ matrix.
  In particular, $\mathbf{G} = (\mathbf{u} + i \mathbf{v})\teh + (\mathbf{u} - i \mathbf{v})\teh + \lambda \mathbf{e}_3\teh$ such that $\mathbf{u}$, $\mathbf{v}$, and $\mathbf{e}_3$ form an orthonormal basis of $\mathbb{R}^3$.
  Let $\mathcal{F}$ be any set of signatures. Then,
  \[
    \holts(\mathcal{F} \cup \{\mathbf{F}\}) =_T \holts(\mathcal{F} \cup \{\mathbf{G}\}) \, .
  \]
\end{corollary}

The main strategy of proving hardness is by showing that domain restriction can be realized within $\mathcal{F}$ and using the Boolean domain dichotomy.
\begin{proposition}\label{lem:domain-restriction}
  Let $\mathcal{F}$ be a set of signatures such that $(=_{\db \dg}) \in \mathcal{F}$.
  Then,
  \[
    \holbs(\mathcal{F}\domres{\db, \dg}) \le_T \holts(\mathcal{F}) \, .
  \]
\end{proposition}
\begin{proof}
  Let $\Omega$ be a signature grid of $\holbs(\mathcal{F}\domres{\db, \dg}) $.
  We may construct a $\holts(\mathcal{F})$ signature grid $\Omega'$ with the same Holant value by.
  For each vertex with signature $\mathbf{F}\domres{\db, \dg}$ in $\Omega$, we replace it with $\mathbf{F}$ in $\Omega'$;
  for each edge in $\Omega$, we replace it with $(=_{\db \dg})$ in $\Omega'$.
  Then, every edge in $\Omega'$ can only take values $\{\db, \dg\}$, so the Holant value of $\Omega$ and $\Omega'$ are the same.
\end{proof}

\begin{proposition}\label{lem:geneq-eqbg}
  Let $\mathbf{F} =  (a, b, 0)\teh + (c, d, 0)\teh$ for $a, b, c, d \in \mathbb{R}$ such that $(a, b, 0)$ and $(c, d, 0)$ are nonzero orthogonal vectors.
  Let $\mathcal{F}$ be any set of signatures containing $\mathbf{F}$.
  Then, 
  \[
    \holts(\mathcal{F} \cup \{(=_{\db \dg})\}) \le_T \holts(\mathcal{F}) \, .
  \]
\end{proposition}
\begin{proof}
  By \cref{lem:geneq-normalization}, we may assume that $a^2 + b^2 = c^2 + d^2 = 1$.
  Then, $T = \begin{bsmallmatrix}
    a & b \\
    c & d
    \end{bsmallmatrix}$ is orthogonal, and $T^{\tt T}T = I$.
  Let $\mathbf{u} \in \mathbb{R}^3$ such that $\begin{bsmallmatrix}
    a & b & 0 \\
    c & d & 0
    \end{bsmallmatrix} \mathbf{u} = \begin{bsmallmatrix}
    1 \\ 1 
  \end{bsmallmatrix}$, which exists since $(a, b, 0)$ and $(c, d, 0)$ are linearly independent.
  Then,
  \[
    \langle \mathbf{F}, \mathbf{u} \rangle = (a, b, 0)\tew + (c, d, 0)\tew = \begin{bmatrix}
      a^2 + c^2 & ab + cd & 0 \\
      ab + cd & b^2 + d^2 & 0 \\
      0 & 0 & 0
      \end{bmatrix} = 
      \begin{bmatrix}
          T^{\tt T} T & 0 \\
          0 & 0
      \end{bmatrix}
      = \begin{bmatrix}
      1 & 0 & 0\\
      0 & 1 & 0 \\
      0 & 0 & 0
    \end{bmatrix}.
  \]
\end{proof}

\begin{proposition}\label{lem:z-eqbg}
  Let $\mathbf{F} = \boldsymbol{\beta} \teh + \overline{\boldsymbol{\beta}} \teh$ where $\boldsymbol{\beta} = \frac{1}{\sqrt{2}} (1, i, 0)$.
  Let $\mathcal{F}$ be any set of signatures containing $\mathbf{F}$.
  Then, 
  \[
    \holts(\mathcal{F} \cup \{(=_{\db \dg})\}) \le_T \holts(\mathcal{F}) \, .
  \]
\end{proposition}
\begin{proof}
  We can see that $\langle \mathbf{F}, \mathbf{e}_1 \rangle = \frac{1}{\sqrt{2}}( \boldsymbol{\beta} \tew + \overline{\boldsymbol{\beta}} \tew)$, which has the matrix form $\begin{bsmallmatrix}
    1 & 0 & 0 \\
    0 & -1 & 0 \\
    0 & 0 & 0
  \end{bsmallmatrix}$ after normalization.
  Connecting two of them gives $(=_{\db \dg})$.
\end{proof}

\subsection{Rank 2 Type \texorpdfstring{\ternarytractgeneq}{A}}
Let $\mathbf{F}$ be a type \ternarytractgeneq signature of rank $2$.
After reordering the domain, we may assume that $\mathbf{F} = a \mathbf{e}_1 \teh + b \mathbf{e}_2 \teh$, for nonzero $a, b \in \mathbb{R}$.
By \cref{lem:geneq-normalization}, we may assume $a, b = 1$.

\begin{lemma}\label{lem:dichotomy-single-ternary-rank-2-geneq-single-binary}
  Let $\mathbf{F} = \mathbf{e}_1 \teh + \mathbf{e}_2 \teh$.
  Let $\mathbf{G}$ be a nondegenerate, real-valued, symmetric, binary domain $3$ signature.
  Then, $\holts(\mathbf{F}, \mathbf{G})$ is computable in polynomial time if one of the following conditions holds. Otherwise, it is \sph.
  \begin{enumerate}
    \item $\mathbf{G}$ is $\db \dg | \dr$ and $\holbs(\mathbf{G}\domres{\db, \dg}, [1, 0, 0, 1])$ is tractable.
    \item $\mathbf{G}$ is $\db \dr | \dg$ or $\bdgr$.
    \item $\mathbf{G}$ is \swbg, \swgr, or \swbr.
    \item $\mathbf{G}$ is in $\mathcal{D}$.
    \item $\mathbf{G} = c \begin{bmatrix}
        1 & x & -x \alpha \\
        x & x^2 & \alpha \\
        -x \alpha & \alpha & 0
      \end{bmatrix}$ for $\alpha = \pm \sqrt{1 + x^2}$ and nonzero $c, x \in \mathbb{R}$.
  \end{enumerate}
\end{lemma}
\begin{proof}
  Observe that $\mathbf{F}$ is $\db \dg | \dr$, $\db \dr | \dg$ and $\bdgr$ simultaneously.
  So, case $1$ is the tractable class \tractBGR.
  Case $2$ and $3$ are also \tractBGR after renaming the domain.
  Case $4$ is the tractable class \tractBG since $\supp \mathbf{F} \subseteq \{\db, \dg\}^*$.
  For case $5$, we may apply a $\db \dg | \dr$ orthogonal holographic transformation $T = \begin{bsmallmatrix}
    -x/\alpha & 1/\alpha & 0 \\
    1/\alpha & x/\alpha & 0 \\
    0 & 0 & 1
  \end{bsmallmatrix}$ and get
  \[
    \mathbf{G}' = T \tew \mathbf{G} = T \mathbf{G} T\transpose = (1 + x^2) \begin{bmatrix}
      0 & 0 & 1 \\
      0 & 1 & 0 \\
      1 & 0 & 0
    \end{bmatrix}
  \] while $\mathbf{F}' = T \teh \mathbf{F} = (-x/\alpha, 1/\alpha, 0) \teh + (1/\alpha, x/\alpha, 0) \teh$.
  We can check that $\hols(\mathbf{F}', \mathbf{G}')$ falls under the class \tractBGGRBR.

  To show that the above cases are the only tractable cases, we will start by assuming that $\holts(\mathbf{F}, \mathbf{G})$ is not \sph.
  Then, we use the known dichotomy theorems to show that $\mathbf{G}$ must take a form stated in the lemma. 
  By \cref{lem:geneq-eqbg} and \cref{lem:domain-restriction}, if $\holts(\mathbf{F}, \mathbf{G})$ is not \sph, then a domain restriction to the Boolean domain $\{\db, \dg\}$ cannot be \sph. 
  Then, we may use \cref{thm:dich-sym-Boolean} on $\hols(\{(\mathbf{G}^k)\domres{\db, \dg} : k \ge 0\}  \cup \{\mathbf{F} \domres{\db, \dg}\})$, which corresponds to looking at $2 \times 2$ submatrix of $\mathbf{G}^k$ corresponding to the domain $\{\db, \dg\}$ for some $k$.
  We may also consider applying \cref{thm:dich-single-sym-ter-dom3-real} on $\holts(\mathbf{G} \teh \mathbf{F})$, which corresponds to looking at the first two columns of the matrix form of $\mathbf{G}$.

  Let $
    \mathbf{G} = \begin{bsmallmatrix}
      a & b & c \\
      b & d & e \\
      c & e & f
    \end{bsmallmatrix}$. 
  For $\mathbf{G}\teh \mathbf{F} = (a, b, c)\teh + (b, d, e)\teh$ to be a tractable ternary signature, by \cref{thm:dich-single-sym-ter-dom3-real}, we need that $(a, b, c)$ and $(b, d, e)$ are either linearly dependent or orthogonal.
  If $(a, b, c)$ and $(b, d, e)$ are linearly dependent, then $\mathbf{G} \in \mathcal{D}$.
  So, we consider the cases when $(a, b, c)$ and $(b, d, e)$ are (non-zero and)
  orthogonal.
  We may also look at $\mathbf{G}\domres{\db, \dg} = [a, b, d]$. 
  This must be compatible with $\mathbf{F}\domres{\db, \dg} = [1, 0, 0, 1]$.
  The only possibilities are then that $[a, b, d]$ is degenerate, or $[*, 0, *]$, or $[0, *, 0]$.

  Suppose $[a, b, d] = [*, 0, *]$. 
  For $(a, 0, c)$ and $(0, d, e)$ to be orthogonal, it must be the case that $ce = 0$.
  If $c = 0$, then $\mathbf{G}$ is $\bdgr$.
  If $e = 0$, then $\mathbf{G}$ is $\db \dr | \dg$.
  Both cases are tractable.

  Suppose $[a, b, d] = [0, *, 0]$.
  For $(0, b, c)$ and $(b, 0, e)$ to be orthogonal, it must be the case that $ce = 0$.
  \begin{itemize}
    \item 
      Suppose $c = 0$.
      We may apply the same analysis to $\mathbf{G}^2 = \begin{bsmallmatrix}
        b^2 & 0 & be \\
        0 & b^2 + e^2 & ef \\
        be & ef & e^2 + f^2
      \end{bsmallmatrix}$.
      The first two columns of $\mathbf{G}^2$ being linearly dependent implies $b = 0$. Then, $\supp \mathbf{G} \subseteq \{\dg, \dr\}^*$, so $\mathbf{G}$ is $\bdgr$.
      The first two columns of $\mathbf{G}^2$ being orthogonal implies $b e^2 f = 0$. 
      We may assume $b \ne 0$.
      If $f = 0$, $\mathbf{G}$ is \swbr.
      If $e = 0$, $\mathbf{G} \in \mathcal{E}$.
    \item
      Suppose $e = 0$.
      We may apply the same analysis to $\mathbf{G}^2 = \begin{bsmallmatrix}
        b^2 + c^2 & 0 & cf \\
        0 & b^2 & bc \\
        cf & bc & c^2 + f^2
      \end{bsmallmatrix}$.
      The first two columns being linearly dependent implies $b = 0$. Then, $\supp \mathbf{G} \subseteq \{\db, \dr\}^*$.
      The first two columns being orthogonal implies $b c^2 f = 0$.
      We may assume $b \ne 0$.
      If $c = 0$, we go back to the previous case.
      If $f = 0$, $\mathbf{G}$ is \swgr.
  \end{itemize}

  Suppose $[a, b, d]$ is degenerate.
  We may assume $b \ne 0$, since that has been already handled.
  If $b \ne 0$ and $[a, b, d]$ is degenerate, then $a, d \ne 0$ as well, so we may normalize $\mathbf{G}$ so that $[a, b, d] = [1, x, x^2]$ for some nonzero $x \in \mathbb{R}$.
  We write $\mathbf{G}$ in the following way:
  \[
    \mathbf{G} = \begin{bmatrix}
      1 & x & c \\
      x & x^2 & e \\
      c & e & f 
      \end{bmatrix} = \begin{bmatrix}
      \vrule & \vrule & \vrule \\
      \mathbf{v}_1 & \mathbf{v}_2 & \mathbf{v}_3 \\
      \vrule & \vrule & \vrule
    \end{bmatrix}
  \]
  where $\mathbf{v}_i$ denotes the $i$th column of $\mathbf{G}$. By assumption, $\langle \mathbf{v}_1, \mathbf{v}_2 \rangle = x + x^3 + ce = 0$, which implies $ce \ne 0$, since $x \ne 0$ and $x^2 + 1 \ne 0$.
  Then,
  \[
    \mathbf{G}^2 = \begin{bmatrix}
      \langle \mathbf{v}_1, \mathbf{v}_1 \rangle & 0 & \langle \mathbf{v}_1, \mathbf{v}_3 \rangle \\
      0 & \langle \mathbf{v}_2, \mathbf{v}_2 \rangle & \langle \mathbf{v}_2, \mathbf{v}_3 \rangle \\
      \langle \mathbf{v}_1, \mathbf{v}_3 \rangle & \langle \mathbf{v}_2, \mathbf{v}_3 \rangle & \langle \mathbf{v}_3, \mathbf{v}_3 \rangle 
      \end{bmatrix} = \begin{bmatrix}
      p & 0 & q \\
      0 & r & s \\
      q & s & t
    \end{bmatrix}
    \, .
  \]
  Since $p, r \ne 0$, the first two columns cannot be linearly dependent. 
  Therefore, the first two columns must be orthogonal for $\holts((\mathbf{G}^2)\teh \mathbf{F})$ to be not \sph, which implies that $q = 0$ or $s = 0$.

  Then we have the following three signatures at our disposal:
  $\mathbf{F} = \mathbf{e}_1 \teh + \mathbf{e}_2 \teh$, $\mathbf{G}\teh \mathbf{F} = (1, x, c)\teh + (x, x^2, e)\teh$ and $(\mathbf{G}^2)\teh \mathbf{F} = (p, 0, q) \teh + (0, r, s)\teh$.
  We apply holographic transformation by $T =
  \begin{bsmallmatrix}
    -x/\alpha & 1/\alpha & 0 \\
    1/\alpha & x/\alpha & 0 \\
    0 & 0 & 1
  \end{bsmallmatrix}$ for $\alpha = \sqrt{1 + x^2}$, which gives us 
  \[
    \mathbf{H}_1 = 
    T \teh (\mathbf{G} \teh \mathbf{F}) = \begin{bmatrix}
      0 \\ (1 + x^2)/\alpha \\ c
      \end{bmatrix}\teh + \begin{bmatrix}
      0 \\ (x + x^3)/\alpha \\ e
    \end{bmatrix}\teh
  \]
  and
  \[
    \mathbf{H}_2 = 
    T\teh ( (\mathbf{G}^2)\teh \mathbf{F}) = \begin{bmatrix}
      -px/\alpha \\ p/\alpha \\ q
      \end{bmatrix}\teh + \begin{bmatrix}
      r/\alpha \\ rx/\alpha \\ s
    \end{bmatrix}\teh
    \, .
  \]
  By \cref{lem:geneq-eqbg},  $\mathbf{H}_1$ can realize $(=_{\dg \dr})$, so by \cref{lem:domain-restriction}, we may apply Boolean domain dichotomy on 
  $\mathbf{H}_1\domres{\dg, \dr}$ and $\mathbf{H}_2\domres{\dg, \dr}$.
  Recall that $qs = 0$.
  Assume $q = 0$.
  Then, since $x, p, r \ne 0$, if $s \ne 0$, then $\mathbf{H}_2\domres{\dg, \dr}$ is nondegenerate.
  Then $\holbs(\mathbf{H}_2\domres{\dg, \dr})$ is \sph  since $prx/\alpha^2 + q s = prx/\alpha^2  \ne 0$.
  Therefore, it must be the case that $s = 0$, meaning $\mathbf{v}_1, \mathbf{v}_2, \mathbf{v}_3$ are pairwise orthogonal.
  It follows that $\mathbf{v}_3 = \lambda(-x, 1, 0)$ for some $\lambda \in \mathbb{R}$, and it must be the case that $\lambda^2 = 1 + x^2$ since $\langle \mathbf{v}_1, \mathbf{v}_2  \rangle = x + x^3 - \lambda^2 x = 0$.
  Therefore, $\mathbf{G}$ is case $5$.
  We may use the same argument for the case when $s = 0$.
\end{proof}

\subsection{Rank 3 Type \texorpdfstring{\ternarytractgeneq}{A}}
We first prove a lemma about $\mathcal{D}$ signatures when a rank $2$ \geneq with a different support is present.
\begin{lemma}\label{lem:gr-equality-D-dichotomy}
  Let $\mathbf{F} = \mathbf{e}_2 \teh + \mathbf{e}_3 \teh$ and $\mathbf{G} \in \mathcal{D}$ be symmetric.
  Then, $\holts(\mathbf{F}, \mathbf{G})$ is computable in polynomial time if $\holbs(\mathbf{G}\domres{\dg, \dr}, [1, 0, 0, 1])$ is tractable and $\mathbf{G}$ is degenerate, $\db \dg | \dr$, \swbg, or \strspt.
  Otherwise, it is \sph.
\end{lemma}
\begin{proof}
  We will assume that $\holts(\mathbf{F}, \mathbf{G})$ is not \sph and show that $\mathbf{G}$ must be of the form stated in the lemma.
  Since $\mathbf{G} \in \mathcal{D}$, we can write $\mathbf{G} = \begin{bsmallmatrix}
    a x^2 & a xy & bx \\
    axy & ay^2 & by \\
    bx & by & c
  \end{bsmallmatrix}$.
  By \cref{lem:geneq-eqbg} and \cref{lem:domain-restriction}, we need $\mathbf{G}\domres{\dg, \dr} = [ay^2, by, c]$ to be compatible with $\mathbf{F}\domres{\dg, \dr} = [1, 0, 0, 1]$.
  For $\holbs([1, 0, 0, 1], [ay^2, by, c])$ to be tractable, $[ay^2, by, c]$ must be degenerate, $[*, 0, *]$ or $[0, *, 0]$ by Boolean domain dichotomy.

  Suppose $[ay^2, by, c] = [*, 0, *]$
  Then, either $b = 0$ or $y = 0$.
  If $b = 0$, then $\mathbf{G}$ is $\db \dg | \dr$.
  If $y = 0$, then $\supp \mathbf{G} \subseteq \{\db, \dr\}^*$.

  Suppose $[ay^2, by, c] = [0, *, 0]$.
  Then, $c = 0$ and either $a = 0$ or $y = 0$.
  If $a = 0$, then $\mathbf{G}$ is \swbg.
  If $y = 0$, then $\supp \mathbf{G} \subseteq \{\db, \dr\}^*$.

  Suppose $[ay^2, by, c]$ is degenerate. 
  We may assume that $y, b \ne 0$, since otherwise, we go back to the previous case.
  Then, the vectors $(ay^2, by)$ and $(by, c)$ are linearly dependent nonzero vectors.
  In particular, there exists some $\lambda \in \mathbb{R}$ such that $(by, c) = \lambda (ay^2, by)$.
  Then, the three vectors $(axy, bx)$, $(ay^2, by)$ and $(by, c)$ are linearly dependent, which means the bottom two rows of $\mathbf{G}$ are linearly dependent.
  This implies that rank of $\mathbf{G}$ is $1$, so $\mathbf{G}$ is degenerate.
\end{proof}

Let $\mathbf{F}$ be a type \ternarytractgeneq signature of rank $3$, so $\mathbf{F} = a \mathbf{e}_1 \teh + b \mathbf{e}_2 \teh + c \mathbf{e}_3\teh$, for nonzero $a, b, c \in \mathbb{R}$.
By \cref{lem:geneq-normalization}, we may assume $a, b, c = 1$.
\begin{lemma}\label{lem:dichotomy-single-ternary-rank-3-geneq-single-binary}
  Let $\mathbf{F} = \mathbf{e}_1 \teh + \mathbf{e}_2 \teh + \mathbf{e}_3 \teh$.
  Let $\mathbf{G}$ be a nondegenerate, real-valued, symmetric, binary domain $3$ signature.
  Then, $\holts(\mathbf{F}, \mathbf{G})$ is computable in polynomial time if one of the following conditions holds. Otherwise, it is \sph.
  \begin{enumerate}
    \item For some $i, j, k \in \{\db, \dg, \dr\}$, $\mathbf{G}$ is $ij | k$ and $\holbs(\mathbf{G}\domres{i, j}, [1, 0, 0, 1])$ is tractable.
    \item $\mathbf{G}$ is \swbg, \swgr, or \swbr.
  \end{enumerate}
\end{lemma}
\begin{proof}
  We observe that $\mathbf{F}$ is $\db \dg | \dr$, $\db \dr | \dg$, and $\bdgr$.
  Case 1 and 2 are both the tractable class \tractBGR.

Suppose $\holts(\mathbf{F}, \mathbf{G})$ is not \sph.
  By \cref{lem:geneq-normalization}, $\mathbf{F}$ can realize $\mathbf{F}_1 = \mathbf{e}_1\teh + \mathbf{e}_2 \teh$, $\mathbf{F}_2 = \mathbf{e}_1 \teh + \mathbf{e}_3 \teh$ and $\mathbf{F}_3 = \mathbf{e}_2 \teh + \mathbf{e}_3 \teh$.
  $\holts(\mathbf{F}_1, \mathbf{G})$ cannot be \sph, so we may assume that $\mathbf{G}$ is of a tractable form in \cref{lem:dichotomy-single-ternary-rank-2-geneq-single-binary}.

  We only need to consider the case 2, 4, and 5 of \cref{lem:dichotomy-single-ternary-rank-2-geneq-single-binary}, since others are tractable.
  By \cref{lem:geneq-eqbg} and \cref{lem:domain-restriction} on domain $\{\db, \dr\}$ and $\{\dg, \dr\}$ we need $\mathbf{G}\domres{\db, \dr}$ and $\mathbf{G}\domres{\dg, \dr}$ to be compatible with $[1, 0, 0, 1] = \mathbf{F}_2\domres{\db, \dr} = \mathbf{F}_3\domres{\dg, \dr}$.
  Therefore, if $\mathbf{G}$ is in case 2 of \cref{lem:dichotomy-single-ternary-rank-2-geneq-single-binary}, i.e. either $\db \dr | \dg$ or $\bdgr$, its domain restriction must be compatible with $[1, 0, 0, 1]$, meaning it must be in case 1 of this lemma.

  For the case of $\mathbf{G} \in \mathcal{D}$, we use \cref{lem:gr-equality-D-dichotomy}.
  Note that if $\mathbf{G}$ is \strspt on $\{i, j\}$, then it is also $i j | k$.

  If $\mathbf{G}$ is in case 5 of \cref{lem:dichotomy-single-ternary-rank-2-geneq-single-binary}, then $\mathbf{G}\domres{\dg, \dr} = [x^2, \alpha, 0]$ for nonzero $x$ and $\alpha = \pm \sqrt{1 + x^2}$.
  Since both $x^2, \alpha \ne 0$, it is not compatible with $[1, 0, 0, 1]$ by the Boolean domain dichotomy. 
\end{proof}

\subsection{Rank 2 Type \texorpdfstring{\ternarytractz}{B}}
\begin{lemma}\label{lem:dichotomy-single-ternary-rank-2-z-single-binary}
  Let $\mathbf{F} = \boldsymbol{\beta}\teh + \overline{\boldsymbol{\beta}}\teh$ where $\boldsymbol{\beta} = \frac{1}{\sqrt{2}}(1, i, 0)\transpose$.
  Let $\mathbf{G}$ be a nondegenerate, real-valued, symmetric, binary domain $3$ signature.
  Then, $\holts(\mathbf{F}, \mathbf{G})$ is computable in polynomial time if one of the following conditions holds. Otherwise, it is \sph.
  \begin{enumerate}
    \item $\mathbf{G}$ is $\db \dg | \dr$ and $\mathbf{G}\domres{\db, \dg}$ is tractable with type $\typeii$ signatures.
    \item $\mathbf{G}$ is $c \begin{bmatrix}
        1 & 0 & 0 \\
        0 & 0 & \alpha \\
        0 & \alpha & 0
        \end{bmatrix}$ or $c \begin{bmatrix}
        0 & 0 & \alpha \\
        0 & 1 & 0 \\
        \alpha & 0 & 0
      \end{bmatrix}$ for $\alpha = \pm 1$ and nonzero $c \in \mathbb{R}$.
    \item $\mathbf{G}$ is in $\mathcal{D}$.
    \item $\mathbf{G} = c  \begin{bmatrix}
        1 & x & -x \alpha \\
        x & x^2 & \alpha \\
        -x \alpha & \alpha & 0
      \end{bmatrix}$ for $\alpha = \pm \sqrt{1 + x^2}$ and nonzero $c, x \in \mathbb{R}$.
  \end{enumerate}
\end{lemma}
\begin{proof}
  Observe that $\mathbf{F}$ is $\db \dg | \dr$ and $\supp \mathbf{F} \subseteq \{\db, \dg\}^*$.
  Case $2$ is the tractable class \tractBGGRBR.
  Case $3$ is the tractable class \tractBG since $\supp \mathbf{F} \subseteq \{\db, \dg\}^*$.
  For case $5$, we may apply a $\db\dg|\dr$ orthogonal holographic transformation $T = \begin{bsmallmatrix}
    -x/\alpha & 1/\alpha & 0 \\
    1/\alpha & x/\alpha & 0 \\
    0 & 0 & 1
  \end{bsmallmatrix}$ and get
  \[
    T \tew \mathbf{G} = T \mathbf{G} T\transpose = (1 + x^2) \begin{bmatrix}
      0 & 0 & 1 \\
      0 & 1 & 0 \\
      1 & 0 & 0
    \end{bmatrix}
  \] while $\supp T \teh \mathbf{F} \subseteq \{\db, \dg\}^*$.
  We can check that $\hols(T \teh \mathbf{F}, T \teh\mathbf{G})$ falls under the class \tractBGGRBR.

  Suppose $\holts(\mathbf{F}, \mathbf{G})$ is not \sph.
  Similar to the proof of \cref{lem:dichotomy-single-ternary-rank-2-geneq-single-binary}, we will analyze the possible forms of $\mathbf{G}$ using the known dichotomy theorems.
  By \cref{lem:z-eqbg} and \cref{lem:domain-restriction}, we may look the domain restriction.

  Let $\mathbf{G} = \begin{bsmallmatrix}
    a & b & c \\
    b & d & e \\
    c & e & f
    \end{bsmallmatrix} = \begin{bsmallmatrix}
    \vert & \vert & \vert \\
    \mathbf{v}_1 & \mathbf{v}_2 & \mathbf{v}_3 \\
    \vert & \vert & \vert
  \end{bsmallmatrix}$.
  For $\mathbf{G}\teh \mathbf{F} = (\mathbf{v}_1 + i \mathbf{v}_2)\teh + (\mathbf{v}_1 - i \mathbf{v}_2)\teh$ to be a tractable ternary signature, by \cref{cor:non-orthogonal-linearly-independent-hard}, we need either $\mathbf{v}_1, \mathbf{v}_2$ to be orthogonal or linearly dependent.
  If they are linearly dependent, then $\mathbf{G} \in \mathcal{D}$.
  So, we consider the case when $\mathbf{v}_1$ and $\mathbf{v}_2$ are orthogonal.
  We also look at $\mathbf{G}\domres{\db, \dg} = [a, b, d]$.
  It must be compatible with type $\typeii$, so the possibilities are that $[a, b, d]$ is degenerate, $[x, y, -x]$ or $[x, 0, x]$ for some $x, y \in \mathbb{R}$.

  Suppose $[a, b, d] = [x, 0, x]$ for $x \ne 0$.
  Then, for $\mathbf{v}_1$ and $\mathbf{v}_2$ to be orthogonal, we need $c = 0$ or $e = 0$.
  \begin{itemize}
    \item 
      Suppose $c = 0$, then $(\mathbf{G}^2)\domres{\db, \dg} = [x^2, 0, x^2 + e^2]$.
      This can be only compatible with type $\typeii$ if $e = 0$ or $x = 0$.
      If $x = 0$, $\mathbf{v}_1$ and $\mathbf{v}_2$ are linearly dependent.
      If $e = 0$, then $\mathbf{G} = \begin{bsmallmatrix}
          x & 0 & 0 \\
          0 & x & 0 \\
          0 & 0 & f
      \end{bsmallmatrix}$, so it is case 1.
    \item 
      Suppose $e = 0$. The argument is similar. 
  \end{itemize}

  Suppose $[a, b, d] = [x, y, -x]$.
  Then, we have $\langle \mathbf{v}_1, \mathbf{v}_2 \rangle = 0 = ab + bd + ce = xy - yx + ce$, so $ce = 0$.
  \begin{itemize}
    \item Suppose $c = 0$.
      Then, $(\mathbf{G}^2)\domres{\db, \dg} = [x^2 + y^2, 0, x^2 + y^2 + e^2]$.
      For this to be compatible with type $\typeii$, we need $e = 0$ or $x^2 + y^2 = 0$.
      If $x^2 + y^2 = 0$, then $x = y = 0$ so $\mathbf{v}_1$ and $\mathbf{v}_2$ are linearly dependent.
      If $e = 0$, then $\mathbf{G}$ is $\db \dg | \dr$, so it is case 1.

    \item Suppose $e = 0$. The argument is similar.
  \end{itemize}

  Suppose $[a, b, d]$ is $[1, 0, 0]$ after scaling.
  Then, for $\mathbf{v}_1$ and $\mathbf{v}_2$ to be orthogonal, we need $c = 0$ or $e = 0$.
  If $e = 0$, then $\mathbf{v}_2 = 0$, so $\mathbf{v}_1$ and $\mathbf{v}_2$ are linearly dependent.
  So, assume $c = 0$ and $e \ne 0$. 
  Then $(\mathbf{G}^2)\domres{\db, \dg} = [1, 0, e^2]$, so we need $e^2 = 1$ to be compatible with type $\typeii$.
  Taking one more square, we have $(\mathbf{G}^4)\domres{\db, \dg} = [1, 0, e^4 + e^2f^2] = [1, 0, 1 + f^2]$.
  We must have $f = 0$, which implies $\mathbf{G}$ is in case 2.
  The case for $[a, b, d] = [0, 0, 1]$ after scaling follows a similar argument.

  Suppose $[a, b, d]$ is degenerate, and we may now assume that it is $[1, x, x^2]$ after scaling.
  We have
  \[
    \mathbf{G}^2 = \begin{bmatrix}
      \langle \mathbf{v}_1, \mathbf{v}_1 \rangle & 0 & \langle \mathbf{v}_1, \mathbf{v}_3 \rangle \\
      0 & \langle \mathbf{v}_2, \mathbf{v}_2 \rangle & \langle \mathbf{v}_2, \mathbf{v}_3 \rangle \\
      \langle \mathbf{v}_1, \mathbf{v}_3 \rangle & \langle \mathbf{v}_2, \mathbf{v}_3 \rangle & \langle \mathbf{v}_3, \mathbf{v}_3 \rangle 
      \end{bmatrix} = \begin{bmatrix}
      p & 0 & q \\
      0 & r & s \\
      q & s & t
    \end{bmatrix}
    \, .
  \]
  Since $p, r \ne 0$, the first two columns cannot be linearly dependent. Therefore, to be tractable, they must be orthogonal, which implies that $q = 0$ or $s = 0$.
  Also, since $p, r > 0$, it must be the case that $p = r$.

  Then we have the following three signatures at our disposal:
  $\mathbf{F} = \boldsymbol{\beta}\teh + \overline{\boldsymbol{\beta}}\teh$, $\mathbf{G}\teh \mathbf{F} = (\mathbf{v}_1 + i \mathbf{v}_2)\teh + (\mathbf{v}_1 - i \mathbf{v}_2)\teh$ and $(\mathbf{G}^2)\teh \mathbf{F} = \left(\begin{bmatrix}
    p \\ 0 \\ q
    \end{bmatrix} + i \begin{bmatrix}
    0 \\ r \\ s
    \end{bmatrix}\right) \teh + \left(\begin{bmatrix}
    p \\ 0 \\ q
    \end{bmatrix} - i \begin{bmatrix}
    0 \\ r \\ s
  \end{bmatrix} \right)\teh$.
  We apply holographic transformation by $T =
  \begin{bsmallmatrix}
    -x/\alpha & 1/\alpha & 0 \\
    1/\alpha & x/\alpha & 0 \\
    0 & 0 & 1
  \end{bsmallmatrix}$ for $\alpha = \sqrt{1 + x^2}$ which gives us 
  \[
    \mathbf{H}_1 = 
    T \teh (\mathbf{G} \teh \mathbf{F}) = \left(\begin{bmatrix}
        0 \\ \frac{1 + x^2}{\alpha} \\ c
        \end{bmatrix} + i \begin{bmatrix}
        0 \\ \frac{x + x^3}{\alpha} \\ e
    \end{bmatrix}\right)\teh + \left(
    \begin{bmatrix}
      0 \\ \frac{1 + x^2}{\alpha} \\ c
    \end{bmatrix}
    -i  \begin{bmatrix}
      0 \\ \frac{x + x^3}{\alpha} \\ e
    \end{bmatrix}
  \right)\teh
\]
and
\[
  \mathbf{H}_2 = 
  T\teh ( (\mathbf{G}^2)\teh \mathbf{F}) = \left(\begin{bmatrix}
      -px/\alpha \\ p/\alpha \\ q
      \end{bmatrix} + i \begin{bmatrix}
      r/\alpha \\ rx/\alpha \\ s
  \end{bmatrix}\right)
  \teh + \left(
    \begin{bmatrix}
      -px/\alpha \\ p/\alpha \\ q
      \end{bmatrix} -i      \begin{bmatrix}
      r/\alpha \\ rx/\alpha \\ s
  \end{bmatrix} \right)\teh
  \, .
\]
By \cref{lem:Z-normalization} and \cref{lem:z-eqbg}, 
we can realize $(=_{\dg \dr})$ using $\mathbf{H}_1$, so by \cref{lem:domain-restriction}, we may apply Boolean domain dichotomy on 
$\mathbf{H}_1\domres{\dg, \dr}$ and $\mathbf{H}_2\domres{\dg, \dr}$.
Assume $q = 0$.
Then, since $x, r, p \ne 0$, if $s \ne 0$, then $\mathbf{H}_2\domres{\dg, \dr}$ is nondegenerate.
Then $\hols(\mathbf{H}_2\domres{\dg, \dr})$ is \sph since $prx/\alpha^2 + q \cdot s \ne 0$.
Therefore, it must be the case that $s = 0$, meaning $\mathbf{v}_1, \mathbf{v}_2, \mathbf{v}_3$ are pairwise orthogonal.
It follows that $\mathbf{v}_3 = \lambda(-x, 1, 0)$ for some $\lambda \in \mathbb{R}$, and it must be the case that $\lambda^2 = 1 + x^2$ by orthogonality of $\mathbf{v}_1$ and $\mathbf{v}_2$, which is of the aforementioned tractable form.
We may use the same argument for the case when $s = 0$.
\end{proof}

\subsection{Rank 3 Type \texorpdfstring{\ternarytractz}{B}}
Let $\mathbf{F}$ be a type \ternarytractz signature of rank $3$, so $\mathbf{F} = (1, i, 0)\teh + (1, -i, 0)\teh + \lambda \mathbf{e}_3\teh$ for some $\lambda \in \mathbb{R}$.
By \cref{lem:Z-normalization}, we may assume that $\lambda = 1$.
\begin{lemma}\label{lem:dichotomy-single-ternary-rank-3-z-single-binary}
  Let $\mathbf{F} = \boldsymbol{\beta}\teh + \overline{\boldsymbol{\beta}}\teh + \mathbf{e}_3 \teh$ where $\boldsymbol{\beta} = \frac{1}{\sqrt{2}}(1, i, 0)\transpose$.
  Let $\mathbf{G}$ be a nondegenerate, real-valued, symmetric, binary domain $3$ signature.
  Then, $\holts(\mathbf{F}, \mathbf{G})$ is computable in polynomial time if one of the following conditions holds. Otherwise, it is \sph.
  \begin{enumerate}
    \item $\mathbf{G}$ is $\db \dg | \dr$ and $\mathbf{G}\domres{\db, \dg}$ is tractable with type $\typeii$ signatures.
    \item $\mathbf{G}$ is \swbg.
  \end{enumerate}
\end{lemma}
\begin{proof}
We see that both cases fall in the tractable class \tractBGR.

Suppose $\holts(\mathbf{F}, \mathbf{G})$ is not \sph.
  By \cref{lem:Z-normalization}, we can realize $\mathbf{H} = \boldsymbol{\beta}\teh + \overline{\boldsymbol{\beta}}\teh$ using $\mathbf{F}$, so we may assume that $\mathbf{G}$ is of the tractable form in \cref{lem:dichotomy-single-ternary-rank-2-z-single-binary} which lists four cases. Case 1 is 
  stated here to be tractable.


  For case 2, suppose $\mathbf{G} = \begin{bsmallmatrix}
    1 & 0 & 0 \\
    0 & 0 & \alpha \\
    0 & \alpha & 0
  \end{bsmallmatrix}$ for $\alpha = \pm 1$.
  Then, $\mathbf{G}\teh \mathbf{H} = (1, 0, i)\teh + (1, 0, -i)\teh$.
  By \cref{lem:z-eqbg}, $\mathbf{G}\teh \mathbf{H}$ can realize $(=_{\db \dr})$.
  By \cref{lem:Z-normalization}, we can realize a signature $\mathbf{A}$ using $\mathbf{F}$ such that $\mathbf{A}\domres{\db, \dr} = [1, 0, 0, 1]$.
  We see that $\holbs(\mathbf{A}\domres{\db, \dr}, (\mathbf{G}\teh \mathbf{H})\domres{\db, \dr})$ is \sph. 
  A similar argument works for other forms of case 2.

  For case 3, we write $\mathbf{G} = \begin{bsmallmatrix}
    a x^2 & axy & bx \\
    axy & ay^2 & by \\
    bx & by & c
    \end{bsmallmatrix} = \begin{bsmallmatrix}
    \vert & \vert & \vert \\
    x \mathbf{v} & y \mathbf{v} & \mathbf{u} \\
    \vert & \vert & \vert
  \end{bsmallmatrix}$ where $\mathbf{v} = (ax, ay, b)\transpose$.
  We may assume that $\mathbf{v}$ and $\mathbf{u}$ are linearly independent since otherwise $\mathbf{G}$ is degenerate.
  Also, $(x, y) \ne (0, 0)$.
  Let
  \[
    \mathbf{F}_{p,q} = \left(\begin{bmatrix}
      p \\ -q \\ 0
    \end{bmatrix} + i \begin{bmatrix}
    q \\ p \\ 0
    \end{bmatrix}\right)\teh + 
\left(\begin{bmatrix}
      p \\ -q \\ 0
    \end{bmatrix} - i \begin{bmatrix}
    q \\ p \\ 0
    \end{bmatrix}\right)\teh + 
    \mathbf{e}_3\teh \, ,
  \]
  for some $p, q \in \mathbb{R}$ such that $p^2 + q^2 = 1$ to be determined later.
  By \cref{cor:z-normalization-orthogonal}, $\mathbf{F}_{p, q}$ is realizable using $\mathbf{F}$.
  Let, $\mathbf{B}_{p,q} = \mathbf{G}\teh \mathbf{F}_{p,q}= z_{p, q} \mathbf{v}\teh + \mathbf{u}\teh$, where $z_{p, q} = 2 \Re( (px - qy) + i (qx + py) )^3$.
  Note that since $(x, y) \ne (0, 0)$, $z_{p, q} = 2 (px - qy)^3 - 6 (px - qy)(qx + py)^2$ is not an identically zero polynomial in $p, q$.
  Therefore, there exists some $p_0, q_0 \in \mathbb{R}$ such that $(p_0, q_0) \ne (0, 0)$ and $z_{p_0, q_0} \ne 0$.
  Then, since $z_{p,q}$ is homogeneous in $p, q$, for $(p_1, q_1) = \frac{1}{\sqrt{p_0^2 + q_0^2}}(p_0, q_0)$, we have $z_{p_1, q_1} \ne 0$ and $p_1^2 + q_1^2 = 1$.
  Let $\mathbf{B} = \mathbf{B}_{p_1, q_1} = z \mathbf{v}\teh + \mathbf{u}\teh$ for $z = z_{p_1, q_1} \ne 0$.

  If $\mathbf{v}$ and $\mathbf{u}$ are not orthogonal, then by \cref{cor:non-orthogonal-linearly-independent-hard}, $\holts(\mathbf{B})$ is \sph.
  If they are orthogonal, we apply the orthogonal holographic transformation $T = \begin{bsmallmatrix}
    x/\gamma & y/\gamma & 0 \\
    -y /\gamma & x /\gamma & 0 \\
    0 & 0 & 1
  \end{bsmallmatrix}$ where $\gamma = \sqrt{x^2 + y^2}$.
  Then, $\mathbf{B}' = T \teh \mathbf{B} = z(a \gamma, 0, b)\teh + (b \gamma, 0, c)\teh$, which can realize $(=_{\db \dr})$ by \cref{lem:geneq-eqbg}.
  By \cref{cor:z-normalization-orthogonal}, $T \teh \mathbf{F}$ can realize a signature $\mathbf{A}$ such that $\mathbf{A}\domres{\db, \dr} = [1, 0, 0, 1]$.
  For $\holbs((\mathbf{B}')\domres{\db, \dr}, \mathbf{A}\domres{\db, \dr})$ to be tractable, we must have $((a \gamma, b), (b \gamma, c)) \sim (\mathbf{e_1}, \mathbf{e_2})$.
  This implies that either $b = 0$ or $a, c = 0$.
  If $b = 0$, $\mathbf{G}$ is $\db \dg | \dr$.
  If $a, c = 0$, $\mathbf{G}$ is \swbg.

  For case 4, we may apply orthogonal holographic transformation by $T = \begin{bsmallmatrix}
    -x/\alpha & 1/\alpha & 0 \\
    1/\alpha & x/\alpha & 0 \\
    0 & 0 & 1
    \end{bsmallmatrix}$ and get $\mathbf{G} = \begin{bsmallmatrix}
    0 & 0 & 1 \\
    0 & 1 & 0 \\
    1 & 0 & 0
  \end{bsmallmatrix}$ after normalization. 
  By \cref{cor:z-normalization-orthogonal}, $\mathbf{F'} = T \teh \mathbf{F}$ can realize $\mathbf{F}$, so we go back to case 2.
\end{proof}

%% file: sections/two_ternaries.tex
\section{Two Ternary Signatures}\label{sec:two-ternary}
We show the dichotomy of $\holts(\mathbf{F}, \mathbf{G})$ for two ternary signatures $\mathbf{F}$ and $\mathbf{G}$.
After an orthogonal transformation, we may assume that $\mathbf{F}$ is in the canonical form in \cref{thm:dich-single-sym-ter-dom3-real}.
We separate into the cases when at least one of the signatures is rank $3$ and when both are rank $2$.
Note that we do not need to consider a rank $1$ signature since such signature is degenerate.
Further, we separate into the cases when $\mathbf{F}$ is of type \ternarytractgeneq and type \ternarytractz.
Similar to \cref{sec:single-ternary-signle-binary}, we normalize the individual constants occurring in tensor decomposition using \cref{lem:geneq-normalization} and \cref{lem:Z-normalization}. 
Note that in the proof in this section, we do not refer to the dichotomy theorems proved in \cref{sec:single-ternary-signle-binary}.

\subsection{Rank 2 Type \texorpdfstring{\ternarytractgeneq}{A}}
\begin{lemma}\label{lem:dichotomy-ternary-ternary-rank-2-geneq}
  Let $\mathbf{F} = \mathbf{e}_1\teh + \mathbf{e}_2\teh$.
  Let $\mathbf{G}$ be a nondegenerate, real-valued, symmetric, ternary signature of rank $2$.
  Then, $\holts(\mathbf{F}, \mathbf{G})$ is computable in polynomial time if one of the following conditions holds. Otherwise, it is \sph.
  \begin{enumerate}
    \item $\mathbf{G}$ is \geneq.
    \item For some nonzero $c \in \mathbb{R}$, $c \mathbf{G} = \mathbf{e}_i\teh + \mathbf{v} \teh$ for $\mathbf{v} \in \mathbb{R}^3$ such that $\langle \mathbf{e}_i, \mathbf{v} \rangle = 0$.
    \item $\mathbf{G} = \mathbf{u}\teh + \mathbf{v}\teh$ for $\mathbf{v} = (v_1, v_2, v_3), \mathbf{u} = (u_1, u_2, u_3) \in \mathbb{R}^3$ such that $(u_1, u_2)$ and $(v_1, v_2)$ are linearly dependent and $\langle \mathbf{u}, \mathbf{v} \rangle = 0$.
    \item $\mathbf{G} = (\mathbf{u} + i \mathbf{v})\teh + (\mathbf{u} - i \mathbf{v})\teh$ for $\mathbf{v} = (v_1, v_2, v_3), \mathbf{u} = (u_1, u_2, u_3) \in \mathbb{R}^3$ such that $(u_1, u_2)$ and $(v_1, v_2)$ are linearly dependent, $\langle \mathbf{u}, \mathbf{v} \rangle = 0$ and $\langle \mathbf{u}, \mathbf{u} \rangle = \langle \mathbf{v}, \mathbf{v} \rangle$.
  \end{enumerate}
\end{lemma}
\begin{proof}
  For tractability, we see that case 1 is class \tractE.
  Case 2 is class \tractBGR.
  For case $3$, we write $(u_1, u_2) = (px, py)$ and $(v_1, v_2) = (qx, qy)$ for some $p, q, x, y \in \mathbb{R}^3$.
  We may assume that $x^2 + y^2 = 1$ by scaling $p$ and $q$.
  We apply $\db \dg | \dr$ orthogonal holographic transformation $T = \begin{bsmallmatrix}
    -y & x & 0 \\
    x & y & 0 \\
    0 & 0 & 1
  \end{bsmallmatrix}$.
  Then, $\supp T\teh \mathbf{F} \subseteq \{\db, \dg\}^*$ and $\supp T\teh \mathbf{G} \subseteq \{\dg, \dr\}^*$.
  We see that 
  \[\holts(T\teh \mathbf{F}) =_T \holbs((T \teh \mathbf{F})\domres{\db, \dg}) = \holbs((-y, x)\teh + (x, y)\teh) \] 
  \[\holts(T\teh \mathbf{G}) =_T \holbs((T \teh \mathbf{G})\domres{\dg, \dr}) = \holbs( (p, u_3)\teh + (q, v_3)\teh)\]
  are tractable since $pq + u_3 v_3 = pq (x^2 + y^2) + u_3 v_3 = \langle \mathbf{u}, \mathbf{v} \rangle = 0$.
  Thus, this falls under class \tractBGGRBR.
  Same argument shows that case 4 is also class \tractBGGRBR.

Suppose $\holts(\mathbf{F}, \mathbf{G})$ is not \sph.
Then, we may assume that $\mathbf{G}$ is of a tractable form in \cref{thm:dich-single-sym-ter-dom3-real}. 
We analyze each case. 

  Suppose $\mathbf{G}$ is of type \ternarytractgeneq, so $\mathbf{G} = \mathbf{u}\teh + \mathbf{v}\teh$ for some $\mathbf{u} = (u_1, u_2, u_3), \mathbf{v} = (v_1, v_2, v_3) \in \mathbb{R}^3$ with $\langle \mathbf{u}, \mathbf{v} \rangle = 0$.
  By \cref{lem:geneq-eqbg} and \cref{lem:domain-restriction}, $\holbs([1, 0, 0, 1], (u_1, u_2) \teh + (v_1, v_2)\teh)$ must be not \sph.
  If $(u_1, u_2)$ and $(v_1, v_2)$ are linearly dependent, then we get case 3.
  If $(u_1, u_2)$ and $(v_1, v_2)$ are linearly independent, then it must be the case that $((u_1, u_2), (v_1, v_2)) \sim (\mathbf{e}_1, \mathbf{e}_2)$ to be tractable by \cref{cor:non-orthogonal-linearly-independent-hard}.
  Without loss of generality, we may assume that $u_2 = 0$ and $v_1 = 0$.
  Then, for $\langle \mathbf{u}, \mathbf{v} \rangle = 0$, we must have $u_3 v_3 = 0$.
  This implies that $\mathbf{G}$ is in case 2.

  Suppose $\mathbf{G}$ is of type \ternarytractz, so 
  $\mathbf{G} = (\mathbf{u} + i \mathbf{v})\teh + (\mathbf{u} - i \mathbf{v})\teh$ for $\mathbf{v} = (v_1, v_2, v_3), \mathbf{u} = (u_1, u_2, u_3) \in \mathbb{R}^3$ such that
  $\langle \mathbf{u}, \mathbf{v} \rangle = 0$ and $\langle \mathbf{u}, \mathbf{u} \rangle = \langle \mathbf{v}, \mathbf{v} \rangle$.
  Again, we may look at $\holbs([1, 0, 0, 1], ((u_1, u_2) + i (v_1, v_2))\teh + ((u_1, u_2) - i (v_1, v_2))\teh)$, which is \sph if $(u_1, u_2)$ and $(v_1, v_2)$ are linearly independent.
  Therefore, they must be linearly dependent and hence case 4.
\end{proof}

\subsection{Rank 3 Type \texorpdfstring{\ternarytractgeneq}{A}}
\begin{lemma}\label{lem:dichotomy-ternary-ternary-rank-3-geneq}
  Let $\mathbf{F} = \mathbf{e}_1\teh + \mathbf{e}_2\teh + \mathbf{e}_3\teh$.
  Let $\mathbf{G}$ be a nondegenerate, real-valued, symmetric, ternary signature.
  Then, $\holts(\mathbf{F}, \mathbf{G})$ is computable in polynomial time if one of the following conditions holds. Otherwise, it is \sph.
  \begin{enumerate}
    \item $\mathbf{G}$ is \geneq.
    \item For some nonzero $c \in \mathbb{R}$, $c \mathbf{G} = \mathbf{e}_i\teh + \mathbf{v} \teh$ for $\mathbf{v} \in \mathbb{R}^3$ such that $\langle \mathbf{e}_i, \mathbf{v} \rangle = 0$.
  \end{enumerate}
\end{lemma}
\begin{proof}
  For tractability, we see that case 1 is class \tractE.
  Case 2 is class \tractBGR.

  Suppose $\holts(\mathbf{F}, \mathbf{G})$ is not \sph.
  Then, we may assume that $\mathbf{G}$ is a tractable signature in \cref{thm:dich-single-sym-ter-dom3-real}, so its rank must be $2$ or $3$.

  Suppose $\mathbf{G}$ is rank $2$.
  We may assume that $\mathbf{G}$ is of one of the tractable cases in \cref{lem:dichotomy-ternary-ternary-rank-2-geneq}, since $\mathbf{F}$ can realize $\mathbf{e}_1\teh + \mathbf{e}_2\teh$ by \cref{lem:geneq-normalization}.
  Case 1 and case 2 are stated here to be tractable.
  
  Suppose $\mathbf{G}$ is case 3.
  By \cref{lem:geneq-normalization}, $\mathbf{F}$ can realize $\mathbf{e}_2\teh + \mathbf{e}_3\teh$, and we may consider $\mathbf{G}\domres{\dg, \dr}$.
  By the same analysis as in the proof of \cref{lem:dichotomy-ternary-ternary-rank-2-geneq}, or using it directly, we see that $(u_2, u_3)$ and $(v_2, v_3)$ must be either $\mathbf{e}_1$ and $\mathbf{e}_2$ up to scaling or linearly dependent.
  If $(u_2, u_3) \sim \mathbf{e}_2$ and $(v_2, v_3) \sim \mathbf{e}_3$, then $u_3 = 0$ and $v_2 = 0$.
  Orthogonality of $\mathbf{u}$ and $\mathbf{v}$ implies $u_1 v_1 = 0$, putting $\mathbf{G}$ in case 2 of this lemma. 
  If $(u_2, u_3)$ and $(v_2, v_3)$ are linearly dependent, then we look at the matrix 
  $M = \begin{bsmallmatrix}
    u_1 & u_2 & u_3 \\
    v_1 & v_2 & v_3
  \end{bsmallmatrix}$.
  $M$ has rank $2$ by assumption, but its two $2 \times 2$ minors corresponding to columns $1$ and $2$ and columns $2$ and $3$ have determinant $0$.
  So, it must be the case that $(u_1, u_3)$ and $(v_1, v_3)$ are linearly independent.
  By \cref{lem:geneq-normalization}, $\mathbf{F}$ can realize $\mathbf{e}_1 \teh + \mathbf{e}_3 \teh$, so we may consider $\mathbf{G}\domres{\db, \dr}$ and apply the same argument.

  Suppose $\mathbf{G}$ is case 4.
  We may apply the above argument, that there must exist some $1 \le i < j \le 3$ such that $(u_i, u_j)$ and $(v_i, v_j)$ are linearly independent.
  Then, the corresponding domain restriction of $\mathbf{G}$ is $\mathbf{G}' = ((u_i, u_j) + i (v_i, v_j))\teh + ((u_i, u_j) - i (v_i, v_j))\teh$, 
  and $\holbs(\mathbf{G}', [1, 0, 0, 1])$ is \sph.

  Suppose $\mathbf{G}$ is rank $3$ type \ternarytractgeneq.
  Then $\mathbf{G} = \mathbf{v}_1\teh + \mathbf{v}_2\teh + \mathbf{v}_3\teh$ for nonzero $\mathbf{v}_1, \mathbf{v}_2, \mathbf{v}_3 \in \mathbb{R}^3$ such that $\langle \mathbf{v}_i, \mathbf{v}_j \rangle = 0$ for $i \ne j$.
  By \cref{lem:geneq-normalization}, we may realize $\mathbf{G}_1 = \mathbf{v}_1\teh + \mathbf{v}_2\teh$.
  Since $\mathbf{G}_1$ is rank $2$, we may apply the above, which implies $\mathbf{G}_1$ is a \geneq or $\mathbf{e}_i\teh + \mathbf{u}\teh$ for some $i$ and $\mathbf{u} \in \mathbb{R}^3$ such that $\langle \mathbf{u}, \mathbf{e}_i \rangle = 0$.
  By \cref{prop:uniqueness-tensor-decomposition}, if $\mathbf{G}_1$ is a \geneq, then $\mathbf{v}_1, \mathbf{v}_2$ are scalar multiples of standard basis vectors.
  By orthogonality, it must be that $\{\mathbf{v}_1, \mathbf{v}_2, \mathbf{v}_3\} \sim \{\mathbf{e}_1, \mathbf{e}_2, \mathbf{e}_3\}$, which implies $\mathbf{G}$ is a \geneq.
  For the other case, assume $i = 1$ without loss of generality, so $\mathbf{G}_1 = \mathbf{e}_1\teh + (0, a, b)\teh$ for $a, b \ne 0$.
  This implies that $\mathbf{v}_3 \sim (0, -b, a)$ by orthogonality, and we may realize $\mathbf{G}_2 = (0, a, b)\teh + (0, -b, a)\teh$.
  However, $\holts(\mathbf{F}, \mathbf{G}_2)$ is \sph by the above.

  Suppose $\mathbf{G}$ is rank $3$ type \ternarytractz.
  We may write $\mathbf{G} = (\mathbf{v}_1 + i \mathbf{v}_2) \teh + (\mathbf{v}_1 - i \mathbf{v}_2)\teh + \mathbf{v}_3\teh$ for nonzero $\mathbf{v}_1, \mathbf{v}_2, \mathbf{v}_3 \in \mathbb{R}^3$ such that $\langle \mathbf{v}_i, \mathbf{v}_j \rangle = 0$ for $i \ne j$.
  By \cref{lem:Z-normalization}, we may realize $\mathbf{G}_1 = (\mathbf{v}_1 + i \mathbf{v}_2)\teh + (\mathbf{v}_1 - i \mathbf{v}_2)\teh$, for which $\holts(\mathbf{F}, \mathbf{G}_1)$ is \sph.
\end{proof}

\subsection{Rank 2 Type \texorpdfstring{\ternarytractz}{B}}
Before the main lemma, we prove a simple proposition about a type $\typeii$ Boolean domain signature.
\begin{proposition}\label{prop:typeii-same-norm}
  Let $\mathbf{F}$ be a real-valued type $\typeii$ Boolean domain signature.
  If $\mathbf{F} = (\mathbf{u} + i \mathbf{v}) \teh + (\mathbf{u} - i \mathbf{v})\teh$ for some linearly independent $\mathbf{u}, \mathbf{v} \in \mathbb{R}^2$,
  then it must be the case that $\langle \mathbf{u}, \mathbf{u} \rangle = \langle \mathbf{v}, \mathbf{v} \rangle$ and $\langle \mathbf{u}, \mathbf{v} \rangle = 0$.
\end{proposition}
\begin{proof}
  We write $\mathbf{u} = (u_1, u_2)$ and $\mathbf{v} = (v_1, v_2)$.
  By assumption, $\mathbf{F}$ has the form $[x, y, -x, -y]$ for some $x, y \in \mathbb{R}$.
  $((u_1, u_2) + i (v_1, v_2))\teh + ((u_1, u_2) - i (v_1, v_2))\teh = [x, y, -x, -y]$ yields the following linear system.
  Let $A = \begin{bsmallmatrix}
  u_1 + i v_1 & u_1 - i v_1 \\
  u_2 + i v_2 & u_2 - i v_2
  \end{bsmallmatrix}$, then
  \[
    A \begin{bmatrix}
      (u_1 + i v_1)^2 \\ (u_1 - i v_1)^2
      \end{bmatrix} = - A \begin{bmatrix}
      (u_2 + i v_2)^2 \\ (u_2 - i v_2)^2
    \end{bmatrix} 
    \, .
  \]
  $A$ is full rank because $(u_1, u_2)$ and $(v_1, v_2)$ are linearly independent, so it must be the case that 
  $(u_1 + i v_1)^2 = - (u_2 + i v_2)^2$.
  This implies that $u_1^2 + u_2^2 = v_1^2 + v_2^2$ and $u_1 v_1 + u_2 v_2 = 0$.
\end{proof}

\begin{lemma}\label{lem:dichotomy-terneray-ternary-rank-2-z}
  Let $\mathbf{F} = \boldsymbol{\beta} \teh + \overline{\boldsymbol{\beta}} \teh$ where $\boldsymbol{\beta} = \frac{1}{\sqrt{2}} (1, i, 0)$.
  Let $\mathbf{G}$ be a nondegenerate, real-valued, symmetric, ternary signature of rank $2$.
  Then, $\holts(\mathbf{F}, \mathbf{G})$ is computable in polynomial time if one of the following conditions holds. Otherwise, it is \sph.
  \begin{enumerate}
    \item $\mathbf{G}$ is $\db \dg | \dr$ and $\mathbf{G}\domres{\db, \dg}$ is compatible with type $\typeii$ signatures.
    \item $\mathbf{G} = \mathbf{u}\teh + \mathbf{v}\teh$ for $\mathbf{v} = (v_1, v_2, v_3), \mathbf{u} = (u_1, u_2, u_3) \in \mathbb{R}^3$ such that $(u_1, u_2)$ and $(v_1, v_2)$ are linearly dependent and $\langle \mathbf{u}, \mathbf{v} \rangle = 0$.
    \item $\mathbf{G} = (\mathbf{u} + i \mathbf{v})\teh + (\mathbf{u} - i \mathbf{v})\teh$ for $\mathbf{v} = (v_1, v_2, v_3), \mathbf{u} = (u_1, u_2, u_3) \in \mathbb{R}^3$ such that $(u_1, u_2)$ and $(v_1, v_2)$ are linearly dependent, $\langle \mathbf{u}, \mathbf{v} \rangle = 0$ and $\langle \mathbf{u}, \mathbf{u} \rangle = \langle \mathbf{v}, \mathbf{v} \rangle$.
  \end{enumerate}
\end{lemma}
\begin{proof}
  For tractability, we see that case 1 is class \tractBGR.
  Case 2 and 3 are class \tractBGGRBR by applying an orthogonal transformation as in the proof of \cref{lem:dichotomy-ternary-ternary-rank-2-geneq}.

  Suppose $\holts(\mathbf{F}, \mathbf{G})$ is not \sph. 
  Then, we may assume that $\mathbf{G}$ is of a tractable form in \cref{thm:dich-single-sym-ter-dom3-real}. 
  We analyze each case. 

  Suppose $\mathbf{G}$ is of type \ternarytractgeneq, so $\mathbf{G} = \mathbf{u}\teh + \mathbf{v}\teh$ for some $\mathbf{u} = (u_1, u_2, u_3), \mathbf{v} = (v_1, v_2, v_3) \in \mathbb{R}^3$ with $\langle \mathbf{u}, \mathbf{v} \rangle = 0$.
  By \cref{lem:z-eqbg} and \cref{lem:domain-restriction}, $\holbs([1, 0, -1, 0], (u_1, u_2)\teh + (v_1, v_2)\teh)$ must not be \sph.
  Therefore, $(u_1, u_2)$ and $(v_1, v_2)$ must be linearly dependent, in which case $\mathbf{G}$ is case 2.
  Note that this includes the case where $\mathbf{G} = (a, b, 0)\teh + \mathbf{e}_3 \teh$, which may also fall under case 1.

  Suppose $\mathbf{G}$ is of type \ternarytractz, so
  $\mathbf{G} = (\mathbf{u} + i \mathbf{v})\teh + (\mathbf{u} - i \mathbf{v})\teh$ for $\mathbf{v} = (v_1, v_2, v_3), \mathbf{u} = (u_1, u_2, u_3) \in \mathbb{R}^3$ such that
  $\langle \mathbf{u}, \mathbf{v} \rangle = 0$ and $\langle \mathbf{u}, \mathbf{u} \rangle = \langle \mathbf{v}, \mathbf{v} \rangle$.
  Then $\mathbf{G}\domres{\db, \dg} = ((u_1, u_2) + i (v_1, v_2))\teh + ((u_1, u_2) - i (v_1, v_2))\teh$.
  By \cref{cor:non-orthogonal-linearly-independent-hard}, $\holbs(\mathbf{G}\domres{\db \dg})$ is \sph unless $(u_1, u_2)$ and $(v_1, v_2)$ are linearly dependent or orthogonal.
  If they are linearly dependent, then $\mathbf{G}$ is case 3.
  If $(u_1, u_2)$ and $(v_1, v_2)$ are linearly independent and orthogonal, then $\mathbf{G}\domres{\db, \dg}$ must be a $\typeii$ signature, so by \cref{prop:typeii-same-norm}, $u_1^2 + u_2^2 = v_1^2 + v_2^2$.
  Recall that $\langle \mathbf{u}, \mathbf{v} \rangle = 0$ and $\langle \mathbf{u}, \mathbf{u} \rangle = \langle \mathbf{v}, \mathbf{v} \rangle$.
  Since $u_1 v_1 + u_2 v_2 = 0$ by assumption, we get $u_3 v_3 = 0$.
  The only way for  $u_1^2 + u_2^2 = v_1^2 + v_2^2$
  and $\langle \mathbf{u}, \mathbf{u} \rangle = \langle \mathbf{v}, \mathbf{v} \rangle$ to both hold is  $u_3 = v_3 = 0$.
  Then, $\supp \mathbf{G} \subseteq \{\db, \dg\}^*$, and thus $\mathbf{G}$ is $\db \dg | \dr$, so it is case 1.
\end{proof}

\subsection{Rank 3 Type \texorpdfstring{\ternarytractz}{B}}
\begin{lemma}\label{lem:dichotomy-terneray-ternary-rank-3-z}
  Let $\mathbf{F} = \boldsymbol{\beta} \teh + \overline{\boldsymbol{\beta}} \teh + \mathbf{e}_3 \teh$ where $\boldsymbol{\beta} = \frac{1}{\sqrt{2}} (1, i, 0)$.
  Let $\mathbf{G}$ be a nondegenerate, real-valued, symmetric, ternary signature.
  Then, $\holts(\mathbf{F}, \mathbf{G})$ is computalbe in polynomial time if $\mathbf{G}$ is $\db \dg | \dr$ and $\mathbf{G}\domres{\db, \dg}$ is compatible with type $\typeii$ signatures.
  Otherwise, it is \sph.
\end{lemma}
\begin{proof}
  If $\mathbf{G}$ satisfies the conditions, then it is in the tractability class \tractBGR.

  Suppose $\holts(\mathbf{F}, \mathbf{G})$ is not \sph.
  Then, we may assume that $\mathbf{G}$ is a tractable signature in \cref{thm:dich-single-sym-ter-dom3-real}, so its rank must be $2$ or $3$.
  
  Suppose $\mathbf{G}$ is rank $2$.
  By \cref{lem:Z-normalization}, since $\mathbf{F}$ can realize $\boldsymbol{\beta}\teh + \overline{\boldsymbol{\beta}}\teh$, we may assume that $\mathbf{G}$ is in one of the tractable cases of \cref{lem:dichotomy-terneray-ternary-rank-2-z}.

  Suppose $\mathbf{G}$ is case 2, so $\mathbf{G} = \mathbf{u}\teh + \mathbf{v}\teh$ for some $\mathbf{u} = (u_1, u_2, u_3), \mathbf{v} = (v_1, v_2, v_3) \in \mathbb{R}^3$ with $\langle \mathbf{u}, \mathbf{v} \rangle = 0$.
  Writing $(u_1, u_2) = (px, py)$ and $(v_1, v_2) = (qx, qy)$ with $x^2 + y^2 = 1$, we apply $\db \dg | \dr$ orthogonal holographic transformation $T = \begin{bsmallmatrix}
    -y & x & 0 \\
    x & y & 0 \\
    0 & 0 & 1
  \end{bsmallmatrix}$.
  Let $\mathbf{G}' = T \teh \mathbf{G} = (0, p, u_3)\teh + (0, q, v_3)\teh $.
  Then, $\supp \mathbf{G}' \subseteq \{\dg, \dr\}^*$.
  By \cref{lem:geneq-eqbg}, $\mathbf{G}'$ can realize $(=_{\dg \dr})$.
  On the other hand, by \cref{cor:z-normalization-orthogonal}, $T\teh \mathbf{F}$ can realize a ternary signature $\mathbf{A}$ such that $\mathbf{A}\domres{\dg, \dr} = [1, 0, 0, 1]$.
  For $\holbs((\mathbf{G}')\domres{\dg, \dr}, \mathbf{A}\domres{\dg, \dr})$ to be tractable, it must be the case that $((p, u_3), (q, v_3)) \sim (\mathbf{e}_1, \mathbf{e}_2)$.
  $p = 0$ or $q = 0$ both imply that $\mathbf{G}$ is $\db \dg | \dr$ and $\mathbf{G}\domres{\db, \dg}$ is degenerate.
  Therefore, we get the tractable case.

  Suppose $\mathbf{G}$ is case 3.
  We apply the same transform as case 2 and get $\mathbf{G}' = T \teh \mathbf{G}$ such that $\supp \mathbf{G}' \subseteq \{\dg, \dr\}^*$ and $(\mathbf{G}')\domres{\db, \dg}$ is of type $\typeii$. By the same argument as above,
  we can realize a signature $\mathbf{A}$ from $T \teh \mathbf{F}$ such that $\mathbf{A}\domres{\dg, \dr} = [1, 0, 0, 1]$.
  Then, $\holbs((\mathbf{G}')\domres{\dg, \dr}, \mathbf{A}\domres{\dg, \dr})$ is \sph. 

  Suppose $\mathbf{G}$ is rank $3$.
  By \cref{lem:dichotomy-ternary-ternary-rank-3-geneq}, if $\mathbf{G}$ is a rank $3$ type \ternarytractgeneq signature, then $\holts(\mathbf{F}, \mathbf{G})$ is \sph.
  Thus, assume $\mathbf{G}$ is type \ternarytractz. 
  We may write $\mathbf{G} = (\mathbf{v}_1 + i \mathbf{v}_2) \teh + (\mathbf{v}_1 - i \mathbf{v}_2)\teh + \mathbf{v}_3\teh$ for nonzero $\mathbf{v}_1, \mathbf{v}_2, \mathbf{v}_3 \in \mathbb{R}^3$ such that $\langle \mathbf{v}_i, \mathbf{v}_j \rangle = 0$ for $i \ne j$.
  By \cref{lem:Z-normalization}, we can realize $\mathbf{G}_1 = (\mathbf{v}_1 + i \mathbf{v}_2)\teh + (\mathbf{v}_1 - i \mathbf{v}_2)\teh$.
  If $\mathbf{v}_3 \not \sim \mathbf{e}_3$, then $\supp \mathbf{G}_1 \not \subseteq \{\db, \dg\}^*$.
  However, this brings us back to the previous case, where there is no tractable case.
  If $\mathbf{v}_3 \sim \mathbf{e}_3$, then $\mathbf{G}$ is $\db \dg | \dr$ and $\mathbf{G}\domres{\db, \dg}$ is $\typeii$, so tractable.
\end{proof}

%% file: sections/higher_arity.tex
\section{A Single Signature Dichotomy}\label{sec:higher-arity}
In this section, we prove a dichotomy of $\holts(\mathbf{F})$ for an  
$\mathbf{F}$ of arbitrary arity $n$. The case $n\le 3$ was proved
in~\cite{cai_dichotomy_2013}. Let $n \ge 4$.
Similar to the Boolean domain case, it is natural to expect that higher arity tractable signatures also have the same tensor decomposition form, i.e. $a\mathbf{e}_1^{\otimes n} + b\mathbf{e}_2^{\otimes n} + c\mathbf{e}_3^{\otimes n}$ or $\boldsymbol{\beta}^{\otimes n} + \overline{\boldsymbol{\beta}}^{\otimes n} + \lambda \mathbf{e}_3^{\otimes n}$ after some orthogonal transformation.
We show that indeed this is the case.

The proof is an induction on the arity.
First, we show that a signature of arity $4$ must be of the same form, using the dichotomy of two ternary signatures.
To do so, we use the fact that the set $\{\langle \mathbf{F}, \mathbf{u} \rangle : \mathbf{u} \in \mathbb{R}^3\}$ is a vector space.
That allows us to add signatures, and we argue that unless the signatures are perfectly aligned, the sum of signatures cannot be tractable using the dichotomy of two ternary signatures of \cref{sec:two-ternary}.
The reason for this roundabout strategy, compared to the proof of the Boolean domain case, is that it is hard to characterize how a non-tractable signature may look like,
since the dichotomy theorems we have are not explicitly defined on the entries of the signatures, but talk about the tensor decomposition form.

Once we show that arity $4$ signatures must have tensor decomposition of  at most $3$ linearly independent vectors, 
then we argue that the dichotomy of two arity $4$ signatures must be essentially same as the dichotomy of the two ternary signatures, since unary signatures allow us to decrease the arity, while preserving the vectors in the tensor decomposition.
Finally, we notice that in the proof of the above two statements, the fact that the arities of the signatures were $3$ and $4$ does not matter, and the proof can be made inductive. 
The logical dependency is visualized in the dashed box of \cref{fig:logical-dependency}.

\subsection{Subspace of Signatures}\label{subsec:subspace-signatures}
Let $\mathbf{F}$ be a real-valued symmetric signature of arity $4$.
Consider the set $\mathscr{F} = \{\langle \mathbf{F}, \mathbf{u} \rangle : \mathbf{u} \in \mathbb{R}^3\}$. 
Note that $\mathbf{u} \mapsto \langle \mathbf{F}, \mathbf{u} \rangle$ is a linear map from $\mathbb{R}^3$ to $\mathsf{S}^3(\mathbb{C}^3)$, the space of complex-valued symmetric signatures of arity $3$.
In particular, $\mathscr{F}$ is a vector space, so it is closed under linear combinations.
Also,  $\mathscr{F}$ only consists of real-valued signatures if $\mathbf{F}$ is a real-valued signature.
\newcommand{\symtersig}{\mathsf{S}^3_{\mathbb{R}}}
Let $\symtersig$ denote the set of all real-valued symmetric ternary signatures.

We show a dichotomy for an arbitrary subspace $\mathscr{F} \subseteq \symtersig$.
For a nonempty $\mathscr{F}$, each ternary signature must be a tractable signature, otherwise
the problem is already \#\P-hard.
We prove this dichotomy for a subspace $\mathscr{F}$
by considering each canonical form such a tractable signature can take
under a holographic transformation $T$. Note that $T \mathscr{F}$ is also a subspace.

\begin{proposition}\label{prop:geneq-sum-hard}
Let $\mathbf{u} \in \mathbb{R}^2$ be a nonzero vector.
Let $\mathbf{F} = \mathbf{v}_1\teh + \mathbf{v}_2 \teh$ be a Boolean domain signature for nonzero $\mathbf{v}_1, \mathbf{v}_2 \in \mathbb{R}^2$ such that $\langle \mathbf{v}_1, \mathbf{v}_2 \rangle = 0$.
Let $\mathbf{G} = \mathbf{F} + \mathbf{u}\teh$.
Then, $\holbs(\mathbf{F}, \mathbf{G})$ is \sph unless $\mathbf{u} \sim \mathbf{v}_1$ or $\mathbf{u} \sim \mathbf{v}_2$.
\end{proposition}
\begin{proof}
Suppose $\mathbf{u} \not\sim \mathbf{v}_1$ and $\mathbf{u} \not \sim \mathbf{v}_2$.
Then, we may apply an orthogonal transformation $T = [\mathbf{v}_1' \, \mathbf{v}_2']\transpose$ where $\mathbf{v}_i' = \frac{1}{\|\mathbf{v}_i\|}\mathbf{v}_i$.
Then, $\mathbf{F}' = T \teh \mathbf{F}  = c_1 \mathbf{e}_1\teh + c_2 \mathbf{e}_2\teh$ for some nonzero $c_1, c_2 \in \mathbb{R}$ and $T \mathbf{u} = (a, b)\transpose$ where $a, b \ne 0$.
For $\mathbf{G}' = T\teh \mathbf{G}$ to be tractable with $\mathbf{F}'$, it must be the case that $\mathbf{G}'$ is degenerate or equal to $d_1 \mathbf{e}_1 \teh + d_2 \mathbf{e}_2$ for some $d_1, d_2 \in \mathbb{C}$.
By \cref{cor:pairiwse-linearly-independent-tensor}, for any $\mathbf{v}_1, \mathbf{v}_2, \mathbf{v}_3, \mathbf{v}_4 \in \mathbb{C}^3$ that are pairwise linearly independent, 
$\mathbf{v}_1\teh, \mathbf{v}_2 \teh, \mathbf{v}_3 \teh, \mathbf{v}_4\teh$ are linearly independent.
Therefore, it cannot be the case that $\mathbf{G}' = c_1 \mathbf{e}_1\teh + c_2 \mathbf{e}_2 \teh + (a, b)\teh = \mathbf{v}\teh$ for some $\mathbf{v} \in \mathbb{C}^3$ whether $\mathbf{v}$ is pairwise linearly independent to the other three or not. 
Also, it cannot be the case that $\mathbf{G} = d_1 \mathbf{e}_1 \teh + d_2 \mathbf{e}_2 \teh$ by a similar reason.
\end{proof}
\begin{proposition}\label{prop:z-sum-hard}
Let $(a, b) \in \mathbb{R}^2$ be a nonzero vector.
Let $\mathbf{F} = \boldsymbol{\beta}\teh + \overline{\boldsymbol{\beta}}\teh$ be a Boolean domain signature where $\boldsymbol{\beta} = \frac{1}{\sqrt{2}}(1, i)\transpose$.
Let $\mathbf{G} = \mathbf{F} + (a, b)\teh$.
Then, $\holbs(\mathbf{F}, \mathbf{G})$ is \sph.
\end{proposition}
\begin{proof}
  The proof is same as \cref{prop:geneq-sum-hard}.
\end{proof}
\begin{proposition}\label{prop:bg-gr-sum-hard}
Let $\mathbf{F}$ and $\mathbf{G}$ be nondegenerate, real valued, symmetric, ternary signatures such that $\supp \mathbf{F} \subseteq \{\db, \dg\}^*$ and $\supp \mathbf{G} \subseteq \{\dg, \dr\}^*$.
Let $\mathbf{H} = \mathbf{F} + \mathbf{G}$.
Then, $\holts(\mathbf{F}, \mathbf{G}, \mathbf{H})$ is \sph unless $\mathbf{F}, \mathbf{G}$ are both \geneq.
\end{proposition}
\begin{proof}
Assume $\holts(\mathbf{F}, \mathbf{G}, \mathbf{H})$ is not \sph.
Then, we may assume that $\holts(\mathbf{F})$ and $\holts(\mathbf{G})$ are tractable.
  We write $\mathbf{F}, \mathbf{G}, \mathbf{H}$ as in \cref{fig:f-g-h}.
\begin{figure}
  \centering
  \begin{subfigure}[b]{0.3\textwidth}
    \centering
  \begin{tikzpicture}[scale=0.6]
  \node at (0, 0) {$a$};
  \node at (-1, -1) {$b$};
  \node at (1, -1) {$0$};
  \node at(-2, -2) {$c$};
  \node at (0, -2) {$0$};
  \node at (2, -2) {$0$};
  \node at (-3, -3) {$d$};
  \node at (-1, -3) {$0$};
  \node at (1, -3) {$0$};
  \node at (3, -3) {$0$};
\end{tikzpicture}
\caption{$\mathbf{F}$}
\end{subfigure}
\hfill
\begin{subfigure}[b]{0.3\textwidth}
  \centering
  \begin{tikzpicture}[scale=0.6]
  \node at (0, 0) {$0$};
  \node at (-1, -1) {$0$};
  \node at (1, -1) {$0$};
  \node at(-2, -2) {$0$};
  \node at (0, -2) {$0$};
  \node at (2, -2) {$0$};
  \node at (-3, -3) {$x$};
  \node at (-1, -3) {$y$};
  \node at (1, -3) {$z$};
  \node at (3, -3) {$w$};
\end{tikzpicture}
\caption{$\mathbf{G}$}
\end{subfigure}
\hfill
\begin{subfigure}[b]{0.35\textwidth}
  \centering
  \begin{tikzpicture}[scale=0.6]
  \node at (0, 0) {$a$};
  \node at (-1, -1) {$b$};
  \node at (1, -1) {$0$};
  \node at(-2, -2) {$c$};
  \node at (0, -2) {$0$};
  \node at (2, -2) {$0$};
  \node at (-3, -3) {$d + x$};
  \node at (-1, -3) {$y$};
  \node at (1, -3) {$z$};
  \node at (3, -3) {$w$};
\end{tikzpicture}
\caption{$\mathbf{H}$}
\end{subfigure}
  \caption{$\mathbf{F}$, $\mathbf{G}$, and $\mathbf{H} = \mathbf{F} + \mathbf{G}$.} \label{fig:f-g-h}
\end{figure}
By \cref{lem:geneq-eqbg} and \cref{lem:z-eqbg}, we may realize $(=_{\db \dg})$ using $\mathbf{F}$ and $(=_{\dg \dr})$ using $\mathbf{G}$, regardless of the types of $\mathbf{F}$ and $\mathbf{G}$.

Suppose $d \ne 0$, then $\holbs(\mathbf{G}\domres{\dg, \dr}, \mathbf{H}\domres{\dg, \dr})$ is \sph by \cref{prop:geneq-sum-hard} and \cref{prop:z-sum-hard} unless $\mathbf{G}$ is \geneq because $\mathbf{H}\domres{\dg, \dr} = d \mathbf{e}_1\teh + \mathbf{G}\domres{\dg, \dr}$.
If $\mathbf{G}$ is \geneq, then $y, z = 0$ and thus it must be the case that $x \ne 0$ because otherwise $\mathbf{G}$ is degenerate.
Then, we may make the same argument on $\holbs(\mathbf{F}\domres{\db, \dg}, \mathbf{H}\domres{\db, \dg})$, which implies that the only way to escape hardness is for $\mathbf{F}$ to be \geneq as well.

Therefore, we may assume that $d, x = 0$.
Then, since $\mathbf{F}$ and $\mathbf{G}$ are assumed to be nondegenerate, $\mathbf{F}$ and $\mathbf{G}$ cannot be \geneq, so $b$ and $c$ cannot both be zero and $y$ and $z$ cannot both be zero.
Consider a unary signature $\mathbf{u} = (p, q, r)$ for $p, q, r$ to be determined later.
We will use the polynomial argument as in \cite{cai_dichotomy_2013}.
Let,
\[
  \mathbf{H}_{\mathbf{u}} = 
  \langle \mathbf{H}, \mathbf{u} \rangle = \begin{bmatrix}
    ap + bq & bp + cq & 0 \\
    bp + cq & cp + yr & yq + zr \\
    0 & yq + zr & zq + wr
  \end{bmatrix} \,.
\]
Regardless of the Boolean domain tractable type of $\mathbf{F}\domres{\db, \dg}$, if $y \ne 0$, we can choose some $p, q, r$ such that $(\mathbf{H}_{\mathbf{u}})\domres{\db, \dg} = [ap + bq, bp + cq, cp + yr]$ is not compatible with $\mathbf{F}\domres{\db, \dg}$ because $bp + cq$ and $cp + yr$ can be made to arbitrary values.
Similarly, if $c \ne 0$, we can choose some $p, q, r$ such that $\holbs(\mathbf{G}\domres{\dg, \dr}, (\mathbf{H}_{\mathbf{u}})\domres{\dg, \dr})$ is \sph.
If $c, y = 0$, then it must be that $b, z \ne 0$.
We have $\mathbf{F}\domres{\db, \dg} = [a, b, 0, 0]$ and $\mathbf{G}\domres{\dg, \dr} = [0, 0, x, y]$.
If $\mathbf{F}\domres{\db, \dg}$ is a $\typei(\alpha, \beta)$ signature, then $\alpha \ne 0$ since $\typei(0, \beta)$ corresponds to \geneq.
However, the recurrence implies $\alpha b + \beta \cdot 0 = \alpha \cdot 0$, so $b = 0$.
This contradicts the assumption that $\mathbf{F}$ is nondegenerate.
If $\mathbf{F}\domres{\db, \dg}$ is a $\typeii$ signature, then we must have $a = 0, b = 0$, again a contradiction.
Therefore, $\holbs(\mathbf{F}\domres{\db, \dg})$ is \sph, contrary to the assumption made in the beginning.
%
%
\end{proof}

\subsubsection{Rank 3 Type \texorpdfstring{\ternarytractgeneq}{A}}
Note that in the following lemmas in \cref{subsec:subspace-signatures}, we do not normalize the individual constants using \cref{lem:geneq-normalization} and \cref{lem:Z-normalization}. This is because the although the normalized signature is \textit{realizable} using $\mathscr{F}$, it may not \textit{belong} to $\mathscr{F}$.
\begin{lemma}\label{lem:dichotomy-subspace-rank-3-geneq}
  Let $\mathscr{F}$ be a subspace of $\symtersig$.
  Suppose $\mathbf{F} = c_1 \mathbf{e}_1\teh + c_2 \mathbf{e}_2\teh + c_3 \mathbf{e}_3 \teh \in \mathscr{F}$ for nonzero $c_1, c_2, c_3 \in \mathbb{R}$.
  Then, $\holts(\mathscr{F})$ is \sph unless every $\mathbf{G} \in \mathscr{F}$ is a \geneq.
\end{lemma}
\begin{proof}
Assume $\holts(\mathscr{F})$ is not \sph.
  Then, we may assume that every $\mathbf{G} \in \mathscr{F}$ is of the tractable form in \cref{lem:dichotomy-ternary-ternary-rank-3-geneq}.
  Without loss of generality, assume $\mathbf{G} = d_1 \mathbf{e}_1 \teh + d_2 (0, a, b)\teh$ for nonzero $a, b, d_1, d_2 \in \mathbb{R}$.
  By assumption, we have $\mathbf{H} = \mathbf{F} + \mathbf{G} = (c_1 + d_1)\mathbf{e}_1\teh + c_2 \mathbf{e}_2\teh + c_3 \mathbf{e}_3 \teh + d_2(0, a, b)\teh \in \mathscr{F}$.
  By \cref{lem:geneq-eqbg}, $\mathbf{F}$ can realize $(=_{\dg \dr})$, and by \cref{lem:domain-restriction}, we may look at the domain restrictions $\mathbf{F}\domres{\db, \dr} = [c_2, 0, 0, c_3]$ and $\mathbf{H}\domres{\db, \dr} = c_2 \mathbf{e}_1\teh + c_3 \mathbf{e}_2 \teh + d_2(a, b)\teh$.
  By \cref{prop:geneq-sum-hard}, $\holbs(\mathbf{F}\domres{\dg, \dr}, \mathbf{H}\domres{\dg, \dr})$ is \sph.
  By a similar argument, if $\mathbf{G}$ is a degenerate signature not a scalar multiple $\mathbf{e}_i \teh$, we may restrict to a subdomain and show hardness. 
\end{proof}

\subsubsection{Rank 3 Type \texorpdfstring{\ternarytractz}{B}}
\begin{lemma}\label{lem:dichotomy-subspace-rank-3-z}
  Let $\mathscr{F}$ be a subspace of $\symtersig$.
  Suppose $\mathbf{F} = c(\boldsymbol{\beta}\teh + \overline{\boldsymbol{\beta}}\teh ) + \lambda \mathbf{e}_3 \teh \in \mathscr{F}$ where $\boldsymbol{\beta} = \frac{1}{\sqrt{2}}(1, i, 0)\transpose$ and nonzero $c, \lambda \in \mathbb{R}$.
  Then, $\holts(\mathscr{F})$ is \sph unless every $\mathbf{G} \in \mathscr{F}$ is such that
  $\mathbf{G}$ is $\db \dg | \dr$ and $\mathbf{G}\domres{\db, \dg} = [x, y, -x, -y]$ for some $x, y \in \mathbb{R}$.
\end{lemma}
\begin{proof}
Assume $\holts(\mathscr{F})$ is not \sph.
  Then, we may assume that every $\mathbf{G} \in \mathscr{F}$ is of the tractable form in \cref{lem:dichotomy-terneray-ternary-rank-3-z}.
  Only cases we need to consider are if $\mathbf{G}$ is degenerate or $\db \dg | \dr$ and degenerate on $\{\db, \dg\}$.
  Suppose $\mathbf{G} = (a, b, 0)\teh + c \mathbf{e}_3 \teh$ for some $a, b, c \in \mathbb{R}$ such that not both of $a, b$ are zero.
  By assumption, $\mathbf{H} = \mathbf{F} + \mathbf{G} \in \mathscr{F}$.
  By \cref{lem:z-eqbg}, $\mathbf{F}$ can realize $(=_{\db \dg})$, and by \cref{lem:domain-restriction}, we may look at $\mathbf{H}\domres{\db, \dg} = c (1, i)\teh + c (1, -i)\teh + (a, b)\teh$.
  By \cref{prop:z-sum-hard}, $\holbs(\mathbf{F}\domres{\dg, \dr}, \mathbf{H}\domres{\dg, \dr})$ is \sph.

  Suppose $\mathbf{G}$ is degenerate, so $\mathbf{G} = (a, b, c)\teh$ for some $a, b, c \in \mathbb{R}$.
  If $a, b = 0$, then $\mathbf{G}$ is of the form in the lemma statement.
  Otherwise, we may apply the same argument as the previous case.
\end{proof}

\subsubsection{Rank 2 Type \texorpdfstring{\ternarytractgeneq}{A}}
\begin{lemma}\label{lem:dichotomy-subspace-rank-2-geneq}
Let $\mathscr{F}$ be a subspace of $\symtersig$ such that all signatures in $\mathscr{F}$ are of rank at most $2$.
  Suppose $\mathbf{F} = c_1 \mathbf{e}_1 \teh + c_2 \mathbf{e}_2\teh \in \mathscr{F}$ for nonzero $c_1, c_2 \in \mathbb{R}$.
  Then, $\holts(\mathscr{F})$ is \sph unless every $\mathbf{G} \in \mathscr{F}$ is $d_1 \mathbf{e}_1 \teh + d_2 \mathbf{e}_2\teh$ for some $d_1, d_2 \in \mathbb{R}$.
\end{lemma}
\begin{proof}
Assume $\holts(\mathscr{F})$ is not \sph.
  Then, we may assume that every $\mathbf{G} \in \mathscr{F}$ is of the tractable form in \cref{lem:dichotomy-ternary-ternary-rank-2-geneq}.
  Suppose $\mathbf{G} = \mathbf{v}\teh$ for $\mathbf{v} \in \mathbb{R}^3$ such that $\mathbf{v}$ is not a scalar multiple of $\mathbf{e}_1$ or $\mathbf{e}_2$.
  If $\mathbf{e}_1, \mathbf{e}_2, \mathbf{v}$ are linearly independent, then $\mathbf{F} + \mathbf{G}$ has rank $3$ by \cref{lem:linearly-independent-tensor-rank}, contradicting the assumption that $\mathscr{F}$ does not contain a rank $3$ signature.
  Otherwise, $\mathbf{v} = (a, b, 0)$ for nonzero $a, b \in \mathbb{R}$, and we get hardness by \cref{prop:geneq-sum-hard}.

  Suppose $\mathbf{G} = d \mathbf{e}_i \teh + \mathbf{v}\teh$ for some $\mathbf{v} \in \mathbb{R}^3$ such that $\langle \mathbf{e}_i, \mathbf{v} \rangle = 0$ and $\mathbf{v}$ is not a scalar multiple of any $\mathbf{e}_j$.
  If $i = 3$, then we may apply \cref{prop:geneq-sum-hard} on $(\mathbf{F} + \mathbf{G})\domres{\db, \dg}$.
  If $i = 1$, then $\mathbf{H} = \mathbf{F} + t \mathbf{G} = (c_1 + td) \mathbf{e}_1 + c_2 \mathbf{e}_2 + t \mathbf{v}\teh \in \mathscr{F}$, and we may choose a nonzero $t$ such that $c_1 + td \ne 0$.
  By assumption, $\mathbf{v} \not \sim \mathbf{e}_2$ and is orthogonal to $\mathbf{e}_1$, which means $\mathbf{v} = (0, c, d)$ for nonzero $d$.
  In particular, $\mathbf{e}_1, \mathbf{e}_2, \mathbf{v}$ are linearly independent, and thus $\mathbf{H}$ has rank $3$ by \cref{lem:linearly-independent-tensor-rank}, a contradiction.

  Suppose $\mathbf{G}$ is of case 3 or 4 in \cref{lem:dichotomy-ternary-ternary-rank-2-geneq}.
  In both cases, we may apply a $\db \dg | \dr$ orthogonal holographic transformation $T$ such that for $\mathbf{F}' = T \teh \mathbf{F}$ and $\mathbf{G}' = T \teh \mathbf{G}$, we have $\supp \mathbf{F}' \subseteq \{\db, \dg\}^*$ and $\supp \mathbf{G}' \subseteq \{\dg, \dr\}^*$.
The set $T \mathscr{F}$ is also a vector space, so $\mathbf{H}' = \mathbf{F}' + \mathbf{G}' \in T \mathscr{F}$.
  Then, by \cref{prop:bg-gr-sum-hard}, $\holts(\mathbf{F}', \mathbf{G}', \mathbf{H}')$ is \sph unless $\mathbf{F}'$ and $\mathbf{G}'$ are \geneq.
  But in that case, $\mathbf{F}' + t \mathbf{G}' \in T \mathscr{F}$ is a rank $3$ \geneq for some $t \in \mathbb{R}$.
  Since holographic transformation preserves the rank, we get a contradiction since we started by assuming that $\mathscr{F}$ does not contain a rank $3$ signature.
  \end{proof}

\subsubsection{Rank 2 Type \texorpdfstring{\ternarytractz}{B}}
\begin{lemma}\label{lem:dichotomy-subspace-rank-2-z}
  Let $\mathscr{F}$ be a subspace of $\symtersig$ such that all signatures in $\mathscr{F}$ are of rank at most $2$.
  Suppose $\mathbf{F} = \boldsymbol{\beta}\teh + \overline{\boldsymbol{\beta}}\teh \in \mathscr{F}$ where $\boldsymbol{\beta} = \frac{1}{\sqrt{2}}(1, i, 0)\transpose$.
  Then, $\holts(\mathscr{F})$ is \sph unless every $\mathbf{G} \in \mathscr{F}$ is such that
  $\supp \mathbf{G} \subseteq \{\db, \dg\}^*$ and $\mathbf{G}\domres{\db, \dg} = [x, y, -x, -y]$ for some $x, y \in \mathbb{R}$.
\end{lemma}
\begin{proof}
Assume $\holts(\mathscr{F})$ is not \sph.
  Then, we may assume that every $\mathbf{G} \in \mathscr{F}$ is of the tractable form in \cref{lem:dichotomy-terneray-ternary-rank-2-z}.
  Suppose $\mathbf{G} = \mathbf{v}\teh$ for nonzero $\mathbf{v} \in \mathbb{R}^3$.
  Then, if $\mathbf{v} \not \sim \mathbf{e}_3$, then by \cref{prop:z-sum-hard}, $\holts(\mathbf{F}, \mathbf{F} + \mathbf{G})$ is \sph.
  If $\mathbf{v} \sim \mathbf{e}_3$, then $\mathbf{F} + \mathbf{G}$ has rank 3, contrary to the assumption.

  Suppose $\mathbf{G}$ is not degenerate.
  If $\mathbf{G}$ is rank $2$, $\db \dg | \dr$ and compatible with $\typeii$ on the domain $\{\db, \dg\}$, there are two possibilities:
  $\mathbf{G}$ is $(a, b, 0)\teh + c \mathbf{e}_3\teh$ for some $a, b, c \in \mathbb{R}$ or $\supp \mathbf{G} \subseteq \{\db, \dg\}^*$ and $\mathbf{G}\domres{\db, \dg} = [x, y, -x, -y]$ for some $x, y \in \mathbb{R}$. 
  The second case is the form stated in the lemma, while the first case is \sph by \cref{prop:z-sum-hard}.

  For $\mathbf{G}$ of the case 2 and 3 in \cref{lem:dichotomy-terneray-ternary-rank-2-z}, we may use the same argument as in the proof of \cref{lem:dichotomy-subspace-rank-2-geneq} by using \cref{prop:bg-gr-sum-hard} to derive \#\P-hardness.
\end{proof}

\subsubsection{Rank 1}
\begin{lemma}\label{lem:subspace-rank-1}
Let $\mathscr{F}$ be a subspace of $\symtersig$ such that all signatures in $\mathscr{F}$ are of rank at most $1$.
Then, there exists some $\mathbf{v} \in \mathbb{R}^3$ such that $\mathbf{F} = \{\lambda \mathbf{v}\teh : \lambda \in \mathbb{R}\}$.
\end{lemma}
\begin{proof}
  If $\mathscr{F} = \{\mathbf{0}\}$, then we are done.
  Suppose there exists some nonzero $\mathbf{v} \in \mathbb{R}^3$ such that $\mathbf{v}\teh \in \mathscr{F}$.
  If there exists some $\mathbf{u} \in \mathbb{R}^3$ such that $\mathbf{u} \not \sim \mathbf{v}$ and $\mathbf{u}\teh \in \mathscr{F}$, then $\mathbf{u}\teh + \mathbf{v}\teh \in \mathscr{F}$ but $\mathbf{u}\teh + \mathbf{v}\teh$ has rank $2$ by \cref{lem:linearly-independent-tensor-rank}. 
  Therefore, any $\mathbf{u}$ such that $\mathbf{u}\teh \in \mathscr{F}$ must be such that $\mathbf{u} \sim \mathbf{v}$, proving the statement.
\end{proof}

\subsection{A Single Signature of Arity 4}
\begin{lemma}\label{lem:dichotomy-single-arity-4}
Let $\mathbf{F}$ be a nondegenerate, real-valued, symmetric signature of arity $4$.
Then, $\holts(\mathbf{F})$ is computable in polynomial time if there exists some real orthogonal matrix $T$ such that one of the following conditions holds.
Otherwise, it is \sph.
\begin{enumerate}
  \item $T^{\otimes 4} \mathbf{F} = a \mathbf{e}_1^{\otimes 4} + b \mathbf{e}_2^{\otimes 4} + c \mathbf{e}_3^{\otimes 4}$ for some $a, b, c \in \mathbb{R}$.
  \item $T^{\otimes 4}\mathbf{F} = \boldsymbol{\beta}^{\otimes 4} + \overline{\boldsymbol{\beta}}^{\otimes 4} + \lambda \mathbf{e}_3^{\otimes 4}$ where $\boldsymbol{\beta} = \frac{1}{\sqrt{2}}(1, i, 0)\transpose$ for some $\lambda \in \mathbb{R}$.
\end{enumerate}
\end{lemma} 
\begin{proof}
  The tractability is immediate.

Assume $\holts(\mathbf{F})$ is not \sph.
Let $\mathscr{F} = \{\langle \mathbf{F}, \mathbf{u} \rangle : \mathbf{u} \in \mathbb{R}^3\}$.
Then, $\mathscr{F}$ is a vector space and $\holts(\mathscr{F}) \le_T \holts(\mathbf{F})$, so we may apply the lemmas that we have proven in \cref{subsec:subspace-signatures}.
We may assume that for any $\mathbf{G} \in \mathscr{F}$, $\holts(\mathbf{G})$ is tractable.
Fix any $\mathbf{G} \in \mathscr{F}$.
We may apply a real orthogonal holographic transformation $T$ such that $\mathbf{G}' = T\teh \mathbf{G}$ is of the canonical form.
Note that $\mathbf{G}' \in \mathscr{F}' = T \mathscr{F}$, and $\mathscr{F}'$ is a subspace of $\symtersig$.

Let $\mathbf{F}' = T^{\otimes 4} \mathbf{F}$.
We will look at the following representation of $\mathbf{F}'$.
\begin{figure}
\centering
\begin{minipage}[t]{0.5\linewidth}
\centering
  \begin{tikzpicture}[scale=0.5, every node/.style={scale=1}]
  \node at (0, 0) {$a$};
  \node at (-1, -1) {$b$};
  \node at (1, -1) {$c$};
  \node at(-2, -2) {$d$};
  \node at (0, -2) {$e$};
  \node at (2, -2) {$f$};
  \node at (-3, -3) {$g$};
  \node at (-1, -3) {$h$};
  \node at (1, -3) {$j$};
  \node at (3, -3) {$k$};
  \node at (-4, -4) {$l$};
  \node at (-2, -4) {$m$};
  \node at (0, -4) {$n$};
  \node at (2, -4) {$o$};
  \node at (4, -4) {$p$};

  \draw[dashed] (-1, 0) -- (3.5, -4.5) -- (-5.5, -4.5) -- cycle;

  \draw (0, 1.4) -- (6.1, -4.7) -- (-6.1, -4.7) -- cycle;

\end{tikzpicture}
\caption{$\mathbf{F}'$ and $\langle \mathbf{F}', \mathbf{e}_2 \rangle$} \label{fig:arity-4-signature}
\end{minipage}
\hfill
\begin{minipage}[t]{0.4\linewidth}
\centering
  \begin{tikzpicture}[scale=0.5, every node/.style={scale=0.7}]
  \node at (0, 0) {$*$};
  \node at (-1, -1) {$*$};
  \node at (1, -1) {$0$};
  \node at(-2, -2) {$*$};
  \node at (0, -2) {$0$};
  \node at (2, -2) {$0$};
  \node at (-3, -3) {$*$};
  \node at (-1, -3) {$0$};
  \node at (1, -3) {$0$};
  \node at (3, -3) {$*$};
\end{tikzpicture}
\caption{A generic $\db \dg | \dr$ signature} \label{fig:generic-bg-r}
\end{minipage}

\end{figure}
We categorize the entries in  \cref{fig:arity-4-signature} in the following way.
A \textit{corner entry} is an entry at the corner of the triangle: $a, l, p$.
An \textit{outer entry} is an entry at the perimeter of the triangle: 
$a,b,d,g,l,m,n,o,p,k,f,c$.
An \textit{inner entry} is an entry at the inside of the triangle: $e, h , j$.
The entries of a subsignature correspond to a smaller 
triangle. 
For example, the signature $\langle \mathbf{F}', \mathbf{e}_2 \rangle$ is  
the  triangle given by  the dashed lines in \cref{fig:arity-4-signature}. 
Note that $\langle \mathbf{F}', \mathbf{e}_i \rangle \in \mathscr{F}'$ for all $1 \le i \le 3$.

Suppose $\mathbf{G}'$ is rank $3$ type \ternarytractgeneq.
Then, by \cref{lem:dichotomy-subspace-rank-3-geneq}, $\mathscr{F}'$ only contains \geneq signatures.
We claim that $\mathbf{F}'$ must be a \geneq.
The idea is simple. 
Note that a \geneq signature has nonzero value at only the corner entries.
The three subsignatures, $\langle \mathbf{F}', \mathbf{e}_i \rangle$ for $1 \le i \le 3$ must all be a \geneq.
For that to be possible, all internal entries of $\mathbf{F}'$ must be $0$, since otherwise, we have some subsignature with a nonzero value at a non-corner entry.
Similarly, all outer entries of $\mathbf{F}'$ except the corner must be $0$ as well.
Therefore, $\mathbf{F}'$ is a \geneq.

Suppose $\mathbf{G}'$ is rank $3$ type \ternarytractz.
Then, by \cref{lem:dichotomy-subspace-rank-3-z}, $\mathscr{F}'$ only contains $\db \dg | \dr$ signatures with $\typeii$ on the domain $\{\db, \dg\}$.
If an internal entry is nonzero, then there is some subsignature that is not $\db \dg | \dr$ (see \cref{fig:generic-bg-r}).
Also, $\mathbf{F}'$ cannot have nonzero entries at the outer entries $c, f, k, m, n, o$, because that means there is a subsignature that is not $\db \dg | \dr$.
Therefore, all nonzero entries must be at the domain $\{\db, \dg\}$ part and $p$.
Further, $[a, b, d, g, l]$ must satisfy the type $\typeii$ recurrence.
This is because $\langle \mathbf{F}', \mathbf{e}_1 \rangle \in \mathscr{F}'$,
and thus $(\langle \mathbf{F}', \mathbf{e}_1 \rangle)\domres{\db, \dg} = [a, b, d, g]$ satisfies $a = -d$ and $b = -g$, and similarly, $[b, d, g, f]$ satisfies $b = -g$ and $d = -f$.

Now, suppose $\mathscr{F}'$ does not contain any signature of rank $3$.
If there is a signature of rank $2$, we may apply the same arguments by using \cref{lem:dichotomy-subspace-rank-2-geneq} and \cref{lem:dichotomy-subspace-rank-2-z}.
The only difference is that $p$ must be $0$.

Suppose $\mathscr{F}'$ only contains rank $1$ signatures.
Then, by \cref{lem:subspace-rank-1}, $\mathscr{F}' = \{\lambda \mathbf{v}\teh : \lambda \in \mathbb{R}\}$ for some $\mathbf{v} \in \mathbb{R}^3$.
We claim that $\mathbf{F}' = c \mathbf{v}^{\otimes 4}$ for some $c$, so degenerate.
Suppose $\langle \mathbf{F}', \mathbf{e}_i \rangle = \lambda_i \mathbf{v} \teh$.
Write $\mathbf{v} = (x, y, z)$.
If the corner entries are all $0$, then since each of the subsignatures are degenerate, we must have $\mathbf{F}' = 0$.
This is because, for example, $a = 0$ and $\langle \mathbf{F}', \mathbf{e}_1 \rangle$ being degenerate imply that $a, b, c, d, e, f = 0$.
So, the size $3$ subtriangles at the corners are all $0$, making all the entries of $\mathbf{F}'$ to be $0$.
Therefore, without loss of generality, assume that $a \ne 0$, which implies $\lambda_1, x \ne 0$.
If $b = 0$, then the triangle with corners $b, g, j$ must be all $0$ since $\langle \mathbf{F}', \mathbf{e}_2 \rangle$ is degenerate.
Then, it cannot be the case that $l \ne 0$ because $g = \lambda_1 y^3 = 0$ implies $y = 0$ so $l = \lambda_2 y^3 = 0$.
Therefore, $m, n, o = 0$ as well.
If $c = 0$, then, by the same argument, we must have $c, f, k, p = 0$, so the only nonzero entry is at $a$, and thus $\mathbf{F}'$ is degenerate.
If $c \ne 0$, then $\lambda_3, z \ne 0$ and $\lambda_3 x^3 = \lambda_1 x^2 z$.
So, $\lambda_3 = \lambda_1 z/x$.
We can see that $\mathbf{F}' = \frac{\lambda_1}{x} (x,0, z)^{\otimes 4}$ gives consistent result since $p = \lambda_3 z^3 = \lambda_1 z^4/x$ from $\langle \mathbf{F}', \mathbf{e}_3 \rangle = \lambda_3 \mathbf{v}\teh$.

Suppose $b \ne 0$.
Then, $\lambda_2, y \ne 0$ and $\lambda_2 x^3 = \lambda_1 x^2 y$, so $\lambda_2 = \lambda_1 y/x$.
If $c = 0$, we get the case symmetric as the above when we considered $b = 0$ and $c \ne 0$.
If $c \ne 0$, again we have $\lambda_3 = \lambda_1 z/x$.
Then, we have $\mathbf{F}' = \frac{\lambda_1}{x}(x, y, z)^{\otimes 4}$.
We can check the bottom row gives consistent result by the same calculation as above.
\end{proof}

\subsection{Arbitrary Arity}\label{subsec:arbitrary-arity}
\begin{proposition}\label{prop:geneq-get-all-arity}
  Let $\mathbf{v}_1, \mathbf{v}_2, \mathbf{v}_3 \in \mathbb{R}^3$ be pairwise orthogonal vectors.
  Let $\mathbf{F}_k = \mathbf{v}_1^{\otimes k} + \mathbf{v}_2^{\otimes k}  + \mathbf{v}_3^{\otimes k}$.
  Then, for any $k \ge 3$ and any set of signatures $\mathscr{F}$,
  \[
    \holts(\mathscr{F} \cup \{\mathbf{F}_k\}) =_T \holts(\mathscr{F} \cup \{\mathbf{F}_i : i \ge 3\}) \, .
  \]
\end{proposition}
\begin{proof}
  Suppose all $\mathbf{v}_i$ are nonzero.
  Then, since they are pairwise orthogonal, they are linearly independent as well, so 
  for any $c_1, c_2, c_3 \in \mathbb{R}$,
  there exists some $\mathbf{u} \in \mathbb{R}^3$ such that $\langle \mathbf{v}_i, \mathbf{u} \rangle = c_i$ for all $i$.
  We may connect $k - 3$ copies of $\mathbf{u}$ to $\mathbf{F}_k$ to obtain $\mathbf{F}_3$, so $\holts(\mathscr{F} \cup \{\mathbf{F}_3\}) \le_T \holts(\mathscr{F} \cup \{\mathbf{F}_k\})$.
  By the construction in \cref{lem:geneq-normalization}, $\mathbf{F}_3$ can be used to obtain $\mathbf{F}_i$ for all $i \ge 3$, proving the claim.

  A similar argument works also when $\mathbf{v}_i = 0$ for some $i$.
\end{proof}

\begin{proposition}\label{prop:z-get-all-arity}
  Let $\boldsymbol{\beta} = \frac{1}{\sqrt{2}}(1, i, 0)\transpose$ and $\lambda \in \mathbb{R}$.
  Let $\mathbf{F}_k = \boldsymbol{\beta}^{\otimes k} + \overline{\boldsymbol{\beta}}^{\otimes k} + \lambda \mathbf{e}_3^{\otimes k}$.
  Then, for any $k \ge 3$ and any set of signatures $\mathscr{F}$,
  \[
    \holts(\mathscr{F} \cup \{\mathbf{F}_k\}) =_T \holts(\mathscr{F} \cup \{\mathbf{F}_i : i \ge 3\}) \, .
  \]
\end{proposition}
\begin{proof}
  We choose $\mathbf{u} = (1, 0, 1)$ to realize $\mathbf{F}_3$ from $\mathbf{F}_k$.
  Then, we use the construction in \cref{lem:Z-normalization} to realize $\mathbf{F}_i$ for all $i \ge 3$.
\end{proof}

We state a meta claim about the dichotomy theorems of two signatures of same arity.
\begin{claim}\label{lem:meta-push-arity}
  Let $\mathbf{F}, \mathbf{G}$ be a pair of real-valued symmetric arity $4$ signatures.
  Then, the dichotomy statement of $\holts(\mathbf{F}, \mathbf{G})$ is same as that of two ternary signatures in \cref{lem:dichotomy-ternary-ternary-rank-2-geneq,lem:dichotomy-ternary-ternary-rank-3-geneq,lem:dichotomy-terneray-ternary-rank-2-z,lem:dichotomy-terneray-ternary-rank-3-z}, after changing the tensor powers from $3$ to $4$.
\end{claim}
\begin{proof}
  By \cref{lem:dichotomy-single-arity-4}, after a holographic transformation, any nondegenerate tractable arity $4$ signature has tensor decomposition as
  $a\mathbf{e}_1^{\otimes 4} + b\mathbf{e}_2^{\otimes 4} + c\mathbf{e}_3^{\otimes 4}$ or $(1, i, 0)^{\otimes 4} + (1, -i, 0)^{\otimes 4} + \lambda \mathbf{e}_3^{\otimes 4}$ for some $a, b, c, \lambda \in \mathbb{R}$.
  By \cref{prop:geneq-get-all-arity} and \cref{prop:z-get-all-arity}, we may realize ternary signatures $\mathbf{F}_3$ and $\mathbf{G}_3$ that has the same tensor decomposition vectors.
  Then, we may apply the dichotomy of two ternary signatures on $\mathbf{F}_3$ and $\mathbf{G}_3$, which determines the vectors tensor decomposition of $\mathbf{F}$ and $\mathbf{G}$ up to a scalar.
\end{proof}

One can inductively prove the following using the arguments in \cref{subsec:subspace-signatures} and \cref{lem:dichotomy-single-arity-4}:
\begin{lemma}\label{lem:dichotomy-single-arity-n}
Let $\mathbf{F}$ be a nondegenerate real-valued symmetric signature of arity $n \ge 3$.
Then, $\holts(\mathbf{F})$ is tractable if there exists some real orthogonal matrix $T$ such that one of the following conditions holds.
\begin{enumerate}
  \item $T^{\otimes n} \mathbf{F} = a \mathbf{e}_1^{\otimes n} + b \mathbf{e}_2^{\otimes n} + c \mathbf{e}_3^{\otimes n}$ for some $a, b, c \in \mathbb{R}$.
  \item $T^{\otimes n}\mathbf{F} = \boldsymbol{\beta}^{\otimes n} + \overline{\boldsymbol{\beta}}^{\otimes n} + \lambda \mathbf{e}_3^{\otimes n}$ where $\mathbf{\beta} = \frac{1}{\sqrt{2}}(1, i, 0)\transpose$ for some $\lambda \in \mathbb{R}$.
\end{enumerate}
Otherwise $\holts(\mathbf{F})$ is \sph.
\end{lemma}
From now on, for notational convenience, we may only consider signatures of arity $2$ and $3$ when proving a dichotomy of a set of signatures $\mathcal{F}$. 
\begin{corollary}\label{cor:replace-arity-3}
    Let $\mathcal{F}$ be an arbitrary set of nondegenerate, real-valued, symmetric signatures.
    Suppose that for all $\mathbf{F} \in \mathcal{F}$, the arity of $\mathbf{F}$ is $\ge 3$ and $\holts(\mathbf{F})$ is tractable.
    Then, there exists a set $\mathcal{F}'$ such that $\mathcal{F}'$ consists of nondegenerate, real-valued, symmetric signatures of arity $3$ and 
    $\holts(\mathcal{F}') =_T \holts(\mathcal{F})$.
\end{corollary}
\begin{proof}
By assumption, any signature $\mathbf{F} \in \mathcal{F}$ of arity $n \ge 3$ is of the tractable form in \cref{lem:dichotomy-single-arity-n}.
Then, by \cref{prop:geneq-get-all-arity} and \cref{prop:z-get-all-arity}, 
$\holts(\mathcal{F} \cup \{\mathbf{F}_3\} - \{\mathbf{F}\}) =_T \holts(\mathcal{F})$, where $\mathbf{F}_3$ is the ternary signature that has the same vectors in a tensor decomposition form as $\mathbf{F}$.
Therefore, we may obtain a new set $\mathcal{F}'$ such that $\mathcal{F}'$ consists of signatures of arity $3$, and $\holts(\mathcal{F}') =_T \holts(\mathcal{F})$.
\end{proof}

%% file: sections/set_of_signatures.tex
\section{Set of Signatures}\label{sec:set-of-signatures}
We are now almost ready to prove the main theorem of the paper, a complexity dichotomy 
for an arbitrary set of real-valued symmetric signatures, \cref{thm:dichotomy-set-of-domain-3}.

\begin{proposition}\label{prop:normalize-lin-dep-orthogonal-vectors}
  Let
  $\mathbf{u} = (u_1, u_2, u_3), \mathbf{v} = (v_1, v_2, v_3) \in \mathbb{R}^3$ be nonzero vectors such that $(u_1, u_2)$ and $(v_1, v_2)$ are linearly dependent and $\langle u, v \rangle = 0$.
  Then, $\mathbf{u} \sim (px, py, -q)$ and $\mathbf{v} \sim (qx, qy, p)$ for some $p, q, x, y$ such that $x^2 + y^2 = 1$ and $p^2 + q^2 = 1$.
\end{proposition}
\begin{proof}
  We write $(u_1, u_2) = (ax, ay)$ and $(v_1, v_2) = (bx, by)$ for some $a, b, x, y \in \mathbb{R}$.
  It cannot be the case that $x, y = 0$ since then $\mathbf{u}$ and $\mathbf{v}$ are linearly dependent nonzero vectors, which cannot be orthogonal.
  We may choose $x, y$ to be such that $x^2 + y^2 = 1$ by absorbing the scalar to $a, b$.
  By orthogonality, we have $ab(x^2 + y^2) + u_3v_3 = ab + u_3 v_3 = 0$.
  Thus, $(a, u_3)$ and $(b, v_3)$ are orthogonal vectors, which we may normalize to $(p, -q) \sim (a, u_3)$ and $(q, p) \sim (b, v_3)$.
  Therefore, $\mathbf{u} \sim (px, py, -q)$ and $\mathbf{v} \sim (qx, qy, p)$.
\end{proof}

\begin{proposition}\label{prop:dichotomy-3-signatures-geneq-b-gr}
  Let $\mathbf{F} = \mathbf{e}_1\teh + \mathbf{e}_2\teh$ and $\mathbf{G} = \mathbf{e}_1\teh + (0, a, b)\teh$ for nonzero $a, b \in \mathbb{R}$.
  Let $\mathbf{H}$ be a nondegenerate, real-valued, symmetric, ternary signature of rank $2$.
  Then, $\holts(\mathbf{F}, \mathbf{G}, \mathbf{H})$ is tractable if $\mathbf{H}$ is $\bdgr$.
  Otherwise, it is \sph.
\end{proposition}
\begin{proof}
  If $\mathbf{H}$ is $\bdgr$, then $\holts(\mathbf{F}, \mathbf{G}, \mathbf{H})$ is in the tractable class \tractBGR.

  By \cref{lem:geneq-normalization}, we may assume that $a^2 + b^2 = 1$. 
  Assume $\holts(\mathbf{F}, \mathbf{G}, \mathbf{H})$ is not \sph.
  Then, 
  we may assume that $\mathbf{H}$ is of the tractable form in \cref{lem:dichotomy-terneray-ternary-rank-2-z} since $\holts(\mathbf{F}, \mathbf{H})$ must not be \sph.
  \begin{enumerate}
    \item If $\mathbf{H}$ is \geneq, then it is $\bdgr$.
    \item Suppose $\mathbf{H}$ is $\mathbf{e}_i\teh + \mathbf{v}\teh$ for some $\mathbf{v} \in \mathbb{R}^3$ such that $\langle \mathbf{v}, \mathbf{e}_i \rangle = 0$ and $\mathbf{v} \not \sim \mathbf{e}_j$ for all $j$.
      If $i = 1$, then $\mathbf{H}$ is $\bdgr$.
      Otherwise, assume without loss of generality that $i = 3$, so $\mathbf{H} = \mathbf{e}_3\teh + (c, d, 0)\teh$ for some nonzero $c, d \in \mathbb{R}$.
      We apply the orthogonal holographic transformation $T = \begin{bsmallmatrix}
        1 & 0 & 0 \\
        0 & a & b \\
        0 & -b & a
      \end{bsmallmatrix}$, and get $T \teh \mathbf{G} = \mathbf{e}_1\teh + \mathbf{e}_2\teh$ and $T \teh \mathbf{H} = (c, ad, -bd)\teh + (0, b, a)\teh$.
      Since $a, b, c, d \ne 0$, $\holbs((T \teh \mathbf{H})\domres{\db, \dg})$ is \sph.
    \item By \cref{prop:normalize-lin-dep-orthogonal-vectors}, we may assume that $\mathbf{H} = (px, py, -q)\teh + (qx, qy, p)\teh$ with $x^2 + y^2 = p^2 + q^2 = 1$.
      We apply the orthogonal holographic transformation $T = \begin{bsmallmatrix}
        -y & x & 0 \\
        x & y & 0 \\
        0 & 0 & 1
      \end{bsmallmatrix}$ and get $T\teh \mathbf{H} = (0, p, -q)\teh + (0, q, p)\teh$ and $T \teh \mathbf{G} = (-y, x, 0)\teh + (ax, ay, b)$.
      We may restrict to the domain $\{\dg, \dr\}$ and obtain $\mathbf{G}' = (T \teh \mathbf{G})\domres{\dg, \dr} = (x, 0)\teh + (ay, b)\teh$. 
      $\holbs(\mathbf{G}')$ is tractable only when $x = 0$ by degeneracy or $y = 0$, so the two vectors are orthogonal.
      If $x = 0$, then $\mathbf{H}$ is $\bdgr$ since $\supp \mathbf{H} \subseteq \{\dg, \dr\}^*$.
      If $y = 0$, then $\mathbf{G}'$ is a \geneq, so $(T \teh \mathbf{H})\domres{\dg, \dr} = (p, -q)\teh + (q, p)\teh$ must also be a \geneq, implying $p = 0$ or $q = 0$.
      Either case, we get $\mathbf{H} = \pm \mathbf{e}_3\teh \pm (x, y, 0)\teh$ which goes back to case (2).
    \item By \cref{prop:normalize-lin-dep-orthogonal-vectors}, we may assume that $\mathbf{H} = (\mathbf{u}+i \mathbf{v})\teh + (\mathbf{u} - i \mathbf{v})\teh$ where $\mathbf{u} = (px, py, -q)$ and $\mathbf{v} = (qx, qy, p)$ with $x^2 + y^2 = p^2 + q^2 = 1$.
      We apply the orthogonal holograhpic transformation $T = \begin{bsmallmatrix}
        -y & x & 0 \\
        x & y & 0 \\
        0 & 0 & 1
      \end{bsmallmatrix}$ and get $T \teh \mathbf{H} = (\mathbf{u}' + i\mathbf{v}')\teh + (\mathbf{u}' - i \mathbf{v}')\teh$ and $T \teh \mathbf{G} = (-y, x, 0)\teh + (ax, ay, b)\teh$
      where $\mathbf{u}' = (0, p, -q)$ and $\mathbf{v}' = (0, q, p)$.
      By \cref{lem:z-eqbg} and \cref{lem:domain-restriction}, $T \teh \mathbf{H}$ can realize $(=_{\dg \dr})$ so we may restrict to domain $\{\dg, \dr\}$.
      The only way for $\holbs((T \teh \mathbf{H})\domres{\dg, \dr}, (T \teh \mathbf{G})\domres{\dg, \dr})$ to be not \sph is when $(T \teh \mathbf{G})\domres{\dg, \dr}$ is degenerate since otherwise it is a type $\typei$ signature.
      This happens only when $x = 0$ since $b \ne 0$.
      If $x = 0$, then $\mathbf{H}$ is $\bdgr$ since $\supp \mathbf{H} \subseteq \{\dg, \dr\}^*$.
      \qedhere
  \end{enumerate}
\end{proof}

\begin{proposition}\label{prop:bg-r-d-do-not-mix}
  Let $\mathbf{F} = (a, b, 0)\teh + \mathbf{e}_3\teh$ for nonzero $a, b \in \mathbb{R}$ such that $a^2 + b^2 = 1$.
  Let $\mathbf{G}$ be a nondegenerate symmetric signature in $\mathcal{D}$.
  Then, $\holts(\mathbf{F}, \mathbf{G})$ is \sph unless one of the following conditions holds.
  \begin{enumerate}
    \item $\mathbf{G}$ is \swbg.
    \item $\mathbf{G}$ is $\db \dg | \dr$.
    \item for $T = \begin{bmatrix}
        -b & a & 0 \\
        a & b & 0 \\
        0 & 0 & 1
      \end{bmatrix}$, $T\tew \mathbf{G}$ is \strspt.
  \end{enumerate}
\end{proposition}
\begin{proof}
  Recall that since $T$ is $\db \dg | \dr$, $T \mathcal{D} = \mathcal{D}$.
  Applying the orthogonal transformation by $T$, we have $T\teh \mathbf{F} = \mathbf{e}_2\teh + \mathbf{e}_3\teh$, 
  and $T\tew \mathbf{G} \in \mathcal{D}$.
  Therefore, we may apply \cref{lem:gr-equality-D-dichotomy} on $T \tew \mathbf{G}$.
  It must be the case that $T\tew \mathbf{G}$ is $\db \dg | \dr$, \swbg, or strictly supported.
  Note that $T$ is a $\db \dg | \dr$ matrix, so if $T\tew \mathbf{G}$ is $\db \dg | \dr$ or \swbg,
  $\mathbf{G}$ must also be $\db \dg | \dr$ or \swbg respectively.
\end{proof}

We are now ready to formalize the idea behind \cref{fig:geometric-intuition-3-signatures}. 
We define a notion of a plane of a rank $2$ signature.
A rank 2 signature of type \ternarytractgeneq or type \ternarytractz of \cref{thm:dich-single-sym-ter-dom3-real} has a symmetric tensor decomposition such that the two vectors occurring inside it are orthogonal, i.e. $\mathbf{v}_1\teh + \mathbf{v}_2\teh$ or $(\mathbf{v}_1 + i \mathbf{v}_2)\teh + (\mathbf{v}_1 - i \mathbf{v}_2)\teh$.
The two vectors are unique up to a scalar, so the plane of a signature $\mathbf{F}$, denoted $P_{\mathbf{F}} = \spn(\mathbf{v}_1, \mathbf{v}_2)$ is well defined.
\begin{definition}
  Let $\mathbf{F}$ be a rank $2$ signature of type \ternarytractgeneq or type \ternarytractz.
  If $\mathbf{F}$ is type \ternarytractgeneq, we may write $\mathbf{F} = \mathbf{v}_1\teh + \mathbf{v}_2\teh$.
  If $\mathbf{F}$ is type \ternarytractz, we may write $\mathbf{F} = (\mathbf{v}_1 + i \mathbf{v}_2)\teh + (\mathbf{v}_1 - i \mathbf{v}_2)\teh$.
  We denote the plane of $\mathbf{F}$ as $P_{\mathbf{F}} = \spn(\mathbf{v}_1, \mathbf{v}_2)$.
\end{definition}
Note that $\mathbf{F}$ is \strspt if and only if $P_{\mathbf{F}}$ is one of the coordinate planes, i.e. $xy$-plane, $yz$-plane or $xz$-plane.
We will say that two planes are orthogonal if their normal vectors are orthogonal.
Note that this notion is well defined because a normal vector to a plane in $3$ dimensional space is unique up to scalar.

\begin{proposition}\label{prop:cannot-strspt-plane}
  Let $\mathcal{F}$ be a set of rank $2$ signatures of type \ternarytractgeneq or type \ternarytractz.
  The following two statements are equivalent:
  \begin{enumerate}
    \item For any orthogonal $T$, there is some signature $\mathbf{F} \in T \mathcal{F}$ that is not \strspt.
    \item There exist $\mathbf{F}, \mathbf{G} \in \mathcal{F}$ such that $P_{\mathbf{F}}$ and $P_{\mathbf{G}}$ are not equal or orthogonal.
  \end{enumerate}
\end{proposition}
\begin{proof}
  We first show that $(2)$ implies $(1)$.
  Suppose there are such $\mathbf{F}$ and $\mathbf{G}$. 
  For $T \teh \mathbf{F}$ to be \strspt, it must be the case that $P_{T \teh \mathbf{F}}$ is a coordinate plane.
  Since $T$ is orthogonal, it preserves the angle between vectors.
  In particular, it preserves the angle between the normal vectors of the planes of signatures.
  Therefore, $P_{T\teh \mathbf{G}}$ cannot be a coordinate plane since the normal vectors of coordinate planes are orthogonal to each other.
  Therefore, $T\teh \mathbf{G}$ cannot be \strspt if $T \teh \mathbf{F}$ is.

  To show that $(1)$ implies $(2)$, consider the contrapositive. 
  Suppose for all $\mathbf{F}, \mathbf{G} \in \mathcal{F}$, either $P_{\mathbf{F}}$ and $P_{\mathbf{G}}$ are same or orthogonal.
  Consider the set
  \[
    N =   \{\mathbf{n}_{\mathbf{F}} : \mathbf{n}_{\mathbf{F}} \text{ is a normal vector to } P_{\mathbf{F}}\} \, .
  \]
  By assumption, all vectors in $N$ are either orthogonal or scalar multiple of each other.
  Therefore, there exists three pairwise orthogonal nonzero vectors $\mathbf{v}_1, \mathbf{v}_2, \mathbf{v}_3 \in \mathbb{R}^3$ such that for all $\mathbf{u} \in N$, $\mathbf{u} \sim \mathbf{v}_i$ for some $i$.
  Then, we may apply an orthogonal transformation $T$ mapping $\mathbf{v}_i \mapsto \mathbf{e}_i$ which brings each of the vectors in $N$ to $x$, $y$, or $z$ axis.
  Then, $P_{T\teh \mathbf{F}}$ is a coordinate plane for all $\mathbf{F} \in \mathcal{F}$, and thus all signatures in $T \mathcal{F}$ are \strspt.
\end{proof}

The next lemma encapsulates the essence of the dichotomy theorems of two rank $2$ ternary signatures in \cref{sec:two-ternary}.
\begin{lemma}\label{lem:set-of-ternary-signatures-must-bg-r-or-separate}
  Let $\mathcal{F}$ be a set of rank $2$ signatures of type \ternarytractgeneq or type \ternarytractz.
  Suppose that for every orthogonal matrix $T$, neither of the following holds:
  \begin{enumerate}
    \item for all $\mathbf{F} \in T \mathcal{F}$, $\mathbf{F}$ is $\bdgr$.
    \item all signatures in $T \mathcal{F}$ are \strspt.
  \end{enumerate}
  Then, $\holts(\mathcal{F})$ is \sph.
\end{lemma}
\begin{proof}
  Since $(2)$ does not hold, for any $T$, there is some signature in $T \mathcal{F}$ that is not strictly supported. 
  Then, by \cref{prop:cannot-strspt-plane}, there exist two signatures $\mathbf{F}, \mathbf{G} \in \mathcal{F}$ such that $P_{\mathbf{F}}$ and $P_{\mathbf{G}}$ are not orthogonal or same.

  Suppose $\mathbf{F}$ is type \ternarytractz.
  Then, we may apply some orthogonal $T$ such that $T\teh \mathbf{F} = (1, i, 0)\teh + (1, -i, 0)\teh$.
  If $\holts(T \teh \mathbf{F}, T \teh \mathbf{G})$ is \sph, then we are done.
  Otherwise,
  considering each case in \cref{lem:dichotomy-terneray-ternary-rank-2-z}, if $P_{T \teh \mathbf{G}}$ is not the same or orthogonal to $P_{T \teh \mathbf{F}}$, then it must be the case that $\holts(T \teh \mathbf{F}, T \teh \mathbf{G})$ is \sph.

  Suppose $\mathbf{F}$ is type \ternarytractgeneq.
  Then, we may apply some orthogonal $T$ such that $T \teh \mathbf{F} = \mathbf{e}_1\teh + \mathbf{e}_2\teh$.
  If $\holts(T \teh \mathbf{F}, T \teh \mathbf{G})$ is \sph, then we are done.
  Otherwise,
  by \cref{lem:dichotomy-ternary-ternary-rank-2-geneq}, 
  $T\teh \mathbf{G}$ must be $\mathbf{e}_1\teh + (0, a, b)\teh$ or $\mathbf{e}_2\teh + (a, 0, b)\teh$
  for nonzero $a, b \in \mathbb{R}$.
  Without loss of generality, assume $T \teh \mathbf{G} = \mathbf{e}_1\teh + (0, a, b)\teh$.
  Then, by \cref{prop:dichotomy-3-signatures-geneq-b-gr}, any other signature $\mathbf{H} \in T\mathcal{F}$ must be $\bdgr$ for $\holts(T \mathcal{F})$ to be not \sph.
  However, this contradicts the assumption that $(1)$ does not hold.
\end{proof}

\subsection{Proof of the Main Theorem}
We are now ready to prove the main theorem.
We restate it here for an easy reference.
\maintheorem*

We will assume that $\holts(\mathcal{F})$ is not \sph, and show that then, $\mathcal{F}$ must fall into one of the tractable classes in the theorem statement.
We may assume that for all $\mathbf{F} \in \mathcal{F}$, $\holts(\mathbf{F})$ is tractable.
By \cref{lem:geneq-normalization}, \cref{lem:Z-normalization} and \cref{cor:replace-arity-3}, we may assume that all signatures are either ternary or binary, 
and all ternary signatures are written in the tensor decomposed form with unit vectors.

We first explain the organization of the proof.
$\mathcal{O}$ is defined in \cref{eq:g-set} as the orbit of the ternary signatures under the monoid action by the binary signatures.
We list the assumptions to be used subsequently.
\begin{enumerate}
    \item $\mathcal{F}$ contains a rank $3$ type \ternarytractz signature.
    \item $\mathcal{F}$ contains a rank $3$ type \ternarytractgeneq signature.
    \item $\mathcal{O}$ can be transformed such that all signatures are supported on $\{\db, \dg\}$.
    \item
    $\mathcal{O}$ can be transformed such that all signatures are in $\mathcal{E}$.
    \item 
    $\mathcal{O}$ can be transformed such that all signatures are $\db \dg | \dr$.
    \item $\mathcal{O}$ can be transformed such that all signatures are \strspt.
\end{enumerate}
With these assumptions, we can show that the following sequence covers all possible cases of $\holts(\mathcal{F})$ being tractable using the dichotomy statements proven so far.
\begin{itemize}
    \item \cref{subsec:contains-rank-3-z}: Assume 1 holds.
    \item \cref{subsec:contains-rank-3-geneq}: Assume 1 does not hold. Assume 2 holds.
    \item \cref{subsec:can-be-made-bg}: Assume 1, 2 do not hold. Assume 3 holds.
    \item \cref{subsec:cannot-be-made-bg}: Assume 1, 2, 3 do not hold. Assume 4 holds.
    \item \cref{subsec:can-be-made-bg-r}: Assume 1, 2, 3, 4 do not hold. Assume 5 holds.
    \item \cref{subsec:cannot-be-made-bg-r}: Assume 1, 2, 3, 4, 5 do not hold. Assume 6 holds.
\end{itemize}

Let $\boldsymbol{\beta} = \frac{1}{\sqrt{2}}(1, i, 0)\transpose$ for the rest of this paper.

\subsubsection{
Rank 3 Type \texorpdfstring{\ternarytractz}{B}
}\label{subsec:contains-rank-3-z}

Suppose $\mathcal{F}$ contains a rank $3$ type \ternarytractz signature.
Then, there exists some real orthogonal matrix $T$ such that $T \mathcal{F}$ contains $\mathbf{F} = \boldsymbol{\beta}\teh + \overline{\boldsymbol{\beta}}\teh + \mathbf{e}_3\teh $.
We may assume that all other ternary signatures in $T \mathcal{F}$ are of the tractable form in \cref{lem:dichotomy-terneray-ternary-rank-3-z}.
Also, we may assume that all binary signatures in $T \mathcal{F}$ are of the tractable form in \cref{lem:dichotomy-single-ternary-rank-3-z-single-binary}.
Therefore, for all $\mathbf{G} \in T \mathcal{F}$, $\mathbf{G}$ is either $\db \dg | \dr$ and $\mathbf{G} \domres{\db, \dg}$ is tractable with $\typeii$ signature, or $\mathbf{G}$ is binary and is \swbg. 
This means $\mathcal{F}$ is in the tractable class \tractBGR.

\subsubsection{
Rank 3 Type \texorpdfstring{\ternarytractgeneq}{A}
}\label{subsec:contains-rank-3-geneq}

Suppose $\mathcal{F}$ contains a rank $3$ type \ternarytractgeneq signature.
Then, there exists some real orthogonal matrix $T$ such that $T \mathcal{F}$ contains $\mathbf{F} = \mathbf{e}_1\teh + \mathbf{e}_2 \teh + \mathbf{e}_3 \teh$.   
We may assume that all other ternary signatures in $T \mathcal{F}$ are of the tractable form in \cref{lem:dichotomy-ternary-ternary-rank-3-geneq}.
Also, we may assume that all binary signatures in $T \mathcal{F}$ are of the tractable form in \cref{lem:dichotomy-single-ternary-rank-3-geneq-single-binary}.
If $T \mathcal{F} \subseteq \mathcal{E}$, then we are done since $\mathcal{F}$ is in the tractable class \tractE.
So, assume that $T \mathcal{F} \not \subseteq \mathcal{E}$.

Suppose there exists a ternary signature $\mathbf{G} \in T \mathcal{F} \setminus \mathcal{E}$.
Since $\holts(\mathbf{F}, \mathbf{G})$ is assumed to be not \sph, by \cref{lem:dichotomy-ternary-ternary-rank-3-geneq}, $\mathbf{G}$ must be of the form $\mathbf{e}_i\teh + \mathbf{v}\teh$ for some $\mathbf{v} \in \mathbb{R}^3$ such that $\langle \mathbf{e}_i, \mathbf{v} \rangle = 0$ and $\mathbf{v} \not \sim \mathbf{e}_j$ for all $j$.
Without loss of generality, assume $i = 3$, so $\mathbf{G} = (a, b, 0)\teh + \mathbf{e}_3\teh$ for some nonzero $a, b \in \mathbb{R}$.
By \cref{prop:dichotomy-3-signatures-geneq-b-gr}, all other non-\geneq ternary signatures in $T \mathcal{F}$ must also be of the form $(c, d, 0)\teh + \mathbf{e}_3\teh$ for some $c, d \in \mathbb{R}$.
Suppose there is a binary signature $\mathbf{H} \in T \mathcal{F}$.
Then, it must be domain separated or one of \swbg, \swbr, or \swgr.
We claim that $\mathbf{H}$ must be $\db \dg | \dr$ or \swbg.
By \cref{prop:bg-r-d-do-not-mix}, if $\mathbf{H}$ is \swbr or \swgr, then it must be \swbg or $\db \dg | \dr$ to be tractable.
Suppose $\mathbf{H}$ is $\db \dr | \dg$ and not $\db \dg | \dr$.
Then, $\mathbf{H} = \begin{bsmallmatrix}
  x & 0 & y \\
  0 & z & 0 \\
  y & 0 & w
\end{bsmallmatrix}$ for some $x, y, z, w \in \mathbb{R}$ with $y \ne 0$.
If $z \ne 0$, then it must be the case that $x = 0$ since we may realize $\mathbf{H}\teh (\mathbf{e}_1\teh + \mathbf{e}_2\teh) = (x, 0, y)\teh + (0, z, 0)\teh$, which is not $\db \dg | \dr$. 
Then by \cref{prop:dichotomy-3-signatures-geneq-b-gr}, $\holts(\mathbf{G}, \mathbf{H}\teh(\mathbf{e}_1\teh + \mathbf{e}_2\teh)$ is \sph, contrary to the assumption.
Similarly, if $z \ne 0$, then it must be the case that $w = 0$ to be not \sph.
Then, $\mathbf{H}\teh \mathbf{G} = (0, bz, ay)\teh + (y, 0, 0)\teh$, which again is not of $\db \dg | \dr$ form.
So, it must be the case that $z = 0$ to be not \sph.

If $z = 0$, then, $\mathbf{H}\teh \mathbf{G} = (ax, 0, ay)\teh + (y, 0, w)\teh$.
If $x, w = 0$, then $\mathbf{H}$ is \swbg.
Otherwise, $(\mathbf{H}\teh \mathbf{G})\domres{\db, \dr} = (ax, ay)\teh + (y, w)\teh$, 
and it is not degenerate if $\mathbf{H}$ is not degenerate.
Since $(x, w) \ne (0, 0)$, $\holbs(\mathbf{F}\domres{\db, \dr}, (\mathbf{H}\teh \mathbf{G})\domres{\db, \dr})$ is \sph
because $(\mathbf{H}\teh \mathbf{G})\domres{\db, \dr})$ is not a \geneq.

Therefore, any binary $\db \dr | \dg$ signature that is not $\db \dg | \dr$ must be \swbg.
Similar argument shows that any binary $\bdgr$ signature that is not $\db \dg | \dr$ must be \swbg.
Hence $\mathcal{F}$ is in the tractable case \tractBGR.

Suppose there is no ternary signature in $T \mathcal{F} \setminus \mathcal{E}$.
Suppose there is a binary signature $\mathbf{G} \in T \mathcal{F} \setminus \mathcal{E}$.
Suppose $\mathbf{G}$ is $\db \dg | \dr$, so $\mathbf{G} = \begin{bsmallmatrix}
  x & y & 0 \\
  y & z & 0 \\
  0 & 0 & w
\end{bsmallmatrix}$ for some $x, y, z, w \in \mathbb{R}$ with $y \ne 0$.
Also, we must have $x, z \ne 0$, since if $x, z = 0$, $\mathbf{G} \in \mathcal{E}$, and if only one of them is $0$, $\holbs(\mathbf{G}\domres{\db, \dg}, [1, 0, 0, 1])$ is \sph.
That means $[x, y, z]$ must be degenerate, and for $\mathbf{G}$ to be nondegenerate, we need $w \ne 0$.
Then, $\mathbf{G}\teh(\mathbf{e}_2\teh + \mathbf{e}_3\teh) = (y, z, 0)\teh + (0, 0, w)\teh$, so we go back to the previous case of having some non-\geneq ternary signature.

Suppose all binary signatures in $T \mathcal{F} \setminus \mathcal{E}$ are \swbg, \swbr, \swgr.
If $\mathbf{G}$ is a such \swbg signature, then $\mathbf{G} = \begin{bsmallmatrix}
  0 & 0 & x \\
  0 & 0 & y \\
  x & y & 0
\end{bsmallmatrix}$.
By assumption, we must have $x, y \ne 0$, and $\mathbf{G}\teh (\mathbf{e_2}\teh + \mathbf{e}_3\teh) = (x, y, 0)\teh + (0, 0, y)\teh$, so we go back to the previous case.

Therefore, if $T \mathcal{F} \not \subseteq \mathcal{E}$ and is tractable, it must be in class \tractBGR.

\subsubsection{Rank 2} \label{subsec:contains-rank-2}
Suppose $\mathcal{F}$ does not contain a rank $3$ ternary signature. 
If there is no rank $2$ ternary signature as well, then $\mathcal{F}$ only consists of binary signatures, so it is in the tractable class \tractbinary.
So, assume that there exists some rank $2$ ternary signature.

Let $\mathcal{T}$ be the set of ternary signatures in $\mathcal{F}$.
Let $\mathcal{B}$ be the set of binary signatures in $\mathcal{F}$.
Let $\langle \mathcal{B} \rangle$ be the monoid generated by $\mathcal{B}$ under multiplication.
We define $\mathcal{O}$ to be
\begin{equation}\label{eq:g-set}
  \mathcal{O} := \{\mathbf{G}\teh \mathbf{F} : \mathbf{F} \in \mathcal{T}, \mathbf{G} \in \langle \mathcal{B} \rangle, \mathbf{G}\teh \mathbf{F} \text{ is non-degenerate} \} \, .
\end{equation}
Combinatorially, $\mathcal{O}$ is the set of all gadgets constructible from connecting the same chain of binary signatures to the three edges of a ternary signatures,
$(\mathbf{G}_1 \mathbf{G}_2 \cdots \mathbf{G}_k)\teh \mathbf{F}$ for some $\mathbf{G}_i \in \mathcal{B}$ and $\mathbf{F} \in \mathcal{T}$.
It can also be viewed as the orbit (ignoring degenerate signatures) of $\mathcal{T}$ under the monoid action of $\langle \mathcal{B} \rangle$, where the action is defined by $\mathbf{G} : \mathbf{F} \mapsto \mathbf{G} \teh \mathbf{F}$ for $\mathbf{G} \in \langle \mathcal{B} \rangle$ and $\mathbf{F} \in \mathcal{T}$.
Note that $\mathcal{O}$ is a set of symmetric ternary signatures, and if $\mathbf{G} \in \mathcal{B}$ and $\mathbf{F} \in \mathcal{O}$, then $\mathbf{G}\teh \mathbf{F} \in \mathcal{O}$ as well.
In addition, for any orthogonal matrix $T$, $T \mathcal{O}$ is the result of the same construction with $T \mathcal{F}$ and $\langle T \mathcal{B} \rangle$.
Also, we have $\holts(\mathcal{O} \cup \mathcal{B}) =_T \holts(\mathcal{F})$.

\subsubsection{Rank 2 Class \texorpdfstring{\tractBG}{C}}
\label{subsec:can-be-made-bg}
Suppose there exists an orthogonal $T$ such that $\supp \mathbf{F} \subseteq \{\db, \dg\}^*$ for all $\mathbf{F} \in T \mathcal{O}$.
We claim that $\mathcal{F}$ must be in the tractable class \tractBG.
\begin{enumerate}
  \item Suppose there is $\mathbf{F} \in T \mathcal{F}$ such that it is of type \ternarytractgeneq.
    Since $\supp \mathbf{F} \subseteq \{\db, \dg\}^*$, we may apply a $\db \dg | \dr$ orthogonal matrix $S$ such that $\mathbf{F}' = S\teh \mathbf{F} = \mathbf{e}_1\teh + \mathbf{e}_2\teh$.
    Note that all signatures in $\mathcal{O}' = S T \mathcal{O}$ are supported on $\{\db, \dg\}$ since $S$ is $\db \dg | \dr$.
    By \cref{lem:geneq-eqbg}, \cref{lem:domain-restriction} and \cref{thm:dich-sym-Boolean}, it must be the case that all signatures in $\mathcal{O}'$ are also \geneq.
    In particular, all signatures in $\mathcal{F}' = S T \mathcal{F}$ are also \geneq.

    We need to show that $\mathcal{B}' = S T \mathcal{B}$ only consists of binary signatures that are $\db \dg | \dr$ or $\mathcal{D}$.
    Suppose $\mathbf{G} \in \mathcal{B}'$. 
    If $\mathbf{G}\teh \mathbf{F}'$ is non-degenerate, then $\mathbf{G}\teh \mathbf{F}' \in \mathcal{O}'$, so by assumption we need $\supp \mathbf{G}\teh \mathbf{F}' \subseteq \{\db, \dg\}^*$.
    This condition is equivalent to the statement that if the first two columns of $\mathbf{G}$ are linearly independent, then the third row of the first two columns must be $0$. 
    In other words, if $\mathbf{G} \notin \mathcal{D}$, then it must be $\db \dg | \dr$.

  \item Suppose there is $\mathbf{F} \in T \mathcal{F}$ such that it is of type \ternarytractz.
    Since $\supp \mathbf{F} \subseteq \{\db, \dg\}^*$, we may apply a $\db \dg | \dr$ orthogonal matrix $S$ such that $\mathbf{F}' = S\teh \mathbf{F} = \boldsymbol{\beta}\teh + \overline{\boldsymbol{\beta}}\teh$.
    By \cref{lem:z-eqbg}, \cref{lem:domain-restriction}, and \cref{thm:dich-sym-Boolean}, it must be the case that all signatures in $\mathcal{O}'$ are also type $\typeii$.
    In particular, all signatures in $\mathcal{F}' = S T \mathcal{F}$ are also type $\typeii$.

    Same argument as case (1) shows that $ST \mathcal{B}$ only consists of binary signatures that are $\db \dg | \dr$ or $\mathcal{D}$.
\end{enumerate}

\subsubsection{Rank 2 Class \texorpdfstring{\tractE}{B} or \texorpdfstring{\tractBGR}{D}}
\label{subsec:cannot-be-made-bg}
From now on, we assume that under any orthogonal transformation $T$, there exists some $\mathbf{F} \in T \mathcal{O}$ such that $\supp \mathbf{F} \not \subseteq \{\db, \dg\}^*$.
Suppose there exists some orthogonal $T$ such that $T \mathcal{O} \subseteq \mathcal{E}$.
We claim that $\mathcal{F}$ must be in the tractable class \tractE or \tractBGR.
Let $\mathcal{O}' = T \mathcal{O}$.
We may assume that there exists $\mathbf{F} = \mathbf{e}_1\teh + \mathbf{e}_2\teh \in \mathcal{O}'$ after permuting the domains.
Also, by the assumption that not all signatures can be supported on $\{\db, \dg\}$, we may assume that there is some $\mathbf{G} = \mathbf{e}_2\teh + \mathbf{e}_3\teh \in \mathcal{O}'$.
We need to show that one of the following holds: 
(1) every signature in $\mathcal{B}' = T \mathcal{B}$ is in $\mathcal{E}$;
or 
(2) every signature in $\mathcal{B}'$ is $\db \dr | \dg$ or \swbr. 

Suppose $\mathbf{H} \in \mathcal{B}'$, 
and let $\mathbf{v}_1, \mathbf{v}_2, \mathbf{v}_3$ be the columns of $\mathbf{H}$.
Then $\mathbf{H}\teh \mathbf{F} = \mathbf{v}_1\teh + \mathbf{v}_2\teh$.
If $\mathbf{H}\teh \mathbf{F} = \mathbf{v}_1\teh + \mathbf{v}_2\teh$ is not degenerate, then it must be the case that it is also a \geneq since $\mathbf{H} \teh \mathbf{F} \in \mathcal{O}' \subseteq \mathcal{E}$.
This implies that $\mathbf{v}_1 \sim \mathbf{e}_i$ and $\mathbf{v}_2 \sim \mathbf{e}_j$ for $i \ne j$.
Since $\mathbf{H}$ is symmetric, this implies $\mathbf{H} \in \mathcal{E}$ immediately, except in one case of $\mathbf{v}_1 \sim \mathbf{e}_1$ and $\mathbf{v}_2 \sim \mathbf{e}_3$, where $\mathbf{H} = \begin{bsmallmatrix}
  a & 0 & 0 \\
  0 & 0 & b \\
  0 & b & c
\end{bsmallmatrix}$.
Since $\mathbf{H} \teh \mathbf{G} = (0, 0, b)\teh + (0, b, c)\teh$ is in $\mathcal{E}$ by assumption, it must be the case that $b = 0$ or $c = 0$, which implies $\mathbf{H} \in \mathcal{E}$.
Suppose $\mathbf{v}_1$ and $\mathbf{v}_2$ are linearly dependent.
Then we may look at $\mathbf{H}\teh \mathbf{G} = \mathbf{v}_2\teh + \mathbf{v}_3\teh$.
If $\mathbf{v}_2$ and $\mathbf{v}_3$ are linearly independent, then by the same argument, $\mathbf{H} \in \mathcal{E}$.
So the only uncovered case is if $\mathbf{v}_2$ and $\mathbf{v}_3$ are also linearly dependent.
If $\mathbf{v}_2 \ne 0$, then $\mathbf{H}$ is degenerate.
Otherwise, if $\mathbf{v}_2 = 0$, then $\mathbf{H}$ must have the form $\begin{bsmallmatrix}
  a & 0 & b \\
  0 & 0 & 0 \\
  b & 0 & c
\end{bsmallmatrix}$.

There are two possible cases.
If $\mathbf{e}_1\teh + \mathbf{e}_3\teh \in \mathcal{O}'$, then by the same analysis, the above $\mathbf{H}$ must be degenerate or in $\mathcal{E}$.
Therefore, $\mathcal{B}' \subseteq \mathcal{E}$, and thus we get the tractable class \tractE.
Otherwise, if $\mathbf{e}_1\teh + \mathbf{e}_3\teh \notin \mathcal{O}'$, $\mathcal{B}'$ may contain binary signatures that are supported on $\{\db, \dr\}$.
Then, we claim that every signature in $\mathcal{B}'$ must be $\db \dr | \dg$ or \swbr.
We have already shown above that the binary signatures in $\mathcal{B}'$ are either \genperm or possibly supported on $\{\db, \dr\}$.
Therefore, only cases we need to further rule out are the symmetric \genperm matrices not supported on $\{\db, \dr\}$, which are of the form $\mathbf{H}_1 = \begin{bsmallmatrix}
  0 & x & 0 \\
  x & 0 & 0 \\
  0 & 0 & y
  \end{bsmallmatrix}$ and $\mathbf{H}_2 = \begin{bsmallmatrix}
  y & 0 & 0 \\
  0 & 0 & x \\
  0 & x & 0
\end{bsmallmatrix}$ for nonzero $x, y$, since if $y = 0$, they are both \swbr, and if $x = 0$, they are both degenerate.
We see that then $\mathbf{H}_1\teh \mathbf{G} = x^3 \mathbf{e}_1\teh + y^3 \mathbf{e}_3\teh$ and $\mathbf{H}_2\teh \mathbf{F} = y^3 \mathbf{e}_1\teh + x^3 \mathbf{e}_3\teh$, which can realize $\mathbf{e}_1\teh + \mathbf{e}_3\teh$. 
Then, we may assume it is in $\mathcal{O}'$ to show that $\mathcal{B}' \subseteq \mathcal{E}$.
Otherwise such symmetric \genperm matrices cannot exist in $\mathcal{B}'$, so $\mathcal{F}$ must be in the tractable class \tractBGR.

\subsubsection{Rank 2 Class \texorpdfstring{\tractBGR}{D} or \texorpdfstring{\tractBGGRBR}{E}}
\label{subsec:can-be-made-bg-r}
From now on, we further assume that under any orthogonal transformation $T$, there exists some $\mathbf{F} \in T \mathcal{O}$ such that $\mathbf{F}$ is not a \geneq.
By \cref{lem:set-of-ternary-signatures-must-bg-r-or-separate}, we may assume that there is some orthogonal $T$ such that all signatures in $T \mathcal{O}$ are $\db \dg | \dr$.
\begin{enumerate}
  \item Suppose there is no signature in $T \mathcal{O}$ that is supported on $\{\db, \dg\}$.
    This means that all signatures in $T \mathcal{O}$ must have the form $(a, b, 0)\teh + \mathbf{e}_3\teh$ for some $a, b \in \mathbb{R}$.
    Note that no type \ternarytractz signature can exist since a rank $2$ type \ternarytractz signature cannot be $\db \dg | \dr$ without being supported on $\{\db, \dg\}$.
    Fix one such $\mathbf{F} = (a, b, 0)\teh + \mathbf{e}_3\teh \in T \mathcal{O}$.
    We may assume that $a^2 + b^2 = 1$.
    We may apply a $\db \dg | \dr$ orthogonal transformation $S = \begin{bsmallmatrix}
      -b & a & 0 \\
      a & b & 0 \\
      0 & 0 & 1
    \end{bsmallmatrix}$ so that $\mathbf{F}' = S\teh \mathbf{F} = \mathbf{e}_2\teh + \mathbf{e}_3\teh$.
    Note that every signature in $\mathcal{O}' = S T \mathcal{O}$ is $\db \dg | \dr$.
    By the assumption that $\mathcal{O}$ cannot be transformed into $\mathcal{E}$, we may assume that there exists some $\mathbf{G}' \in \mathcal{O}$ such that 
    $\mathbf{G}' = (c, d, 0)\teh + \mathbf{e}_3\teh$ for nonzero $c, d \in \mathbb{R}$.
    We claim that $\mathcal{F}$ must be in the tractable class \tractBGR.
    We need to show that all binary signatures in $\mathcal{B}' = ST \mathcal{B}$ are $\db \dg | \dr$ or \swbg.

    Suppose $\mathbf{H} \in \mathcal{B}'$ and let $\mathbf{v}_1, \mathbf{v}_2, \mathbf{v}_3$ be the columns of $\mathbf{H}$.
    If $\mathbf{v}_2, \mathbf{v}_3$ are linearly independent, then we need $\mathbf{H}\teh \mathbf{F}' = \mathbf{v}_2\teh + \mathbf{v}_3\teh$ to also have the form $(e, f, 0)\teh + \mathbf{e}_3\teh$.
    This means that either we have $\mathbf{v}_2 \sim (e, f, 0)$ and $\mathbf{v}_3 \sim \mathbf{e}_3$ or $\mathbf{v}_2 \sim \mathbf{e}_3$ and $\mathbf{v}_3 \sim (e, f, 0)$.
    The first case implies $\mathbf{H}$ is $\db \dg | \dr$ by symmetry.
    The second case implies $\mathbf{H}$ is of the form $\begin{bsmallmatrix}
      x & 0 & e \\
      0 & 0 & f \\
      e & f & 0
    \end{bsmallmatrix}$ for some $x$.
    If $x = 0$, then $\mathbf{H}$ is \swbg, so assume otherwise.
    Then, $\mathbf{H}\teh \mathbf{G}' = (cx, 0, ce + df)\teh + (e, f, 0)\teh$.
    If this is degenerate, then we must have $f = 0$ and $e = 0$, which implies $\mathbf{H}$ is degenerate.
    Otherwise, $e$ must be $0$ for the two vectors to be orthogonal, but if $f \ne 0$, then this signature is not $\db \dg | \dr$, contrary to assumption.

    Now, suppose $\mathbf{v}_2$ and $\mathbf{v}_3$ are linearly dependent.
    Then, $\mathbf{H}\teh \mathbf{G}' = (c \mathbf{v}_1 + d \mathbf{v}_2)\teh + \mathbf{v}_3\teh$.
    If $c \mathbf{v}_1 + d \mathbf{v}_2$ and $\mathbf{v}_3$ are linearly independent, we need either $c \mathbf{v}_1 + d \mathbf{v}_2 \sim (e, f, 0)$ and $\mathbf{v}_3 \sim \mathbf{e}_3$ or $c \mathbf{v}_1 + d \mathbf{v}_2 \sim \mathbf{e}_3$ and $\mathbf{v}_3 \sim (e, f, 0)$.
    The first case implies that $\mathbf{H}$ is $\db \dg | \dr$ by symmetry of $\mathbf{H}$.
    The second case implies $\mathbf{H}$ is of the form $\begin{bsmallmatrix}
      * & * & e \\
      * & * & f \\
      e & f & 0
    \end{bsmallmatrix}$ and linear dependence of $\mathbf{v}_2, \mathbf{v}_3$ further implies $f = 0$.
    So, $\mathbf{H}$ is $\begin{bsmallmatrix}
      x & y & e \\
      y & 0 & 0 \\
      e & 0 & 0
    \end{bsmallmatrix}$ for some $x, y \in \mathbb{R}$.
    For $c \mathbf{v}_1 + d \mathbf{v}_2 \sim \mathbf{e}_3$ to be true we must have $y = 0$, which then implies $x = 0$, so $\mathbf{H}$ is \swbg.
    If $c \mathbf{v}_1 + d \mathbf{v}_2$ and $\mathbf{v}_3$ are linearly dependent, then $\mathbf{v}_1, \mathbf{v}_3$ are linearly dependent as well, so $\mathbf{H}$ is degenerate.
  \item Now, we may assume there is some signature in $T \mathcal{O}$ that is supported on $\{\db, \dg\}$.
    Assume $\mathbf{F}$ is supported on $\{\db, \dg\}$ and is of type \ternarytractgeneq.
    By applying a $\db \dg | \dr$ orthogonal transformation $S$, we have $\mathbf{F}' = S \teh \mathbf{F} = \mathbf{e}_1\teh + \mathbf{e}_2\teh$.
    Let $\mathcal{O}' = S T \mathcal{O}$.
    By the assumptions, we may assume that there is some $\mathbf{G}' \in \mathcal{O}'$ such that $\mathbf{G}' = (a, b, 0)\teh + \mathbf{e}_3\teh$ for nonzero $a, b \in \mathbb{R}$.
    We claim that $\mathcal{F}$ must be in the tractable class \tractBGR or \tractBGGRBR.
    We will narrow down the possible forms of binary signatures in $\mathcal{B}' = ST \mathcal{B}$.
    We will show that either: $\mathcal{B}'$ consists of $\db \dg | \dr$ and \swbg signatures; or $\mathcal{O}'$ and $\mathcal{B}'$ can be transformed so that every signature becomes \strspt or in $\octgroup$.

    Let $\mathbf{H} \in \mathcal{B}'$ and $\mathbf{v}_1, \mathbf{v}_2, \mathbf{v}_3$ the columns.
    Assume $\mathbf{H}$ is not $\db \dg | \dr$.
    Suppose $\mathbf{v}_1$ and $\mathbf{v}_2$ are linearly independent.
    Then $\mathbf{H}\teh \mathbf{F}' = \mathbf{v}_1\teh + \mathbf{v}_2\teh$ is not supported on $\{\db, \dg\}$ since $\mathbf{H}$ is not $\db \dg | \dr$.
    Therefore, it must be the case that $\mathbf{v}_1 \sim (e, f, 0)$ and $\mathbf{v}_2 \sim \mathbf{e}_3$ or $\mathbf{v}_1\sim \mathbf{e}_3$ and $\mathbf{v}_2 \sim (e, f, 0)$ for some $e, f \in \mathbb{R}$ to be $\db \dg | \dr$.
    The first case implies $\mathbf{H} = \begin{bsmallmatrix}
      e & 0 & 0 \\
      0 & 0 & x \\
      0 & x & y
    \end{bsmallmatrix}$ for some $x, y$ with $x \ne 0$.
    Then, $\mathbf{H}\teh \mathbf{G}' = (ae, 0, bx)\teh + (0, x, y)\teh$.
    Since $x \ne 0$, $\mathbf{H}\teh \mathbf{G}'$ cannot be degenerate.
    So,  $\mathbf{H}\teh \mathbf{G}'$ must be $\db \dg | \dr$, which implies that either $e, y = 0$ or $x = 0$.
    In the first case, $\mathbf{H}$ is \swbg, and in the second case, $\mathbf{H}$ is $\db \dg | \dr$.

    If $\mathbf{v}_1 \sim \mathbf{e}_3$ and $\mathbf{v}_2 \sim (e, f, 0)$, then $\mathbf{H}$ must be of the form
    $\begin{bsmallmatrix}
      0 & 0 & x \\
      0 & f & 0 \\
      x & 0 & y
    \end{bsmallmatrix}$ for some $x, y$ with $x \ne 0$.
    Then, $\mathbf{H}\teh \mathbf{G}' = (0, bf, ax)\teh + (x, 0, y)$.
    This cannot be degenerate since $a, x \ne 0$.
    Since $\mathbf{H}\teh \mathbf{G}' \in \mathcal{O}'$, we must have $f = 0$ and $y = 0$, implying that $\mathbf{H}$ is a \swbg.

    So far, we have shown that if $\mathbf{H} \notin \mathcal{D}$, then it is \swbg or $\db \dg | \dr$.
    If $\mathbf{H} \in \mathcal{D}$, then
    by \cref{prop:bg-r-d-do-not-mix}, for $\holts(\mathbf{G}', \mathbf{H})$ to be not \sph, $\mathbf{H}$ must be $\db \dg | \dr$, \swbg, or for $U = \begin{bsmallmatrix}
      -b & a & 0 \\
      a & b & 0 \\
      0 & 0 & 1
    \end{bsmallmatrix}$, $U \tew \mathbf{H}$ is \strspt.
    Note that being $\db \dg | \dr$ and \swbg are invariant under $\db \dg | \dr$ transform, 
    so every signature in $U \mathcal{B}'$ are $\db \dg | \dr$, \swbg, or \strspt.
    We will now analyze
    the \strspt signatures in $U \mathcal{B}'$.
    Also note that $U\teh \mathbf{G}' = \mathbf{e}_2\teh + \mathbf{e}_3\teh$.

    Suppose $\mathbf{H} \in U \mathcal{B}'$ and is supported on $\{\dg, \dr\}$.
    Since $U \teh \mathbf{G}'$ is a \geneq on $\{\dg, \dr\}$, it must be the case that $\mathbf{H}\domres{\dg, \dr}$ is degenerate, $[0, *, 0]$ or $[*, 0, *]$.
    $\mathbf{H}\domres{\dg, \dr}$ being degenerate implies that $\mathbf{H}$ is degenerate.
    Others imply that $\mathbf{H}$ is \swbg and $\db \dg | \dr$ respectively.
    If there is no $\mathbf{H} \in U \mathcal{B}'$ such that $\supp \mathbf{H} \subseteq \{\db, \dr\}^*$, then we are in the tractable case \tractBGR.

    Now, assume that $\mathbf{H} = \begin{bsmallmatrix}
      x & 0 & y \\
      0 & 0 & 0 \\
      y & 0 & z
    \end{bsmallmatrix} \in U \mathcal{B}'$ such that it is not $\db \dg | \dr$ or \swbg.
    This means $x, z$ cannot be both $0$ and $y \ne 0$.
    If such $\mathbf{H}$ exists, then we show $\mathcal{F}$ must be in the tractable class \tractBGGRBR.
    Note that since $U$ is $\db \dg | \dr$, every signature in $U \mathcal{O}'$ is still $\db \dg | \dr$ by assumption.

    We need to show two things: (a) in $U \mathcal{O}'$, there is no signature of the form $\mathbf{I} = (c, d, 0)\teh + \mathbf{e}_3\teh$ for nonzero $c, d$ and hence all signatures are \strspt;
    (b) the binary signatures in $U \mathcal{B}'$ are \strspt or in $\octgroup$.
    For (a), suppose such signature exists, and we have $\mathbf{H}\teh \mathbf{I} = (cx, 0, cy)\teh + (y, 0, z)\teh$.
    By assumption, this needs to be $\db \dg | \dr$ since $\mathbf{H}$ is nondegenerate, but that would require $y = 0$ or $x, z = 0$ contrary to assumption.
    For (b), we already know that the binary signatures that are not \strspt are of the form $\db \dg | \dr$ or \swbg.
    If they are not \strspt, then by composing with $U\teh \mathbf{G}' = \mathbf{e}_2\teh + \mathbf{e}_3\teh$, we get a ternary signature of the above form in (a) which we already have shown to be not in $U \mathcal{O}'$.

  \item We assume now that $\mathbf{F} \in T \mathcal{O}$ is supported on $\{\db, \dg\}$ and is of type \ternarytractz.
    By the assumptions, we may assume that there is some $\mathbf{G} \in T \mathcal{O}$ that is not supported on $\{\db, \dg\}$.
    It must have the form $(a, b, 0)\teh + \mathbf{e}_3\teh$.
    Let $S = \begin{bsmallmatrix}
      -b & a & 0 \\
      a & b & 0 \\
      0 & 0 & 1
    \end{bsmallmatrix}$, and let $\mathcal{O}' = S T \mathcal{O}$ and $\mathcal{B}' = S T \mathcal{B}$.
    Then, $\mathbf{G}' = S \teh \mathbf{G} = \mathbf{e}_2\teh + \mathbf{e}_3\teh$, and by \cref{cor:z-normalization-orthogonal},
    we may assume that there is $\mathbf{F}' \in \mathcal{O}'$ such that $\mathbf{F}' = \boldsymbol{\beta}\teh + \overline{\boldsymbol{\beta}}\teh$.

    By \cref{lem:dichotomy-single-ternary-rank-2-z-single-binary}, the only other possible forms of binary signatures are: 
    (a) $\begin{bsmallmatrix}
      1 & 0 & 0 \\
      0 & 0 & \alpha \\
      0 & \alpha & 0
      \end{bsmallmatrix}$ or $\begin{bsmallmatrix}
      0 & 0 & \alpha \\
      0 & 1 & 0 \\
      \alpha & 0 & 0
    \end{bsmallmatrix}$ for $\alpha = \pm 1$;
    (b) $\mathcal{D}$;
    (c) $\begin{bsmallmatrix}
      1 & x & -x \gamma \\
      x & x^2 & \gamma \\
      - x \gamma & \gamma & 0
    \end{bsmallmatrix}$ for $\gamma = \pm \sqrt{1 + x^2}$ and a nonzero $x \in \mathbb{R}$.
    Let $\mathbf{H} \in \mathcal{B}'$.
    For (a) and (c),  we see that $\mathbf{H}\teh \mathbf{F}$ is not $\db \dg | \dr$.
    For (b), if $\mathbf{H} \in \mathcal{D}$, then by \cref{lem:gr-equality-D-dichotomy} with $\mathbf{G}'$, $\mathbf{H}$ must be $\db \dg | \dr$, \swbg, or \strspt.
    Only case we need to consider is the signatures supported on $\{\dg, \dr\}$.
    We may write $\mathbf{H} = \begin{bsmallmatrix}
      0 & 0 & 0 \\
      0 & x & y \\
      0 & y & z
    \end{bsmallmatrix}$, and since it must be tractable with $\mathbf{G}'$, $[x, y, z]$ must be $[0, *, 0]$, $[*, 0, *]$ or degenerate.
    However, it cannot be degenerate since then $\mathbf{H}$ is degenerate.
    If $y = 0$, then $\mathbf{H}$ is $\db \dg | \dr$.
    If $x, z = 0$, then it is \swbg.

    We have shown so far that all signatures that are not supported on $\{\db, \dr\}$ are either $\db \dg | \dr$ or \swbg. 
    The existence of a binary signature supported on $\{\db, \dr\}$ determines the tractable class between \tractBGR and \tractBGGRBR, and the argument is similar to case (2).
\end{enumerate}

\subsubsection{Rank 2 Class \texorpdfstring{\tractBGGRBR}{E}}
\label{subsec:cannot-be-made-bg-r}
From now on, we further assume that under any orthogonal transformation $T$, there always is a non $\db \dg | \dr$ signature in $T \mathcal{O}$.
By \cref{lem:set-of-ternary-signatures-must-bg-r-or-separate}, this implies that there must exist some $T$ such that all the signatures in $\mathcal{O}' = T \mathcal{O}$ are \strspt to escape \#\P-hardness.
Let $\mathcal{O}_{ij}$ be the signatures in $\mathcal{O}'$ such that are supported on $\{i, j\}$.
By assumption, it must be the case that at least two $\mathcal{O}_{ij}$ are nonempty, and not \geneq.
Without loss of generality, we may assume that $\mathcal{O}_{\db \dg}$ and $\mathcal{O}_{\dg \dr}$ are not empty and contains non-\geneq, say
$\mathbf{F} \in \mathcal{O}_{\db \dg}$ and $\mathbf{G} \in \mathcal{O}_{\dg \dr}$.
We claim that all binary signatures in $\mathcal{B}' = T \mathcal{B}$ must be either \strspt or \genperm.
Then, using the fact that $\mathbf{F}, \mathbf{G}$ are not \geneq, we will show that only special forms of \genperm signatures can be in $\mathcal{B}'$ to escape \#\P-hardness.

We will argue in the following way: 
by definition, if $\mathbf{A} \in \mathcal{O}'$ and $\mathbf{H} \in \mathcal{B}'$, then $\mathbf{H}\teh \mathbf{A} \in \mathcal{O}'$
;
therefore, $\mathbf{H}\teh \mathbf{F}$ and $\mathbf{H}\teh \mathbf{G}$ must be \strspt or degenerate.
In particular, $\mathbf{H}$ must be a symmetric matrix that maps $P_{\mathbf{F}}$ ($xy$-plane) to a coordinate plane or to a single vector.
Since $\mathbf{H}$ is a linear transformation, it must then send $\mathbf{e}_1$ and $\mathbf{e}_2$ to some axes or collapse them to a single vector.
This allows us to analyze the first two columns of $\mathbf{H}$, and similarly we can analyze the second and third columns using $\mathbf{G}$.

\begin{enumerate}
    \item 
Suppose $\mathbf{H}$ sends the $xy$-plane to a single vector.
This means that the first two columns of $\mathbf{H}$ are linearly dependent and hence $\mathbf{H} \in \mathcal{D}$.
We may write $\mathbf{H} = \begin{bsmallmatrix}
  ax^2 & axy & bx \\
  axy & ay^2 & by \\
  bx & by & c
\end{bsmallmatrix}$.
\begin{enumerate}
  \item 
    Suppose the $yz$-plane also goes to a single vector. 
    Then the second and third columns must also be linearly dependent.
    If the second column is nonzero, then $\mathbf{H}$ is degenerate.
    If the second column is zero, then $\mathbf{H}$ is \strspt.

  \item Suppose the $yz$-plane stays on $yz$-plane.
    Then, it must be the case that $axy, bx = 0$.
    If $x = 0$, then $\supp \mathbf{H} \subseteq \{\dg, \dr\}^*$.
    If $x \ne 0$, then $b$ must be $0$, and $ay = 0$.
    If $a = 0$, then $\mathbf{H}$ is degenerate.
    If $y = 0$, then $\supp \mathbf{H} \subseteq \{\db, \dr\}^*$.

  \item Suppose the $yz$-plane goes to the $xy$-plane.
    Then, it must be the case that $by, c = 0$.
    If $y = 0$, then, $\supp \mathbf{H} \subseteq \{\db, \dr\}^*$.
    If $b = 0$, then, $\supp \mathbf{H} \subseteq \{\db, \dg\}^*$.

  \item Suppose the $yz$-plane goes to the $xz$-plane.
    Then, it must be the case that $a y^2, by = 0$.
    If $y = 0$, then $\supp \mathbf{H} \subseteq \{\db, \dr\}^*$.
    If $a, b = 0$, then $\mathbf{H}$ is degenerate.
\end{enumerate}

\item
Suppose $\mathbf{H}$ sends $xy$-plane to $xy$-plane.
Then, $\mathbf{H} = \begin{bsmallmatrix}
  a & b & 0 \\
  b & c & 0 \\
  0 & 0 & d
\end{bsmallmatrix}$.
We may assume $d \ne 0$ since otherwise $\supp \mathbf{H} \subseteq \{\db, \dg\}^*$.
\begin{enumerate}
  \item Suppose $yz$-plane goes to a single vector.
    Then, either $b, c = 0$ or $d = 0$ and $\mathbf{H}$ is \strspt.
  \item Suppose $yz$-plane stays on $yz$-plane.
    Then, $b = 0$, so $\mathbf{H}$ is a \genperm.
  \item Suppose $yz$-plane goes to $xy$-plane.
    Then $d = 0$ so $\mathbf{H}$ is \strspt.
  \item Suppose $yz$-plane goes to $xz$-plane.
    Then, $c = 0$.
    If $b = 0$, then $\mathbf{H}$ is a \genperm, so suppose $b \ne 0$.
    That means $\mathcal{F}_{\db \dr}$ is not empty since $\mathbf{H}\teh \mathbf{G} = (b, 0, 0)\teh + (0, 0, d)\teh \in \mathcal{O}_{\db \dr}$.
    So, we may also analyze where $xz$-plane goes to, which is determined by the first and third columns.
    Since $b, d \ne 0$, the two columns are linearly independent, and thus $xz$-plane does not go to a line.
    Therefore, it must go to a coordinate plane, but since $b \ne 0$, we must have $a = 0$. Then $\mathbf{H}$ is a \genperm.
\end{enumerate}

\item
Suppose $\mathbf{H}$ sends $xy$-plane to $yz$-plane.
Then, $\mathbf{H} = \begin{bsmallmatrix}
  0 & 0 & b \\
  0 & a & c \\
  b & c & d
\end{bsmallmatrix}$.
\begin{enumerate}
  \item Suppose $yz$-plane goes to a single vector.
    Then, it must be the case that $b = 0$ and $\supp \mathbf{H} \subseteq \{\dg, \dr\}^*$.
  \item Suppose $yz$-plane stays on $yz$-plane.
    Then, $b = 0$, so $\supp \mathbf{H} \subseteq \{\dg, \dr\}^*$.
  \item Suppose $yz$-plane goes to $xy$-plane.
    Then, $c, d = 0$, so $\mathbf{H}$ is a \genperm.
  \item Suppose $yz$-plane goes to $xz$-plane.
    Then, $a, c = 0$ so $\supp \mathbf{H} \subseteq \{\db, \dr\}^*$.
\end{enumerate}

\item
Suppose $\mathbf{H}$ sends $xy$-plane to $xz$-plane.
Then, $\mathbf{H} = \begin{bsmallmatrix}
  a & 0 & b \\
  0 & 0 & c \\
  b & c & d
\end{bsmallmatrix}$.
We may assume that the first two columns are linearly independent, so $\mathbf{H}\teh \mathbf{F} \in \mathcal{O}_{\db \dr}$.
\begin{enumerate}
  \item Suppose $yz$-plane goes to a single vector.
    Then, it must be the case that $b, c = 0$ and $\supp \mathbf{H} \subseteq \{\db, \dr\}^*$.
  \item Suppose $yz$-plane stays on $yz$-plane.
    Then, $b = 0$.
    Since $\mathcal{O}_{\db \dr} \ne \emptyset$, we may look at where the $xz$-plane goes to.
    If it goes to a single line, then we must also have $a, c = 0$, so $\mathbf{H}$ is degenerate.
    If $xz$-plane goes to another plane, then we must have either $d = 0$, $a = 0$, or $c = 0$.
    If $d = 0$, $\mathbf{H}$ is \genperm.
    If $a = 0$, $\supp \mathbf{H} \subseteq \{\dg, \dr\}^*$.
    If $c = 0$, $\supp \mathbf{H} \subseteq \{\db, \dr\}^*$.
  \item Suppose $yz$-plane goes to $xy$-plane.
    Then, $c, d = 0$, so $\supp \mathbf{H} \subseteq \{\db, \dr\}^*$.
  \item Suppose $yz$-plane goes to $xz$-plane.
    Then, $c = 0$, so $\supp \mathbf{H} \subseteq \{\db, \dr\}^*$.
\end{enumerate}
\end{enumerate}
Now, we will argue that if a symmetric \genperm matrices is not \strspt, then the absolute values of the nonzero coefficients must be the same for $\holts(\mathcal{F})$ to be not \sph.
After normalization, any symmetric \genperm matrix that is not \strspt has the form
\[
  \begin{bmatrix}
    0 & x & 0 \\
    x & 0 & 0 \\
    0 & 0 & 1
    \end{bmatrix}, \begin{bmatrix}
    0 & 0 & x \\
    0 & 1 & 0 \\
    x & 0 & 0 
    \end{bmatrix}, \begin{bmatrix}
    1 & 0 & 0 \\
    0 & 0 & x \\
    0 & x & 0
  \end{bmatrix} \, .
\]
for some nonzero $x$.
The $2k$-th power of each of them are 
\[
  \begin{bmatrix}
    x^{2k} & 0 & 0 \\
    0 & x^{2k} & 0 \\
    0 & 0 & 1
    \end{bmatrix}, \begin{bmatrix}
    x^{2k} & 0 & 0 \\
    0 & 1 & 0 \\
    0 & 0 & x^{2k} 
    \end{bmatrix}, \begin{bmatrix}
    1 & 0 & 0 \\
    0 & x^{2k} & 0 \\
    0 & 0 & x^{2k}
  \end{bmatrix} \, .
\]
By assumption, the signatures $\mathbf{F}$ and $\mathbf{G}$ are not of the \geneq types.
Then, by \cref{thm:dich-sym-Boolean}, $\holbs([x^{2k}, 0, 1], \mathbf{F}\domres{\db, \dg})$ and $\holbs([x^{2k}, 0, 1], \mathbf{G}\domres{\dg, \dr})$ are not tractable for big enough $k$ if $x \ne \pm 1$.
We see that all the above three matrices are ruled out for $x \ne \pm 1$, 
and hence the only non \strspt symmetric binary signatures that $\mathcal{B}'$ can possibly contain are the scalar multiples of elements of $\octgroup$.
Therefore, we have proven that $\mathcal{F}$ is in class \tractBGGRBR.